\documentclass[11pt]{article}
\usepackage{fullpage}
\usepackage{amsmath,amssymb,amsthm,mathtools}
\usepackage{microtype}
\usepackage{enumitem}
\usepackage{booktabs}
\usepackage[ruled,vlined,linesnumbered]{algorithm2e}
\usepackage{float}
\usepackage{aliascnt}
\usepackage[colorlinks=true,linkcolor=blue,citecolor=blue,urlcolor=blue]{hyperref}
\usepackage[capitalize,noabbrev]{cleveref}
\crefname{enumi}{Part}{Parts}
\newtheorem{theorem}{Theorem}[section]
\newtheorem{question}{Question}[section]
\newtheorem*{theorem*}{Theorem}
\newaliascnt{lemma}{theorem}
\newtheorem{lemma}[lemma]{Lemma}
\aliascntresetthe{lemma}
\crefname{lemma}{Lemma}{Lemmas}
\Crefname{lemma}{Lemma}{Lemmas}
\crefname{question}{Question}{Questions}
\Crefname{question}{Question}{Questions}
\newaliascnt{proposition}{theorem}
\newtheorem{proposition}[proposition]{Proposition}
\aliascntresetthe{proposition}
\crefname{proposition}{Proposition}{Propositions}
\Crefname{proposition}{Proposition}{Propositions}
\newaliascnt{corollary}{theorem}
\newtheorem{corollary}[corollary]{Corollary}
\aliascntresetthe{corollary}
\crefname{corollary}{Corollary}{Corollaries}
\Crefname{corollary}{Corollary}{Corollaries}
\newaliascnt{fact}{theorem}
\newtheorem{fact}[fact]{Fact}
\aliascntresetthe{fact}
\crefname{fact}{Fact}{Facts}
\Crefname{fact}{Fact}{Facts}
\theoremstyle{definition}
\newaliascnt{definition}{theorem}
\newtheorem{definition}[definition]{Definition}
\aliascntresetthe{definition}
\crefname{definition}{Definition}{Definitions}
\Crefname{definition}{Definition}{Definitions}
\theoremstyle{remark}
\newaliascnt{remark}{theorem}

\aliascntresetthe{remark}
\crefname{remark}{Remark}{Remarks}
\Crefname{remark}{Remark}{Remarks}
\DeclarePairedDelimiter{\ket}{\lvert}{\rangle}
\DeclarePairedDelimiter{\bra}{\langle}{\rvert}

\newcommand{\ketbra}[2]{\ket{#1}\!\bra{#2}}
\newcommand{\kett}[1]{|#1\rangle\!\rangle}
\newcommand{\bbra}[1]{\langle\!\langle#1|}
\newcommand{\Co}{\mathbb C}
\newcommand{\R}{\mathbb R}
\newcommand{\E}{\operatorname{\mathbb{E}}}
\newcommand{\qchannel}{\mathsf{QChan}}
\newcommand{\tr}{\operatorname{tr}}
\newcommand{\rank}{\operatorname{rank}}
\newcommand{\supp}{\operatorname{supp}}
\newcommand{\ran}{\operatorname{ran}}
\newcommand{\id}{\operatorname{id}}
\newcommand{\Ent}{\operatorname{Ent}}
\renewcommand{\Re}{\operatorname{Re}}
\renewcommand{\Im}{\operatorname{Im}}
\newcommand{\dnorm}[1]{\left\|#1\right\|_\diamond}
\newcommand{\opnorm}[1]{\left\|#1\right\|_\infty}
\newcommand{\hsnorm}[1]{\left\|#1\right\|_2}
\newcommand{\trnorm}[1]{\left\|#1\right\|_1}
\newcommand{\sS}{\mathsf S}
\newcommand{\calA}{\mathcal A}
\newcommand{\calB}{\mathcal B}
\newcommand{\calD}{\mathcal D}
\newcommand{\calE}{\mathcal E}
\newcommand{\calF}{\mathcal F}
\newcommand{\calG}{\mathcal G}
\newcommand{\calH}{\mathcal H}
\newcommand{\calK}{\mathcal K}
\newcommand{\calL}{\mathcal L}

\newcommand{\calN}{\mathcal N}
\newcommand{\calS}{\mathcal S}
\newcommand{\calT}{\mathcal T}
\newcommand{\calV}{\mathcal V}
\newcommand{\calY}{\mathcal Y}
\newcommand{\kettbbra}[2]{\ensuremath{\kett{#1}\!\bbra{#2}}}

\newcommand{\TwoBatch}{\hyperref[alg:two-batch]{\textsc{TwoBatch}}}

\newcommand{\braket}[2]{\ensuremath{\langle {#1} \vert {#2} \rangle}}

\allowdisplaybreaks
\setlist[enumerate]{leftmargin=2em,itemsep=3pt,topsep=4pt}
\setlist[itemize]{leftmargin=1.6em,itemsep=3pt,topsep=4pt}
\title{Near-optimal incoherent tomography of low-rank quantum channels}
\author{
Kean Chen\thanks{University of Pennsylvania, USA. Email: \texttt{keanchen.gan@gmail.com}}\and
Aadil Oufkir\thanks{University Mohammed VI Polytechnic, Morocco. Email: \texttt{aadil.oufkir@gmail.com}}
}
\date{}
\begin{document}

\maketitle
\begin{abstract}
We study tomography for quantum channels  with input dimension $d_1$, output dimension $d_2$, and Kraus rank at most $r$, to within diamond norm error $\varepsilon$, using adaptive experiments that retain no quantum memory between channel queries.
\begin{itemize}
\item For quantum channels whose non-zero Choi eigenvalues are bounded below by $\Omega(d_1/r)$, we establish optimal query upper and lower bounds $\Theta(d_1d_2r^2/\epsilon^2)$. 
The upper bound is achieved by a nonadaptive algorithm that uses the estimator from \cite{SKKG22}, together with a new diamond-norm analysis.
The lower bound applies to arbitrary adaptive incoherent protocols and follows from a new local family of channels and a uniform one-query Fisher-information bound.

\item For general channels, we establish an upper bound $O(d_1d_2r^2\log(2d_1)/\epsilon^2)$, nearly matching the above lower bound $\Omega(d_1d_2r^2/\varepsilon^2)$.
To achieve this, we generalize the above nonadaptive algorithm by adapting the input state over $O(\log(2d_1))$ rounds with the Matrix Multiplicative Weight Update algorithm.
\end{itemize}
\end{abstract}

\newpage 

\section{Introduction}

Quantum channel learning is the task  of reconstructing an unknown quantum channel from experimental data. It provides an approximate  classical description of the input-output behavior of a physical device and is therefore useful for controlling and validating quantum devices, characterizing noise,  and comparing experiments with theory \cite{PCZ97,MRL08,Eisert2020Jul}. This task has been extensively studied since the early works of \cite{CN97,PCZ97}, through compressed-sensing methods \cite{KKEG19,KRT17}, shadow-tomography approaches \cite{HCP23,KTCT23}, and projected least-squares methods \cite{SKKG22}. 

In this work, we study the \textit{query complexity} of quantum channel tomography: given black-box access to a quantum channel $\calE : \calL(\Co^{d_1}) \to \calL(\Co^{d_2})$, how many queries to $\calE$ are necessary and sufficient to produce a classical description of a quantum channel $\widehat \calE$ satisfying $\dnorm{\widehat\calE - \calE} \le \epsilon$, with probability at least $2/3$? We measure the accuracy of the estimate in diamond norm, which is an operational natural metric since it quantifies the worst-case distinguishability of two channels when arbitrary reference systems,  input states and final measurements are allowed.

The query complexity of this problem depends crucially on the access model. In the coherent setting, a learning algorithm has quantum memory, may choose entangled input state, apply several copies of the channel $\calE$ coherently and perform entangled final measurement. In the  incoherent setting, a learning algorithm has no quantum memory and has to measure the output state after each use of the channel $\calE$. Classical information can be used in the design of new input and measurement. For full-rank state tomography and  channel tomography, coherent algorithms provably outperform incoherent ones. 
For tomography of a full-rank $d$-dimensional quantum state,
the coherent query complexity is $\Theta(d^2/\epsilon^2)$  \cite{ODW16,HHJWY17} and the incoherent one is $\Theta(d^3/\epsilon^2)$  \cite{HHJWY17,GKKT20,CHLLS23}. For full-rank channel tomography, a similar tradeoff occurs for both the diamond norm and the Choi trace norm: the coherent query complexity is $\Theta(d_1^2 d_2^2/\epsilon^2)$  \cite{MB25,CGOYZ26} and the incoherent one is $\Theta(d_1^3d_2^3/\epsilon^2)$  \cite{Ouf23,BGM26}. 

Usually, physical processes are far from full rank. A channel with small Choi rank $r$ has a Choi operator supported on a low-dimensional subspace and hence can be described using $O(d_1d_2r)$ parameters. In the coherent setting,  the query complexity of channel learning is mostly known: \cite{CGOYZ26} shows that   the query  complexity is $\Theta(d_1d_2r/\epsilon^2)$ whenever the dilation rate $\tau \coloneqq rd_2/d_1\ge 1+ \Omega(1)$ (away-from-boundary regime). At the boundary regime $\tau=1$, the complexity is  $\Theta(d_1d_2r/\epsilon)$, exhibiting  Heisenberg scaling \cite{CGOYZ26,HLSZY26}. 
In the incoherent setting, the rank-dependent query complexity was not fully characterized  beyond the special cases $r=1$ \cite{CGOYZ26} and $r=d_1d_2$ \cite{Ouf23,BGM26}.
This motivates the main question of this work.

\begin{question}\label{question}
    What is the  query complexity of quantum channel tomography under the diamond norm   in the incoherent model? 
\end{question}
    
In this paper, we answer this question (up to logarithmic factors). We prove that the query complexity in the incoherent setting is $\widetilde{\Theta}(d_1d_2r^2/\epsilon^2)$. In the away-from-boundary regime $\tau \ge 1+\Omega(1)$, this represents a factor $\widetilde{O}(r)$ overhead compared to the coherent query complexity. Moreover, the incoherent lower bound has classical $1/\epsilon^2$ scaling throughout the feasible parameter regime, including the boundary regime $\tau=1$.

\subsection{Main results}\label{sec:results}

Let $\qchannel_{d_1,d_2}^{r}$ denote the set of quantum channels
$\calE:\calL(\Co^{d_1})\to\calL(\Co^{d_2})$ with Kraus rank at most $r$.
An \emph{incoherent} protocol uses the channel once in each round, with an
arbitrary input entangled with a fresh reference, and immediately measures
or discards all quantum registers.  Subsequent inputs and measurements may
depend on the classical transcript, consisting of the measurement outcomes
and any private random seed, but no quantum information is retained between
rounds.  We count channel uses and impose no computational
restriction on the measurements or classical postprocessing.

Write $D\coloneqq d_1d_2$.  A channel of Kraus rank exactly $r$ exists if and only if
$r\le D$ and $d_1\le rd_2$.  We consider the
nontrivial case $d_2\ge2$ and therefore assume
\begin{equation}\label{eq:feasible}
 d_2\ge2,\qquad 1\le r\le D=d_1d_2,\qquad d_1\le rd_2.
\end{equation}
When $d_2=1$, the trace map is the only channel and no query is needed.

Our first result concerns quantum channels that have $\Omega(d_1/r)$-gapped Choi spectrum.  
\begin{theorem}[Upper bound for the gapped case, \cref{cor:gapped} restated] 
For every $\calE\in\qchannel_{d_1,d_2}^{r}$ such that the non-zero eigenvalues of its Choi operator are at least $\Omega(d_1/r)$, the nonadaptive incoherent protocol in \cref{alg:gapped} outputs a channel $\widehat\calE$
satisfying
$\Pr\left\{\dnorm{\widehat\calE-\calE}>\epsilon\right\}\le 1/3$
using at most 
\begin{equation*}
O\!\left(\frac{d_1d_2r^2}{\epsilon^2}\right)
\end{equation*}
queries.
\end{theorem}
\cref{alg:gapped} uses the linear estimator of \cite{SKKG22} with Haar random measurements on the Choi state. We develop a new analysis of this protocol to handle the diamond distance. 

For general channels, the spectrum need not be gapped. We show that adapting the input state leads to a similar complexity up to a logarithmic factor. 
\begin{theorem}[Upper bound, \cref{thm-9151357} restated]\label{thm:sf-main}
For every $\calE\in\qchannel_{d_1,d_2}^{r}$, the adaptive
incoherent protocol in \cref{alg:adaptive} outputs a channel $\widehat\calE$
satisfying
$
 \Pr\left\{\dnorm{\widehat\calE-\calE}>\epsilon\right\}\le 1/3
$
using at most
\begin{equation*}
 O\!\left(\frac{d_1d_2r^2}{\epsilon^2}\log(2d_1)\right)
\end{equation*}
queries. 
\end{theorem}
 \cref{alg:adaptive} uses $O(\log(2d_1))$ adaptive updates of the input
state and an independent Haar-random basis measurement after each channel
use.  

We complement our upper bounds with the following lower bound.

\begin{theorem}[Lower bound, \cref{thm-9151403} restated]\label{thm:lb-main}
Every adaptive incoherent protocol that learns every  channel in $\qchannel_{d_1,d_2}^{r}$ to diamond-norm error $\epsilon$ with probability at least $2/3$
uses at least
\begin{equation*}
\Omega\!\left(\frac{d_1d_2r^2}{\epsilon^2}\right)
\end{equation*}
queries.  The hard instance family used to prove the lower bound can be chosen so that every non-zero eigenvalue of
its Choi operator belongs to $[d_1/(4r),4d_1/r]$.  
\end{theorem}

Our results are summarized in \cref{tab:rates}.

\begin{table}[ht]
\centering
\caption{Incoherent query complexity of quantum channel tomography at constant success probability.}
\label{tab:rates}
\vspace{3mm}
\begin{tabular}{lcc}
\toprule
Setting & Upper bound & Lower bound \\
\midrule
Gapped, nonadaptive
&
 $O\!\left(\frac{d_1d_2r^2}{\epsilon^2}\right)$ \cref{cor:gapped}
&
 $\Omega\!\left(\frac{d_1d_2r^2}{\epsilon^2}\right)$ \;\cref{thm-9151403}
\\[6pt]
General, adaptive
&
 $O\!\left(
\frac{d_1d_2r^2\log(2d_1)}{\epsilon^2}
\right)$ \cref{thm-9151357}
&
 $\Omega\!\left(
\frac{d_1d_2r^2}{\epsilon^2}
\right)$ \;\cref{thm-9151403}
\\
\bottomrule
\end{tabular}
\end{table}

\subsection{Overview of techniques}

\textbf{Upper bound, Gapped case.} When the smallest positive eigenvalue of the Choi operator $C_\calE$ is at least $\Omega(d_1/r)$ (Gapped case), the protocol of \cite{SKKG22} is optimal up to constants (at constant success probability). It consists of measuring the Choi state with Haar random bases, averaging the formed Haar snapshots (\cref{def-9150321}), and projecting to a valid Choi state $\widetilde X $. In order to control the diamond distance, we split the error $\Delta = \widetilde X - \frac{1}{d_1}C_\calE$ to four blocks according to the projection $P_\calE$ onto the support of the Choi operator $C_\calE$. The support block $P_\calE\Delta P_\calE$ contributes $\kappa(\calE)\|\Delta\|_\infty$ where $\kappa(\calE) = d_1 \|\tr_B P_\calE\|_\infty$. For $\Omega(d_1/r)$-gapped channels $P_\calE  \preceq O(\frac{r}{d_1})C_\calE$ so $\kappa(\calE)\le O(r)$ and the contribution of the support block is at most $O(r\|\Delta\|_\infty)$. The same contribution (up to confidence-dependent factors) can be proven for the complement block $(I-P_\calE) \Delta (I-P_\calE)$ and this is the technical heart: Since we measure with a global Haar basis, and we project on the complement, the law of $\Delta$ is invariant under every unitary supported on $(\supp C_\calE)^{\perp}$, which can be used to flatten the input marginal (\cref{lem:random-complement}). The cross blocks are handled by Cauchy Schwarz inequality. { A scalar Bernstein argument combined with a sphere covering net (\cref{sec:covering-packing}), using the dimension-free subexponential moments of every directional Haar snapshot, gives $\|\Delta\|_\infty=O(\sqrt{(D+\log(1/\delta))/n})$ in the relevant range (\cref{lem:haar-moments}). Together with the block analysis, this yields $\|\widehat{\calE} - \calE\|_\diamond = O(r\sqrt{D/n})$ at constant success probability.}

\noindent\textbf{Upper bound, general case.} From the sketch above, the difficulty for general channels comes from bounding $\kappa(\calE)$ which can be much larger than $r$. In this case, we go beyond fixed maximally entangled input case, allowing for adaptive input states. Now, for an arbitrary input state $\sigma_{\mathrm{in}}$ we prepare the joint output state $X = \calE(\kettbbra{\sqrt{\sigma_{\mathrm{in}}}}{\sqrt{\sigma_{\mathrm{in}}}})$ and the orthogonal projector on its support $P_{X
}$. Following the same previous strategy, the quantity one would like to control is
$
\kappa_{\sigma_{\mathrm{in}}}(\calE)
:=\left\|\sigma_{\mathrm{in}}^{-1/2}(\tr_B P_X)^{\mathrm T}
\sigma_{\mathrm{in}}^{-1/2}\right\|_\infty.
$
It turns out that, for every channel $\calE$,  one  can show that there is input state $\sigma_{\mathrm{in}}$ for which $\kappa_{\sigma_{\mathrm{in}}}(\calE) = O(r)$ and thus can be used to obtain near optimal upper bound on the query complexity. However, finding such input state is nontrivial. The difficulty comes partially from the projector $P_X$ which is a discontinuous function of $X$. To overcome this difficulty, we  soften the  projector $P_X$ using the following filter for $s>0$:
\[f_s(X) = X(X+sI)^{-1}  \underset{s\,\downarrow\, 0}{\longrightarrow } P_X.\]
The new quantity that takes the role of $\kappa_{\sigma_{\mathrm{in}}}(\calE)$ is then $\|H_s(\sigma_{\mathrm{in}})\|_\infty$ 
where $H_s(\sigma_{\mathrm{in}})=\Phi_{\sigma_{\mathrm{in}}}(f_s(X))^\dagger(I_{\mathrm{B}})$ and  $\Phi_{\sigma_{\mathrm{in}}}(M)$ is the linear map whose Choi operator is $\sigma_{\mathrm{in}}^{-\mathrm{T}/2} M \sigma_{\mathrm{in}}^{-\mathrm{T}/2}$. Indeed, we have from \cref{lem:kappa} that  $\lim_{s\downarrow0}\|H_s(I/d_1)\|_\infty=\kappa(\calE)$. So the goal now is to find an input state  $\sigma_{\mathrm{in}}$ such that $\|H_s(\sigma_{\mathrm{in}})\|_\infty= O(r)$. 

Given the filter $f_s(X)$ and defining $X_s = X - sf_s(X) \approx X$ we can decompose the channel $\calE = \Phi_{\sigma_{\mathrm{in}}}(X_s) + s\Phi_{\sigma_{\mathrm{in}}}(f_s(X))$. Since $\|s\Phi_{\sigma_{\mathrm{in}}}(f_s(X))\|_\diamond = s\|H_s(\sigma_{\mathrm{in}})\|_\infty$, by setting $s=O(\epsilon/r)$, we can omit the second term in the decomposition once the goal $\|H_s(\sigma_{\mathrm{in}})\|_\infty \le O(r)$ is achieved. The first term of the decomposition $\Phi_{\sigma_{\mathrm{in}}}(X_s)$ motivates the shift of focus to $X_s = X^2(X+sI)^{-1}$. Furthermore, estimating $X_s$ permits to control $\|H_s(\sigma_{\mathrm{in}})\|_\infty$ since $H_s(\sigma_{\mathrm{in}}) = \frac{1}{s}\left(I_{\rm A}-\Phi_{\sigma_{\mathrm{in}}}(X_s)^\dagger(I_B)\right)$.

The operator $X_s$ can be written as $X_s= XRX$  with $R = (X+sI)^{-1}P_X $ which suggests an estimation using Nystr\"{o}m sandwich: we prepare two \textit{independent} estimators $\widehat{X}_0$ and $\widehat{X}_1$ of $X$ using averages of Haar snapshots. We estimate $\widehat{R} = (S+sI)^{-1}P_S$ where $S$ retains the $r$ largest eigenvalues of $\widehat{X}_0$ and set the remaining to zero. We then estimate $X_s$ with $\widehat{X}_s = \widehat{X}_1 \widehat{R} \widehat{X}_1$.

The sandwich form ensures that $\widehat{X}_s\succeq 0$ and thus $\Phi_{\sigma_{\mathrm{in}}}(\widehat{X}_s)$ is completely positive. Moreover, since $\widehat{R}$ is independent of $\widehat{X}_1 $, the error can be controlled by writing it as 
\[ \widehat{X}_s - X_s = X(\widehat{R}-R)X + E\widehat{R}X + X \widehat{R} E + E\widehat{R}E   \]
where $E = \widehat{X}_1 - X$. The cross terms can be controlled using a L\"{o}wner Cauchy-Schwarz inequality. The same expression can be used to get a multiplicative estimate $\widehat{H}(\sigma_{\mathrm{in}}) = \frac{1}{s}(I_{\rm A} - \Phi_{\sigma_{\mathrm{in}}}(\widehat{X}_s)^\dagger(I_B) )$ of $H_s(\sigma_{\mathrm{in}})$. Concentration inequalities show that the near-optimal upper bound is proven whenever $\|H_s(\sigma_{\mathrm{in}})\|_\infty \le 10r$ (\cref{prop:block}). 

The last part of this upper bound is the design of input state $\sigma_{\mathrm{in}}$ achieving $\|H_s(\sigma_{\mathrm{in}})\|_\infty \le 10r$. It turns out that 
 $H_s(\sigma_{\mathrm{in}}) =  (\calE^c)^\dagger ( (\calE^c(\sigma_{\mathrm{in}}) +sI_r )^{-1})$  where $\calE^c$ is the complementary channel of $\calE$ according to a fixed Kraus list(\cref{lem:resolvent-marginal,lem:ridge-tail}) so $\tr\left(\sigma_{\mathrm{in}} H_s(\sigma_{\mathrm{in}})\right)\le r$ for every $\sigma_{\mathrm{in}}$. 
 Roughly speaking, we want to turn an average guarantee $\tr\left(\sigma_{\mathrm{in}} H_s(\sigma_{\mathrm{in}})\right)\le r$ to a worst case guarantee $\|H_s(\sigma_{\mathrm{in}})\|_\infty \le 10r$. 
 This can be achieved by using the matrix exponentiated-gradient update \cite{TRW05}:
 \[\rho_{j+1} =  \frac{\exp\left(\log\rho_j +\log A_j\right)}{\tr\left(\exp\left(\log\rho_j +\log A_j\right)\right)}, \qquad \sigma_j = \frac{1}{2}\rho_j+\frac{I_{\rm A}}{2d_1}.\]
Here, $A_j$ is obtained from $\widehat{H}(\sigma_j)+rI_{\rm A}$ by replacing every eigenvalue below $r/2$ by $r/2$. This update raises the weight of input directions on which $H_s$ is large. It can be shown that $T=\max\{1,\lceil\log_2d_1\rceil\}$ iterations suffice using two ingredients. Golden--Thompson inequality gives
$
\tr\left(\exp\left(\log\rho_j+\log A_j\right)\right)
\le \tr(\rho_jA_j).$
On the concentration event, $A_j$ is a multiplicative approximation to $A(\rho_j)=H_s(\sigma_j)+rI_{\rm A}$, while the trace budget
$
\tr\bigl(\rho_jA(\rho_j)\bigr)\le3r$
follows from $\rho_j\preceq2\sigma_j$ and $\tr(\sigma_jH_s(\sigma_j))\le r$. Consequently, $\tr(\rho_jA_j)=O(r)$. The second ingredient is the operator convexity of $\rho\mapsto\log A(\rho)$ (\cref{cor:design-convex}), which translates a bound on the average of $\log A_j$ to a bound on $\log\left(H_s(\bar{\sigma})+rI_{\rm A}\right)$, where $\bar{\sigma}=\frac{1}{T}\sum_{j=1}^T\sigma_j$. After $T$ iterations, $H_s(\bar{\sigma})\preceq9rI_{\rm A}$ (\cref{prop:gapping}), so $\bar{\sigma}$ can be used as an input state, leading to a near-optimal query bound.
\\

\noindent\textbf{Lower bound.}  We prove the lower bound by the standard strategy of constructing a hard family and using Fano's inequality to show that any learning algorithm must query the channel at least a certain number of times to learn a channel chosen randomly from the hard family (see e.g., \cite{Flammia2012Sep,HHJWY17,LN22,Fawzi2025Jan}). In the regime $d_1\le rd_2/2$, one could use the hard family constructed in \cite{CGOYZ26} and, instead of bounding the mutual information by the Holevo information as in \cite{Oufkir2026Jan}, use log-Sobolev inequalities for Haar measure \cite[Theorem 5.16]{Meckes19}. In order to obtain a general bound for all parameters such that $d_1\le rd_2$ and $r\le d_1d_2$, we construct a new family of hard channels by first constructing a base channel whose non-zero Choi eigenvalues lie between $\frac{d_1}{2r}$ and $\frac{2d_1}{r}$. Then, the admissible channels of the family are constructed by a small perturbation of the Kraus operators of the base channel and normalization (to ensure the trace-preserving property). This set has dimension $\Omega(d_1d_2r)$ (\cref{lem:dim}), and the hard family is constructed by taking Gaussian weights on a basis. Let $\calE_\Theta$ be a random hard channel in this family, and let $Z$ be the transcript of classical outcomes of a learning algorithm. The lower bound is proved through  bounding  the mutual information $I(\Theta; Z)$  between $\Theta$ and $Z$.  On the one hand, learning $\calE$ to diamond error $\epsilon$ permits the learner to find $\Theta$ in a small ball (\cref{lem:smallball}). Correctness of the learning algorithm and Fano's inequality imply that $I(\Theta; Z) \ge \Omega(d_1d_2r)$ (\cref{lem:fano}). On the other hand, each query reveals only a little information. In fact, using the tester representation of single-query experiments, one can show that the trace of the one-query Fisher information is at most $O(d_1d_2/r)$ uniformly (\cref{lem:fisher}). The crucial point is that, even for adaptive incoherent algorithms, Fisher information adds along the transcript, and the trace of the $n$-query Fisher information $I_Z(\Theta)$ is at most $O(nd_1d_2/r)$ (see \cref{eq:adaptive-fisher}). The Gaussian log-Sobolev inequality (\cref{lem:prior}) translates the bound on the Fisher information into a bound on the mutual information, $I(\Theta;Z) \le O\!\left(\frac{\epsilon^2}{d_1d_2}\right)\tr I_Z(\Theta)\le O(n\epsilon^2/r)$, leading to the lower bound $n\ge \Omega\left(\frac{d_1d_2r^2}{\epsilon^2}\right)$.

\subsection{Related work}

\textbf{Incoherent quantum channel tomography.}
For general Kraus rank $r$, \cite{KRT17,KKEG19} use compressed-sensing techniques to establish low-rank recovery in Schatten norms. Similar guarantees are established using projected least squares \cite{GKKT20,SKKG22}. Converting these bounds to diamond distance, using standard inequalities, yields suboptimal bounds in general. The ancilla-assisted projected least squares protocol of \cite{SKKG22} is the estimator we use in the gapped case.   

For full rank $r=d_1d_2$, \cite{Ouf23} showed the near optimal query complexity $\widetilde{\Theta}(d_1^3d_2^3/\epsilon^2)$ for  nonadaptive incoherent  channel tomography under diamond distance. This was generalized to the adaptive incoherent setting by \cite{BGM26} who established the complexity $\Theta(d_1^3d_2^3/\epsilon^2)$. 

When the rank is minimal $r=1$, channel tomography becomes isometry tomography. \cite{HKOT23} learn a unitary channel on $\mathbb{C}^d$ to diamond error $\epsilon$ using $O(d^2/\epsilon^2)$  incoherent measurements.  Their sharper $O(d^2/\epsilon)$ bound uses a sequential
bootstrap procedure that repeatedly composes the unknown
unitary with a current classical estimate. It therefore  falls outside the incoherent model. Then,  
\cite{CGOYZ26} generalized this result to learn an  isometry channel of dimensions $(d_1, d_2)$ to diamond error $\epsilon$ using $O(d_1d_2/\epsilon^2)$ incoherent measurements. Our results extend both results to higher Kraus rank $r$. 

\noindent\textbf{Coherent quantum channel tomography.} 
The upper bound $O(d_1d_2r/\epsilon^2)$ on the coherent query complexity was established by \cite{MB25,CGOYZ26} and the lower bound $\Omega(d_1d_2r/\epsilon^2)$ was proved by \cite{CGOYZ26} when the dilation rate $\tau = rd_2/d_1 \ge 1 + \Omega(1)$. At the boundary $\tau=1$, channel tomography exhibits Heisenberg scaling $\Theta(d_1d_2r/\epsilon)$ \cite{CGOYZ26,HLSZY26}, which is not the case in the incoherent setting where $\Omega(d_1d_2r^2/\epsilon^2)$ is needed for all $\tau\ge 1$ (\cref{thm:lb-main}).

\noindent\textbf{Quantum state tomography.} State tomography can be seen as a special case of channel tomography with trivial input, i.e., $d_1=1$. The coherent query complexity is $\Theta(d_2r/\epsilon^2)$ \cite{ODW16,HHJWY17}. 
The full rank incoherent query complexity  ${\Theta}(d_2^3/\epsilon^2)$ is due to \cite{HHJWY17,GKKT20} for the upper bound and to \cite{CHLLS23} for the lower bound. The rank-$r$ explicit dependency was upper bounded by $O(d_2r^2/\epsilon^2)$ \cite{HHJWY17,GKKT20}. It was lower bounded by $\Omega(d_2r^2/\epsilon^2)$ by \cite{LN22} for nonadaptive strategies. \cref{thm:lb-main} specialized to $d_1=1$ generalized this result to adaptive strategies  so the incoherent rank-$r$ state query complexity is $\Theta(d_2r^2/\epsilon^2)$. We note that this state lower bound was independently proved by \cite{Keskin2026Sep,Nayak2026Sep} as a special case. The general framework for lower bounds was used in \cite{Flammia2012Sep,HHJWY17,LN22,Fawzi2025Jan}. Related Fisher-information and log-Sobolev arguments also appear in
\cite{Keskin2026Sep,Nayak2026Sep}.

\subsection{Discussion}

Our results answer \cref{question} by characterizing the cost of quantum memory in low-rank quantum channel tomography. For $\Omega(d_1/r)$-gapped channels,  a nonadaptive strategy, with a fixed maximally entangled input state, achieves  the optimal incoherent query complexity 
\[
\Theta\!\left(\frac{d_1d_2r^2}{\epsilon^2}\right)
\]
for constant success probability. For general channels, adaptive input selection permits to overcome the gapped spectrum assumption at the cost of an overhead of  $\log(2d_1)$ in query complexity. Consequently, compared with the coherent channel tomography results in \cite{CGOYZ26}, we see that the advantages of quantum memory lie in the Choi rank $r$, and also the error $\varepsilon$ (in the boundary and near-boundary regimes), up to a logarithmic factor.

A limitation of our upper bounds is that both algorithms are ancilla-assisted as they prepare a purification of the input state and measure the joint output-reference system. It remains open whether ancilla-free algorithms can achieve comparable query complexities for general Kraus rank $r$ (for minimal rank $r=1$ and for full rank $r=d_1d_2$ this is true according to \cite{CGOYZ26} and \cite{Ouf23,BGM26} respectively).

Moreover, our algorithms use exact Haar-random bases. It is natural to
ask whether finite exact or approximate unitary $t$-designs can
replace Haar randomness while preserving the query complexity of channel tomography.

\subsection{Organization}

In \cref{sec:prelim}, we introduce the notation and tools used
throughout the paper. In \cref{sec:upper}, we prove the upper bounds for channel tomography. We
begin in \cref{sec:fixed-probe} with a nonadaptive protocol based on a
fixed maximally entangled input state and establish a query upper bound in terms
of the parameter $\kappa(\calE)$; this yields the optimal query complexity for
gapped channels. We then turn in \cref{sec:algorithm} to general
channels, where adaptive selection of the input state removes the
spectral-gap assumption at a logarithmic overhead. The technical
lemmas used in both upper bounds are collected in
\cref{sec:ub-technical}.

In \cref{sec:lower}, we prove the lower bound for adaptive incoherent
protocols. The hard family consists of gapped channels, so the lower
bound applies both to the gapped subclass and to the full
class of channels of Kraus rank at most $r$. The main argument appears
in \cref{sec:lb-proof}, and the geometric and information-theoretic
estimates are proved in \cref{sec:lb-technical}.

\section{Preliminaries}\label{sec:prelim}
\subsection{Quantum channels and Choi operators}\label{sec:choi} 
For a finite-dimensional Hilbert space $\calH$, write $\calL(\calH)$ for its linear operators and $\calD(\calH)$ for its density operators.  The symbols $\|\cdot\|_1$, $\|\cdot\|_2$, and $\|\cdot\|_\infty$ denote the trace, Hilbert--Schmidt, and operator norms, respectively; vector norms are Euclidean.  On a real linear space of complex matrices, orthogonality is with respect to the real Hilbert--Schmidt inner product $\langle M,N\rangle_\R=\Re\tr(M^\dagger N)$.  For Hermitian $M$, its support is $\supp M=(\ker M)^\perp$. The L\"owner order $M\preceq N$ means that $N-M$ is positive semidefinite. By a \emph{multiplicative estimate} $\widehat M$ of a positive definite matrix $M$, with relative error $0\le\eta<1$, we mean $(1-\eta)M\preceq\widehat M\preceq(1+\eta)M$. All logarithms are natural unless a base is specified, and $\widetilde O$ and $\widetilde\Theta$ suppress logarithmic factors.

Let $\mathrm{A}$ and $\mathrm{B}$ denote the input and output quantum systems, with Hilbert spaces $\mathcal{H}_{\mathrm{A}}\cong\Co^{d_1}$ and $\mathcal{H}_{\mathrm{B}}\cong\Co^{d_2}$. We retain system labels in subscripts: $I_{\mathrm{A}}$ denotes the identity operator on $\mathcal{H}_{\mathrm{A}}$, $\id_{\mathrm{A}}$ the identity map on $\calL(\mathcal{H}_{\mathrm{A}})$, and $\tr_{\mathrm{A}}$ the partial trace over $\mathcal{H}_{\mathrm{A}}$; analogous notation applies to other systems. 

We fix the computational basis, and define transposes and complex conjugates with respect to this basis. Complex conjugation is denoted by a superscript $\star$; in particular, $\ket{\psi^\star}$ is the complex conjugate of $\ket{\psi}$. 
For $K:\mathcal{H}_{\mathrm{A}}\to\mathcal{H}_{\mathrm{B}}$, we write
\begin{equation}\label{eq:vec}
 \begin{aligned}
 \kett K=\sum_{a=1}^{d_1}\bigl(K\ket a\bigr)\otimes\ket a
          \in\mathcal{H}_{\mathrm{B}}\otimes\mathcal{H}_{\mathrm{A}}, \qquad \kett I=\sum_{a=1}^{d_1}\ket a\otimes\ket a         \in\mathcal{H}_{\mathrm{A}}\otimes\mathcal{H}_{\mathrm{A}}.
 \end{aligned}
\end{equation}
The inverse reshaping map is written $\operatorname{mat}$, so $\operatorname{mat}(\kett K)=K$. Therefore, we can easily see
\begin{equation}\label{eq:vec-facts}
 \langle\!\langle K|L\rangle\!\rangle=\tr(K^\dagger L),\quad
 \kett{KM}=(I\otimes M^{\mathrm T})\kett K,\quad
 \tr_{\mathrm{B}}(\kettbbra{K}{L})=(L^\dagger K)^{\mathrm T}.
\end{equation}

For a linear map
$\Phi:\calL(\mathcal{H}_{\mathrm{A}})\to\calL(\mathcal{H}_{\mathrm{B}})$,
define its Choi operator $C_\Phi\in\calL(\mathcal{H}_{\mathrm{B}}\otimes\mathcal{H}_{\mathrm{A}})$ by
\begin{equation}\label{eq:choi-def}
 C_\Phi=(\Phi\otimes\id_{\mathrm{A}})(\kettbbra {I}{I}) =\sum_{a,a'}\Phi(\ketbra a{a'})\otimes\ketbra a{a'}.
\end{equation}
For a channel $\mathcal{E} \in \qchannel_{d_1,d_2}^{r}$, we have $\tr C_\calE=d_1$.
The Choi state refers to the normalized Choi operator $C_\mathcal{E}/d_1$.
Complete positivity of $\Phi$ is equivalent to $C_\Phi\succeq0$, and trace preservation is equivalent to $\tr_{\mathrm{B}} C_\Phi=I_{\mathrm{A}}$.
A Kraus representation and its Choi representation are
\begin{equation}\label{eq:choi-kraus}
 \calE(\rho)=\sum_{i=1}^r K_i\rho K_i^\dagger,\qquad
 \sum_i K_i^\dagger K_i=I_{\mathrm{A}},\qquad
 C_\calE=\sum_i\kettbbra{K_i}{K_i}.
\end{equation}
The minimum number of Kraus operators, called the \emph{Kraus rank} or
\emph{Choi rank}, equals $\rank(C_\calE)$.

We write $\qchannel_{d_1,d_2}^{r}$ for the set of quantum channels
$\calE:\calL(\Co^{d_1})\to\calL(\Co^{d_2})$ with Kraus rank at most $r$.

\begin{proposition}[Diamond norm and Choi operators \cite{Watrous18}]\label{lem:diamond-choi}
Let $\Phi:\calL(\mathcal{H}_{\mathrm{A}})\to\calL(\mathcal{H}_{\mathrm{B}})$ be Hermiticity-preserving, with Choi
operator $C_\Phi$ defined in \cref{eq:choi-def}.  Then
\begin{equation}\label{eq:diamond-choi}
 \dnorm\Phi= \max_{\rho\in\calD(\mathcal{H}_{\mathrm{A}})}  \left\|(I_{\mathrm{B}}\otimes\sqrt{\rho^{\mathrm T}})C_\Phi               (I_{\mathrm{B}}\otimes\sqrt{\rho^{\mathrm T}})\right\|_1.
\end{equation}
In particular, any two channels $\calE,\calF:\calL(\mathcal{H}_{\mathrm{A}})\to\calL(\mathcal{H}_{\mathrm{B}})$ satisfy
\begin{equation}\label{eq:choi-diamond-comparison}
 \frac1{d_1}\|C_\calE-C_\calF\|_1
 \le\|\calE-\calF\|_\diamond
 \le\|C_\calE-C_\calF\|_1.
\end{equation}
\end{proposition}

We also record two elementary trace-norm consequences of the L\"owner order.

\begin{proposition}[Order and trace-norm estimates]\label{lem:loewner-tools}
Let $G\succeq0$ and let $Q$ be Hermitian, both acting on the same finite-dimensional space.
\begin{enumerate}[leftmargin=2.2em,label=\rm(\roman*)]
\item If $-G\preceq Q\preceq G$ then $\trnorm Q\le\tr G$.
\item If $\opnorm Q\le\eta$ and $\Pi$ is the orthogonal projection onto
      $\supp Q$, then $\trnorm{ZQZ^{\dagger}}\le\eta\,\tr(Z\Pi Z^{\dagger})$ for every
      matrix $Z$.
\end{enumerate}
\end{proposition}
\begin{proof}
(i) Let $\Pi_\pm$ be the spectral projections of $Q$ onto non-negative and negative eigenspaces.  Then $\trnorm Q=\tr(Q\Pi_+)-\tr(Q\Pi_-)\le\tr(G\Pi_+)+\tr(G\Pi_-)=\tr G$.
(ii) Write $Q=|Q|^{1/2}S|Q|^{1/2}$ with $\opnorm S\le1$. Then $ZQZ^{\dagger}=(Z|Q|^{1/2}S)(|Q|^{1/2}Z^{\dagger})$ and, by $\trnorm{MK}\le\hsnorm M\hsnorm K$, $\trnorm{ZQZ^{\dagger}}\le\tr(Z|Q|Z^{\dagger})\le\eta\tr(Z\Pi Z^{\dagger})$, using $|Q|\preceq\eta\Pi$.
\end{proof}

\subsection{Single-query testers}
A single-query experiment prepares a state on
$\mathcal H_{\mathrm A}\otimes\mathcal H_{\mathrm R}$, applies
$\calE\otimes\id_{\mathrm R}$, and measures the resulting state. Any mixed
input can be simulated by purifying it, enlarging the reference system, and
extending the measurement trivially to the purifying register. We may
therefore take the input to be a pure state
$\ket{\psi}\in\mathcal H_{\mathrm A}\otimes\mathcal H_{\mathrm R}$.

We allow the measurement to have an arbitrary standard-Borel outcome space
$\calY$. Thus the measurement is described by a POVM
\begin{equation}
\mathsf M:\mathcal B(\calY)\longrightarrow
\calL(\mathcal H_{\mathrm B}\otimes\mathcal H_{\mathrm R}),
\qquad
\mathsf M(\calY)=I_{\mathrm B}\otimes I_{\mathrm R}.
\end{equation}
For every Borel set $B\subseteq\calY$, the outcome law is
\begin{equation}\label{eq:povm-outcome-law}
\Pr_{\calE}\{Y\in B\}
=
\tr\!\left[
\mathsf M(B)
(\calE\otimes\id_{\mathrm R})(\ketbra{\psi}{\psi})
\right].
\end{equation}

Using the copied input space in \cref{eq:vec,eq:choi-def}, write
$\ket{\psi}=(I_{\mathrm A}\otimes L)\kett I$, where
$L:\mathcal H_{\mathrm A}\to\mathcal H_{\mathrm R}$. Then
\begin{equation}
L^\dagger L
=
\tr_{\mathrm R}(\ketbra{\psi}{\psi})^{\mathrm T}.
\end{equation}
Define the operator-valued tester measure
\begin{equation}\label{eq:tester-construction}
\mathsf T(B)
:=
(I_{\mathrm B}\otimes L^\dagger)
\mathsf M(B)
(I_{\mathrm B}\otimes L).
\end{equation}
Since the outcome state is
$(I_{\mathrm B}\otimes L)C_{\calE}(I_{\mathrm B}\otimes L^\dagger)$,
cyclicity of the trace gives
\begin{equation}\label{eq:tester}
\Pr_{\calE}\{Y\in B\}
=
\tr\bigl(\mathsf T(B)C_{\calE}\bigr),
\qquad
\mathsf T(\calY)
=
I_{\mathrm B}\otimes
\tr_{\mathrm R}(\ketbra{\psi}{\psi})^{\mathrm T}.
\end{equation}
The tester depends only on the chosen input and measurement, not on
$\calE$.

If $\calY$ is countable, then
$T_y:=\mathsf T(\{y\})$ and
\begin{equation}
\Pr_{\calE}\{Y=y\}=\tr(T_yC_{\calE}),
\qquad
\sum_{y\in\calY}T_y=\mathsf T(\calY).
\end{equation}

\paragraph{Incoherent protocol.}
An incoherent protocol repeats single-query experiments without retaining
quantum memory between rounds.

Specifically, an $n$-query incoherent protocol chooses, in round $t$, a state on $\mathcal{H}_{\mathrm{A}}\otimes \mathcal{H}_{\mathrm{R}_t}$ and a POVM on $\mathcal{H}_{\mathrm{B}}\otimes \mathcal{H}_{\mathrm{R}_t}$, as functions only of its previous classical outcomes, where $\mathrm{R}_t$ is the ancilla system used in round $t$.  It applies $\calE_{A\to B}\otimes\id_{\mathrm{R}_t}$ once and performs the POVM.  All quantum registers are then discarded.  After $n$ rounds it outputs a channel determined by the transcript.  The protocol is called nonadaptive when the inputs and POVMs are independent of earlier outcomes, otherwise it is called adaptive.

\subsection{Covering and packing nets}\label{sec:covering-packing}
We collect the elementary results of covering and packing nets used in this paper; see also
\cite[Secs.~4.2 and 4.4]{Vershynin18}.  All distances in this subsection are
Euclidean.  In particular, the unit sphere of $\Co^d$ is viewed as a subset
of $\R^{2d}$ when estimating cardinalities.

\begin{definition}[Covering and packing nets]\label{def:covering-packing}
Let $S$ be a subset of a finite-dimensional Euclidean space and let
$\eta>0$.  A finite subset $\calN\subseteq S$ is an
\emph{$\eta$-covering net} of $S$ if every $x\in S$ has some
$x_0\in\calN$ with $\|x-x_0\|_2\le\eta$.  A finite subset
$\mathcal{P}\subseteq S$ is an \emph{$\eta$-packing net} if
$\|x-y\|_2>\eta$ for any distinct $x,y\in\mathcal{P}$.
A packing net is \emph{maximal} if no further point of $S$ can be added
while preserving this separation; maximality is with respect to inclusion,
not a requirement of maximum cardinality.
\end{definition}

\begin{lemma}[Covering nets from packing nets]\label{lem:covering-cardinality}
Let $S$ be a nonempty subset of the Euclidean unit ball in $\R^m$.
For every $\eta>0$, there exists a maximal $\eta$-packing net of $S$.
Every such packing net is also an $\eta$-covering net and has cardinality
at most $(1+2/\eta)^m$.  Consequently, for $d\ge1$, the unit sphere of $\Co^d$ has an
$\eta$-covering net $\calN$ satisfying
\begin{equation}\label{eq:covering-cardinality}
 |\calN|\le\left(1+\frac2\eta\right)^{2d}.
\end{equation}
The same bounds hold for
the unit sphere of any $d$-dimensional complex subspace.
\end{lemma}
\begin{proof}
For any finite $\eta$-packing net $\mathcal{P}\subseteq S$, the Euclidean balls
of radius $\eta/2$ centered at its points are pairwise disjoint and lie
inside the ball of radius $1+\eta/2$ centered at the origin.
Writing $v_m$ for the volume of the unit ball in $\R^m$, comparison of
volumes gives
\[
 |\mathcal{P}|\,v_m(\eta/2)^m\le v_m(1+\eta/2)^m,
 \qquad |\mathcal{P}|\le(1+2/\eta)^m.
\]
Starting with one point of $S$ and repeatedly adding a point at distance
strictly greater than $\eta$ from all previously chosen points must
therefore terminate.  The resulting packing net is maximal.  Any maximal
$\eta$-packing net covers $S$ at radius $\eta$: otherwise, an uncovered
point could be added, contradicting maximality.  Apply the real-dimensional
bound with $m=2d$ to obtain \cref{eq:covering-cardinality}.
\end{proof}

\begin{lemma}[Norm approximation by covering nets]\label{lem:covering-norms}
Let $\calN_q$ be an $\eta$-covering net of the unit sphere of $\Co^q$,
where $0<\eta<1$.
\begin{enumerate}[leftmargin=2.2em,label=\rm(\roman*)]
\item For every $A\in\Co^{p\times q}$,
\begin{equation}\label{eq:covering-linear}
 \opnorm A\le\frac1{1-\eta}\max_{\ket{y}\in\calN_q}\|A\ket{y}\|_2.
\end{equation}
In particular, for every $\ket{z}\in\Co^q$,
\begin{equation}\label{eq:covering-vector}
 \|\ket{z}\|_2\le\frac1{1-\eta}
 \max_{\ket{y}\in\calN_q}|\braket{y}{z}|.
\end{equation}
\item If $\eta<1/2$ and $\calN_p$ is an $\eta$-covering net of the unit
sphere of $\Co^p$, then every $A\in\Co^{p\times q}$ satisfies
\begin{equation}\label{eq:bilinear-net}
 \opnorm A\le\frac1{1-2\eta}
 \max_{\substack{\ket{x}\in\calN_p\\\ket{y}\in\calN_q}}|\bra{x} A\ket{y}|.
\end{equation}
\item If $\eta<1/2$ and $M\in\Co^{q\times q}$ is Hermitian, then
\begin{equation}\label{eq:covering-hermitian}
 \opnorm M\le\frac1{1-2\eta}
 \max_{\ket{y}\in\calN_q}|\bra{y}M\ket{y}|.
\end{equation}
\item For $A\in\Co^{p\times q}$, define
$s_{\min}(A)=\min_{\|\ket{y}\|_2=1}\|A\ket{y}\|_2$.  Then
\begin{equation}\label{eq:covering-smin}
 s_{\min}(A)\ge\min_{\ket{y}\in\calN_q}\|A\ket{y}\|_2-\eta\opnorm A.
\end{equation}
\end{enumerate}
\end{lemma}
\begin{proof}
For each unit vector $\ket{y}$, choose $\ket{y_0}\in\calN_q$ with
$\|\ket{y}-\ket{y_0}\|_2\le\eta$.  The triangle inequality gives
\[
 \bigl|\|A\ket{y}\|_2-\|A\ket{y_0}\|_2\bigr|
 \le\|A(\ket{y}-\ket{y_0})\|_2\le\eta\opnorm A.
\]
Taking the supremum over $\ket{y}$ and rearranging proves
\cref{eq:covering-linear}; applying it to the linear functional
$A=\bra{z}$ proves \cref{eq:covering-vector}.  Taking the infimum
instead proves \cref{eq:covering-smin}.

For (ii), given unit vectors $\ket{x},\ket{y}$, choose $\ket{x_0}\in\calN_p$ and
$\ket{y_0}\in\calN_q$ within distance $\eta$.  Since $\ket{x_0}$ is also a unit vector,
\begin{align*}
 |\bra{x} A\ket{y}-\bra{x_0} A\ket{y_0}|
 &\le|(\bra{x}-\bra{x_0}) A\ket{y}|+|\bra{x_0} A(\ket{y}-\ket{y_0})|\\
 &\le2\eta\opnorm A.
\end{align*}
Use $\opnorm A=\sup_{\|\ket{x}\|_2=\|\ket{y}\|_2=1}|\bra{x} A\ket{y}|$ and rearrange.
For (iii), the same estimate with $\ket{x}=\ket{y}$ and $\ket{x_0}=\ket{y_0}$ gives
\[
 |\bra{y}M\ket{y}-\bra{y_0}M\ket{y_0}|
 \le2\eta\opnorm M.
\]
Now use the Hermitian identity
$\opnorm M=\sup_{\|\ket{y}\|_2=1}|\bra{y}M\ket{y}|$ and rearrange.
\end{proof}

\begin{corollary}[Covering-net union bounds]\label{cor:covering-union}
For a random Hermitian $d\times d$ matrix $M$ and a random complex
$p\times q$ matrix $A$, the following hold for every $t>0$:
\begin{align}
 \Pr\{\opnorm M>2t\}
 &\le9^{2d}\sup_{\|\ket{x}\|_2=1}
 \Pr\{|\bra{x}M\ket{x}|>t\},
 \label{eq:covering-union-hermitian}\\
 \Pr\{\opnorm A>2t\}
 &\le9^{2(p+q)}\sup_{\|\ket{x}\|_2=\|\ket{y}\|_2=1}
 \Pr\{|\bra{x} A\ket{y}|>t\}.
 \label{eq:covering-union-bilinear}
\end{align}
In the second line, $\ket{x}\in\Co^p$ and $\ket{y}\in\Co^q$.
\end{corollary}
\begin{proof}
Choose deterministic $1/4$-covering nets as in
\cref{lem:covering-cardinality}.  For $M$, apply
\cref{eq:covering-hermitian} and a union bound over one covering net.
For $A$, apply \cref{eq:bilinear-net} and a union bound over the Cartesian
product of two covering nets.  The respective cardinality bounds are
$9^{2d}$ and $9^{2(p+q)}$.  
\end{proof}

\subsection{Haar-random basis measurement}\label{sec:haar-snapshots}
\begin{definition}[Haar-random basis]\label{def-9150318}
Let $U\in\mathbb{U}_d$ be a Haar-random unitary. We call $\{U\ket{1},\ldots,U\ket{d}\}$ a Haar-random basis on the $d$-dimensional Hilbert space $\textup{span}\{\ket{1},\ldots,\ket{d}\}$.
\end{definition}
\begin{definition}[Haar snapshot]\label{def-9150321}
Let $\rho$ be a quantum state on $\mathbb{C}^d$. Suppose we apply a Haar-random basis measurement on $\rho$, obtaining an outcome vector $\ket{\psi}$. We call
\begin{equation}\label{eq:ls}
Y=(d+1)\ketbra{\psi}{\psi}-I_d,
\end{equation}
a \emph{Haar snapshot} of $\rho$. {For $n$ independent repetitions, obtaining outcome vectors $\ket{\psi_1}, \ket{\psi_2}, \dots, \ket{\psi_n}$, we call the average of the resulting snapshots an $n$-snapshot estimator:
\begin{equation}\label{eq:ls-n}
\widehat{X}_n=\frac{1}{n}\sum_{t=1}^n \left((d+1)\ketbra{\psi_t}{\psi_t}-I_d\right).
\end{equation}
}
\end{definition}
The snapshot $Y$ has trace one and satisfies $\E Y=\rho$, but need not be positive semidefinite.
To establish unbiasedness and concentration, we use the following concentration inequality.

{
\begin{fact}[Scalar Bernstein for subexponential random variables
{\cite[Proposition~2.7.1 and Theorem~2.8.1]{Vershynin18}}]\label{fact:bernstein}
For a real-valued random variable $Z$, define
\[
 \|Z\|_{\psi_1}:=\inf\bigl\{s>0:\ \E\exp(|Z|/s)\le2\bigr\}.
\]
Up to absolute constants,
$\|Z\|_{\psi_1}=\sup_{p\ge1}p^{-1} \left(\E|Z|^p\right)^{1/p}$.
Let $Z_1,\ldots,Z_n$ be independent centered real-valued random variables,
put $K_t=\|Z_t\|_{\psi_1}$, $V=\sum_tK_t^2$, and
$K=\max_tK_t$.  There is an absolute constant $c>0$ such that, for every
$s\ge0$,
\begin{equation*}
 \Pr\left\{\left|\sum_{t=1}^n Z_t\right|>s\right\}
 \le2\exp\left(-c\min\left\{\frac{s^2}{V},\frac{s}{K}\right\}\right).
\end{equation*}
In particular, if $K_t\le K_0$ for every $t$, then there is an absolute
constant $c_1>0$ such that, for every $u\ge1$,
\[
 \Pr\left\{\left|\frac1n\sum_{t=1}^n Z_t\right|>
 c_1 K_0\left(\sqrt{\frac un}+\frac un\right)\right\}\le2e^{-u}.
\]
The same conclusions hold for complex-valued variables, up to a change of
the absolute constants, by applying the real-valued statement to their real
and imaginary parts.
\end{fact}
}

\begin{lemma}[Haar outcomes and snapshot moments]\label{lem:haar-moments}
Let $\rho$ be a quantum state on $\mathbb{C}^d$ and let $\ket{\psi_1},\ldots,\ket{\psi_n}$ be the
independent outcome vectors obtained by measuring $\rho$ in Haar-random
orthonormal bases.  Let $\widehat{X}_n =\frac1n\sum_{j=1}^nY_j$, with
$Y_j=(d+1)\ketbra{\psi_j}{\psi_j}-I_d$ as in \cref{eq:ls}, and
write $Y=(d+1)\ketbra\psi\psi-I_d$ for a single Haar snapshot.  Then:
\begin{enumerate}[leftmargin=2.2em,label=\rm(\roman*)]
\item the observed vector $\ket{\psi}$ has density
      $p_\rho(\ket{\psi})=d\bra{\psi}\rho\ket{\psi}\le d$ with respect to the uniform
      (Haar) probability measure on the unit sphere of $\Co^{d}$;
\item $\E Y=\rho$, $\E Y^{2}=(d-1)\rho+dI$, $\E(Y-\rho)^{2}\preceq2dI$ and
      $\opnorm{Y-\rho}\le d+1$;
\item \label{item-9150411}
      { there is an absolute constant $c_{\rm H}$ such that, for
      every $0<\zeta<1$,
      \begin{equation}\label{eq:haar-net-concentration}
       \Pr\left\{\opnorm{\widehat{X}_n-\rho}>
       c_{\rm H}\left(\sqrt{\frac{d+\log(2/\zeta)}{n}}
       +\frac{d+\log(2/\zeta)}{n}\right)\right\}\le\zeta;
      \end{equation}}
\item \label{item-9150413} if $U$ is a unitary with $U\rho \, U^{\dagger}=\rho$, then
      $U\widehat{X}_n \, U^{\dagger}$ has the same law as $\widehat{X}_n $.
\end{enumerate}
\end{lemma}

\begin{proof}
(i) Each vector of a Haar-random orthonormal basis is marginally Haar
distributed on the sphere, and outcome $i$ occurs with Born probability
$\bra{\psi_i}\rho\ket{\psi_i}$.  Hence for a test function $f$,
$\E f(\ket{\psi})=\sum_{i=1}^{d}\E_{\rm Haar}
[\bra{\psi_i}\rho\ket{\psi_i} f(\ket{\psi_i})]
 =d\,\E_{\ket{\psi}\sim\rm Haar}[\bra{\psi}\rho\ket{\psi} f(\ket{\psi})]= \E_{\ket{\psi}\sim\rm Haar}[p_\rho(\ket{\psi}) f(\ket{\psi})]$.

For (ii) and (iii) we use the standard moment formula
\begin{equation}\label{eq:sym-moment}
 \E_{\ket{\varphi}\sim\rm Haar}\bigl[(\ketbra{\varphi}{\varphi})^{\otimes k}\bigr]
 =\binom{d+k-1}{k}^{-1}\Pi_{\rm sym}^{(k)},
 \qquad
 \Pi_{\rm sym}^{(k)}=\frac1{k!}\sum_{\pi\in S_k}W_\pi ,
\end{equation}
where $W_\pi = \sum_{i_1, \dots, i_k} \ket{i_1i_2\cdots i_k}\bra{i_{\pi^{-1}(1)} i_{\pi^{-1}(2)} \cdots i_{\pi^{-1}(k)}} $ permutes tensor factors according to the permutation $\pi\in S_k$: the left side is a positive operator
commuting with every $U^{\otimes k}$ and supported on the symmetric subspace,
hence proportional to $\Pi^{(k)}_{\rm sym}$ by Schur's lemma, and the constant
is fixed by the trace \cite[Sec.~7.1]{Watrous18}.

(ii) By (i) and \cref{eq:sym-moment} with $k=2$,
\[
 \E\ketbra{\psi}{\psi}=d\,\tr_1\bigl[(\rho\otimes I)\E_{\rm Haar}
 (\ketbra{\varphi}{\varphi})^{\otimes2}\bigr]
 =\frac{2d}{d(d+1)}\tr_1\Bigl[(\rho\otimes I)\frac{I+W_{(12)}}2\Bigr]
 =\frac{I+\rho}{d+1},
\]
hence $\E Y=\rho$.  Since
$Y^{2}=(d+1)(d-1)\ketbra{\psi}{\psi}+I$ we get $\E Y^{2}=(d-1)\rho+dI$, and
$\E(Y-\rho)^{2}=\E Y^{2}-\rho^{2}\preceq(d-1)I+dI\preceq2dI$.  The eigenvalues of
$Y$ are $d$ and $-1$, so $\opnorm{Y-\rho}\le d+1$.

(iii)
{We next prove the dimension-log-free estimate in
\cref{eq:haar-net-concentration}.  Fix a unit vector $\ket{x}$, put
$P_x=\ketbra{x}{x}$, and set
$Z=(d+1)|\braket{x}{\psi}|^2$.  For every integer $k\ge1$, (i) and
\cref{eq:sym-moment} give
\[
 \E Z^k
 =\frac{d(d+1)^k}{\binom{d+k}{k+1}}
   \tr\!\left[(\rho\otimes P_x^{\otimes k})\Pi_{\rm sym}^{(k+1)}\right]
 \le (k+1)!.
\]
Indeed, after expanding the symmetrizer, every permutation contributes either
$\tr\rho=1$ or $\tr(\rho P_x)=\bra{x}\rho\ket{x}\le1$, and
$d(d+1)^k/\binom{d+k}{k+1}\le(k+1)!$.  Since
$\bra{x}(Y-\rho)\ket{x}=Z-\E Z$, the elementary inequality
$|a-b|^k\le2^{k-1}(|a|^k+|b|^k)$ shows that these centered random variables
have subexponential norm $\|Z-\E Z\|_{\psi_1}$ bounded by an absolute constant, uniformly in
$d$, $\rho$, and $\ket{x}$.  Scalar Bernstein (\cref{fact:bernstein}) therefore yields absolute constants
$c,c_1>0$ such that
\[
 \Pr\left\{\left|\bra{x}\bigl(
 \widehat{X}_n-\rho\bigr)\ket{x}\right|>t\right\}
 \le2\exp\bigl(-cn\min\{t^2,t\}\bigr).
\]
Apply the $1/4$-covering-net union bound in
\cref{eq:covering-union-hermitian} to $M=\widehat{X}_n-\rho$ to obtain
\[
 \Pr\{\opnorm{\widehat{X}_n-\rho}>2t\}
 \le2\exp\bigl(2d\log9-cn\min\{t^2,t\}\bigr).
\]
Taking $t=c_1(\sqrt{(d+\log(2/\zeta))/n}+(d+\log(2/\zeta))/n)$ proves
\cref{eq:haar-net-concentration} after enlarging the absolute constant.}

(iv) By (i) the density of $U\ket{\psi}$ is $d\bra{\psi}U^{\dagger}\rho U\ket{\psi}
=d\bra{\psi}\rho \ket{\psi}$, so $U\ket{\psi}$ and $\ket{\psi}$ have the same law; the
snapshots are independent and identically distributed.
\end{proof}

\paragraph{Haar experiments.}
For the upper bounds in \cref{sec:upper}, we use single-query experiments to obtain Haar
snapshots of the joint output state. Let $\calE\in\qchannel_{d_1,d_2}^r$
have Kraus operators $K_1,\ldots,K_r$, and let
$\sigma_{\mathrm{in}}\in\calD(\mathcal{H}_{\mathrm{A}})$. Prepare its purification
$\kett{\sqrt{\sigma_{\mathrm{in}}}}\in\mathcal{H}_{\mathrm{A}}\otimes
\mathcal{H}_{\mathrm{A}}$ and apply $\calE$ to the first system (in the rest of the paper, when a quantum channel acts on a bipartite system, we assume it acts on the first subsystem, with the identity on the second subsystem).
The resulting joint state is
\begin{equation}\label{eq:weighted-choi}
 \begin{aligned}
 \calE(\kettbbra{\sqrt{\sigma_{\mathrm{in}}}}{\sqrt{\sigma_{\mathrm{in}}}})
 &\coloneqq  (\calE\otimes \id)(\kettbbra{\sqrt{\sigma_{\mathrm{in}}}}{\sqrt{\sigma_{\mathrm{in}}}})\\
 &=(I_{\mathrm{B}}\otimes\sqrt{\sigma_{\mathrm{in}}^{\mathrm T}})C_\calE        (I_{\mathrm{B}}\otimes\sqrt{\sigma_{\mathrm{in}}^{\mathrm T}})\\
 &=\sum_i\kett{K_i\sqrt{\sigma_{\mathrm{in}}}}\bbra{K_i\sqrt{\sigma_{\mathrm{in}}}}.
 \end{aligned}
\end{equation}
It satisfies $\rank \calE(\kettbbra{\sqrt{\sigma_{\mathrm{in}}}}{\sqrt{\sigma_{\mathrm{in}}}})\le r$, $\tr_{\mathrm{B}}\calE(\kettbbra{\sqrt{\sigma_{\mathrm{in}}}}{\sqrt{\sigma_{\mathrm{in}}}})=\sigma_{\mathrm{in}}^{\mathrm T}$, and $\tr \calE(\kettbbra{\sqrt{\sigma_{\mathrm{in}}}}{\sqrt{\sigma_{\mathrm{in}}}})=1$.

Measure this state in an independently chosen Haar-random orthonormal
basis of $\mathcal{H}_{\mathrm{B}}\otimes\mathcal{H}_{\mathrm{A}}$.
The observed basis vector $\ket\psi$ yields the Haar snapshot
$Y=(D+1)\ketbra{\psi}{\psi}-I_D$ from \cref{eq:ls}, with $D=d_1d_2$.
Conditional on the chosen basis, the measurement has $D$ discrete outcomes.

\subsection{Weighted Choi operators and regularization}\label{sec:dictionary}
Fix a channel $\calE$ of Kraus rank $r$ and a positive definite
$\sigma_{\mathrm{in}}\in\calD(\mathcal{H}_{\mathrm{A}})$.  The following definitions relate
estimates of the joint output state $\calE(\kettbbra{\sqrt{\sigma_{\mathrm{in}}}}{\sqrt{\sigma_{\mathrm{in}}}})$ in \cref{eq:weighted-choi}
to channel reconstruction and diamond-norm error.

\paragraph{Reconstruction map.}
For each $M\in\calL(\mathcal{H}_{\mathrm{B}}\otimes\mathcal{H}_{\mathrm{A}})$,
we use $\Phi_{\sigma_{\mathrm{in}}}(M)$ to denote a linear map from $\calL(\mathcal{H}_{\mathrm{A}})$ to $\calL(\mathcal{H}_{\mathrm{B}})$ by
\begin{equation*}
 (\Phi_{\sigma_{\mathrm{in}}}(M)\otimes\id_{\mathrm{A}})
 (\kettbbra{\sqrt{\sigma_{\mathrm{in}}}}{\sqrt{\sigma_{\mathrm{in}}}})=M.
\end{equation*}
Since $\sigma_{\mathrm{in}}\succ0$, this determines $\Phi_{\sigma_{\mathrm{in}}}(M)$ uniquely.
Equivalently, its Choi operator is
\begin{equation}\label{eq:Lomega}
 C_{\Phi_{\sigma_{\mathrm{in}}}(M)}
 =(I_{\mathrm{B}}\otimes(\sigma_{\mathrm{in}}^{\mathrm T})^{-1/2})M
       (I_{\mathrm{B}}\otimes(\sigma_{\mathrm{in}}^{\mathrm T})^{-1/2}).
\end{equation}
We can see that $\Phi_{\sigma_{\mathrm{in}}}(\cdot)$ is also a linear map and $\Phi_{\sigma_{\mathrm{in}}}(M)$ is not necessarily a quantum channel.
In fact, $\Phi_{\sigma_{\mathrm{in}}}(M)$ is a channel exactly when
$M\succeq0$ and $\tr_{\mathrm{B}}M=\sigma_{\mathrm{in}}^{\mathrm T}$.  In particular,
\begin{equation*}
 \Phi_{\sigma_{\mathrm{in}}}(\calE(\kettbbra{\sqrt{\sigma_{\mathrm{in}}}}{\sqrt{\sigma_{\mathrm{in}}}}))=\calE.
\end{equation*}
So we call $\Phi_{\sigma_{\mathrm{in}}}(\cdot)$ the reconstruction map and $\Phi_{\sigma_{\mathrm{in}}}(M)$ the reconstructed map from $M$ using input state $\sigma_{\mathrm{in}}$.

Taking the partial trace and transpose in \cref{eq:Lomega} gives
\begin{equation}\label{eq:momega}
 \bigl(\tr_{\mathrm{B}}C_{\Phi_{\sigma_{\mathrm{in}}}(M)}\bigr)^{\mathrm T}
 =\sigma_{\mathrm{in}}^{-1/2}(\tr_{\mathrm{B}}M)^{\mathrm T}\sigma_{\mathrm{in}}^{-1/2}.
\end{equation}
For fixed $\sigma_{\mathrm{in}}$, the map
$M\mapsto(\tr_{\mathrm{B}}C_{\Phi_{\sigma_{\mathrm{in}}}(M)})^{\mathrm T}$ is linear and
positive.  For Hermitian $M$,
\begin{equation}\label{eq:normalization-adjoint}
 \Phi_{\sigma_{\mathrm{in}}}(M)^\dagger(I_{\mathrm{B}})
 =\bigl(\tr_{\mathrm{B}}C_{\Phi_{\sigma_{\mathrm{in}}}(M)}\bigr)^{\mathrm T},
\end{equation}
and note that $\tr\Phi_{\sigma_{\mathrm{in}}}(M)(\rho)=\tr\bigl(\rho\,\Phi_{\sigma_{\mathrm{in}}}(M)^\dagger(I_{\mathrm{B}})\bigr)$ for every state $\rho$.

\begin{lemma}[Weighted diamond-norm estimates]\label{lem:weighted-norms}
For $\rho\in\calD(\mathcal{H}_{\mathrm{A}})$, define
\begin{gather*}
 L_{\rho,\sigma_{\mathrm{in}}}
 :=I_{\mathrm{B}}\otimes
       \bigl(\sqrt{\rho^{\mathrm T}}\,(\sigma_{\mathrm{in}}^{\mathrm T})^{-1/2}\bigr),\\
 L_{\rho,\sigma_{\mathrm{in}}}\calE(\kettbbra{\sqrt{\sigma_{\mathrm{in}}}}{\sqrt{\sigma_{\mathrm{in}}}})L_{\rho,\sigma_{\mathrm{in}}}^{\dagger}=\calE(\kettbbra{\sqrt{\rho}}{\sqrt{\rho}}).
\end{gather*}
For every $M\in\calL(\mathcal{H}_{\mathrm{B}}\otimes\mathcal{H}_{\mathrm{A}})$,
\begin{equation}\label{eq:trace-transfer}
 \tr\bigl(L_{\rho,\sigma_{\mathrm{in}}}ML_{\rho,\sigma_{\mathrm{in}}}^{\dagger}\bigr)
 =\tr\bigl(\rho\,(\tr_{\mathrm{B}}C_{\Phi_{\sigma_{\mathrm{in}}}(M)})^{\mathrm T}\bigr),
\end{equation}
and for Hermitian $M$,
\begin{equation}\label{eq:diamond-weighted}
 \dnorm{\Phi_{\sigma_{\mathrm{in}}}(M)}
 =\max_{\rho\in\calD(\mathcal{H}_{\mathrm{A}})}
       \trnorm{L_{\rho,\sigma_{\mathrm{in}}}ML_{\rho,\sigma_{\mathrm{in}}}^{\dagger}}.
\end{equation}
Consequently:
\begin{enumerate}[leftmargin=2.4em,label=\rm(\alph*)]
\item\label{it:pos} if $G\succeq0$ then
      $\dnorm{\Phi_{\sigma_{\mathrm{in}}}(G)}
       =\opnorm{\Phi_{\sigma_{\mathrm{in}}}(G)^{\dagger}(I_{\mathrm{B}})}
       =\opnorm{\tr_{\mathrm{B}}C_{\Phi_{\sigma_{\mathrm{in}}}(G)}}$;
\item\label{it:sandwich} if $-aG\preceq M\preceq aG$ with $G\succeq0$
      and $a\ge0$, then
      $\dnorm{\Phi_{\sigma_{\mathrm{in}}}(M)}
       \le a\opnorm{\tr_{\mathrm{B}}C_{\Phi_{\sigma_{\mathrm{in}}}(G)}}$;
\item\label{it:cross} for $U,V\in\Co^{D\times k}$, where $D=d_1d_2$,
      \[
       \dnorm{\Phi_{\sigma_{\mathrm{in}}}(UV^{\dagger}+VU^{\dagger})}
       \le2\sqrt{\opnorm{\tr_{\mathrm{B}}C_{\Phi_{\sigma_{\mathrm{in}}}(UU^{\dagger})}}
                   \,\opnorm{\tr_{\mathrm{B}}C_{\Phi_{\sigma_{\mathrm{in}}}(VV^{\dagger})}}}.
      \]
\end{enumerate}
\end{lemma}
\begin{proof}
Cyclicity of the trace gives \cref{eq:trace-transfer}; substituting
\cref{eq:Lomega} into \cref{eq:diamond-choi} gives
\cref{eq:diamond-weighted}.  Part~\ref{it:pos} follows by positivity,
and part~\ref{it:sandwich} by \cref{lem:loewner-tools} (i).
For part~\ref{it:cross}, write $L=L_{\rho,\sigma_{\mathrm{in}}}$ and use
\begin{align*}
 \trnorm{L(UV^\dagger+VU^\dagger)L^\dagger}
 &\le2\hsnorm{LU}\hsnorm{LV}\\
 &=2\sqrt{\tr\bigl(\rho\,\Phi_{\sigma_{\mathrm{in}}}(UU^\dagger)^\dagger(I_{\mathrm{B}})\bigr)
             \,\tr\bigl(\rho\,\Phi_{\sigma_{\mathrm{in}}}(VV^\dagger)^\dagger(I_{\mathrm{B}})\bigr)}.
\end{align*}
Maximizing over $\rho$ proves the bound.
\end{proof}

\paragraph{Regularization.}
For $s>0$, let $f_s(t):=t/(t+s)$ for $t\ge0$.  For $X\succeq0$, define
\begin{equation}\label{eq:ridge}
 f_s(X)=X(X+sI)^{-1},\qquad
 X_s:=Xf_s(X)=X^2(X+sI)^{-1}=X-sf_s(X).
\end{equation}
Thus $0\preceq X_s\preceq X$, with eigenvalue map
$\lambda\mapsto\lambda^2/(\lambda+s)$.
For the joint output state
$X:=\calE(\kettbbra{\sqrt{\sigma_{\mathrm{in}}}}{\sqrt{\sigma_{\mathrm{in}}}})$, define the
\emph{regularized input-sensitivity operator}
\begin{equation}\label{eq:gram}
 H_s(\sigma_{\mathrm{in}}):=\Phi_{\sigma_{\mathrm{in}}}(f_s(X))^\dagger(I_{\mathrm{B}}).
\end{equation}
Let $\calE^c$ be the complementary channel of $\mathcal{E}$ associated with a choice of Kraus
operators, as given in \cref{eq:choi-kraus}:
\begin{equation*}
 \calE^c(\rho)=\sum_{i,j}\tr(\rho K_j^\dagger K_i)\cdot \ketbra{i}{j}.
\end{equation*}
Its adjoint $(\mathcal{E}^c)^\dag$ is thus $(\calE^c)^\dagger(Q)=\sum_{i,j}\bra{i}Q\ket{j}\cdot K_i^\dagger K_j$.

\begin{lemma}[Complementary-channel formula]\label{lem:resolvent-marginal}
For $s>0$,
\begin{equation}\label{eq:resolvent-marginal}
 H_s(\sigma_{\mathrm{in}})
 =(\calE^c)^\dagger\bigl((\calE^c(\sigma_{\mathrm{in}})+sI_r)^{-1}\bigr).
\end{equation}
Hence, for every $\rho\in\calD(\mathcal{H}_{\mathrm{A}})$,
\begin{equation}\label{eq:H-pairing}
 \tr\bigl(\rho H_s(\sigma_{\mathrm{in}})\bigr)
 =\tr\bigl(\calE^c(\rho)(\calE^c(\sigma_{\mathrm{in}})+sI_r)^{-1}\bigr).
\end{equation}
Moreover,
\begin{equation*}
 \tr\bigl(\sigma_{\mathrm{in}} H_s(\sigma_{\mathrm{in}})\bigr)=\tr f_s(\calE(\kettbbra{\sqrt{\sigma_{\mathrm{in}}}}{\sqrt{\sigma_{\mathrm{in}}}}))\le r,
 \qquad 0\preceq H_s(\sigma_{\mathrm{in}})\preceq s^{-1}I_{\mathrm{A}}.
\end{equation*}
\end{lemma}
\begin{proof}
Let $B_{\sigma_{\mathrm{in}}}:\Co^r\to
\mathcal{H}_{\mathrm{B}}\otimes\mathcal{H}_{\mathrm{A}}$ have columns
$\kett{K_i\sqrt{\sigma_{\mathrm{in}}}}$.  Then
\begin{equation*}
 \calE(\kettbbra{\sqrt{\sigma_{\mathrm{in}}}}{\sqrt{\sigma_{\mathrm{in}}}})=B_{\sigma_{\mathrm{in}}} B_{\sigma_{\mathrm{in}}}^\dagger,\qquad
 B_{\sigma_{\mathrm{in}}}^\dagger B_{\sigma_{\mathrm{in}}}=\calE^c(\sigma_{\mathrm{in}})^{\mathrm T}.
\end{equation*}
For $Q:=(\calE^c(\sigma_{\mathrm{in}})^{\mathrm T}+sI_r)^{-1}$ and $X:=\calE(\kettbbra{\sqrt{\sigma_{\mathrm{in}}}}{\sqrt{\sigma_{\mathrm{in}}}})$, the resolvent identity
and \cref{eq:vec-facts,eq:momega,eq:normalization-adjoint} give
\[
 \begin{aligned}
 f_s(X)&=B_{\sigma_{\mathrm{in}}} Q B_{\sigma_{\mathrm{in}}}^\dagger,\\
 H_s(\sigma_{\mathrm{in}})&=\sum_{i,j}Q_{ij}K_j^\dagger K_i
             =(\calE^c)^\dagger(Q^{\mathrm T}).
 \end{aligned}
\]
This proves \cref{eq:resolvent-marginal}; adjoint duality gives
\cref{eq:H-pairing}.  Equation~\eqref{eq:trace-transfer} with $\rho=\sigma_{\mathrm{in}}$
gives $\tr(\sigma_{\mathrm{in}} H_s(\sigma_{\mathrm{in}}))=\tr f_s(\calE(\kettbbra{\sqrt{\sigma_{\mathrm{in}}}}{\sqrt{\sigma_{\mathrm{in}}}}))\le\rank \calE(\kettbbra{\sqrt{\sigma_{\mathrm{in}}}}{\sqrt{\sigma_{\mathrm{in}}}})\le r$. 
Since
$
0\preceq
(\calE^c(\sigma_{\mathrm{in}})+sI_r)^{-1}
\preceq s^{-1}I_r,
$
positivity and unitality of $(\calE^c)^\dagger$ imply
$
0\preceq H_s(\sigma_{\mathrm{in}})\preceq s^{-1}I_{\mathrm A}.
$

\end{proof}

\begin{lemma}[Regularized channel decomposition]\label{lem:ridge-tail}
Let $X:=\calE(\kettbbra{\sqrt{\sigma_{\mathrm{in}}}}{\sqrt{\sigma_{\mathrm{in}}}})$, and reconstruct its regularization $X_s$ (see \cref{eq:ridge}) as the map
\[
 \calE_s:=\Phi_{\sigma_{\mathrm{in}}}(X_s).
\]
Then $\calE_s$ and $\calE-\calE_s$ are completely positive and
trace-nonincreasing, with
\[
 \calE-\calE_s=s\Phi_{\sigma_{\mathrm{in}}}(f_s(X)).
\]
The regularization bias and output trace are determined by
\begin{equation*}
 \begin{aligned}
 \dnorm{\calE-\calE_s}
    &=s\opnorm{H_s(\sigma_{\mathrm{in}})},\\
 \calE_s^\dagger(I_{\mathrm{B}})
    &=I_{\mathrm{A}}-sH_s(\sigma_{\mathrm{in}}).
 \end{aligned}
\end{equation*}
\end{lemma}
\begin{proof}
By \cref{eq:ridge}, $X_s\succeq0$ and
$X-X_s=sf_s(X)\succeq0$.
Reconstruction gives complete positivity and the displayed decomposition.
The norm identity follows from \cref{lem:weighted-norms}(a) and
\cref{eq:gram}; the adjoint identity follows from the decomposition,
\cref{eq:gram}, and $\calE^\dagger(I_{\mathrm{B}})=I_{\mathrm{A}}$.
Each completely positive summand is trace-nonincreasing because their sum
is $\calE$.
\end{proof}

\paragraph{The Choi-support projection.}
Let $P_\calE$ be the orthogonal projection onto $\supp C_\calE$ and define
\begin{equation}\label{eq:kappa-intro}
 \kappa(\calE)=d_1\|\tr_{\mathrm{B}} P_\calE\|_\infty.
\end{equation}
\begin{lemma}[Marginal of the Choi-support projection]\label{lem:kappa}
Let $\calE$ have Kraus rank $r$, with $\kappa(\calE)$ defined in
\cref{eq:kappa-intro}, and let $\lambda_{\min}^+(C_\calE)$ be the smallest
non-zero Choi eigenvalue of its Choi operator.  Then
\begin{equation}\label{eq:kappa-chain}
 r\le\kappa(\calE)\le
 \min\left\{\frac{d_1}{\lambda_{\min}^+(C_\calE)},
             d_1\min(r,d_2)\right\}.
\end{equation}
For every positive definite $\sigma_{\mathrm{in}}\in\calD(\mathcal{H}_{\mathrm{A}})$, let $P_{\mathcal{E},\sigma_{\mathrm{in}}}$ be the
orthogonal projection onto the support of the joint output state $\calE(\kettbbra{\sqrt{\sigma_{\mathrm{in}}}}{\sqrt{\sigma_{\mathrm{in}}}})$.  With $H_s$ defined in \cref{eq:gram}, we have
\begin{equation}\label{eq:H0}
 \lim_{s\downarrow0}H_s(\sigma_{\mathrm{in}})=\Phi_{\sigma_{\mathrm{in}}}(P_{\mathcal{E},\sigma_{\mathrm{in}}})^\dagger(I_{\mathrm{B}}).
\end{equation}
In particular,
$\lim_{s\downarrow0}\|H_s(I/d_1)\|_\infty=\kappa(\calE)$.
\end{lemma}
\begin{proof}
The positive operator $\tr_{\mathrm{B}}P_\calE$ has trace $r$ and dimension $d_1$,
so $d_1\|\tr_{\mathrm{B}}P_\calE\|_\infty\ge r$.  If
$\lambda=\lambda_{\min}^+(C_\calE)$, then
$\lambda P_\calE\preceq C_\calE$ implies
$\lambda\tr_{\mathrm{B}}P_\calE\preceq I_{\mathrm{A}}$, giving
$\kappa(\calE)\le d_1/\lambda$.
For a unit vector $\ket{v}\in \mathcal{H}_{\mathrm{A}}$,
\[
 \bra{v}\tr_{\mathrm{B}}P_\calE\ket{v}
 =\tr(P_\calE(I_{\mathrm{B}}\otimes\ketbra vv)P_\calE)
 \le\min(r,d_2),
\]
because
$
0\preceq
P_\calE(I_{\mathrm B}\otimes\ketbra vv)P_\calE
\preceq I,
$
and its rank is at most $\min(r,d_2)$.  This proves the remaining bound.
Finally, $f_s(\calE(\kettbbra{\sqrt{\sigma_{\mathrm{in}}}}{\sqrt{\sigma_{\mathrm{in}}}}))\to P_{\mathcal{E},\sigma_{\mathrm{in}}}$, so
\cref{eq:H0} follows from \cref{eq:gram} and continuity of the
reconstruction in \cref{eq:Lomega}.  At $\sigma_{\mathrm{in}}=I_{\mathrm{A}}/d_1$,
$C_{\Phi_{I/d_1}(M)}=d_1M$ and $P_{\mathcal{E},\sigma_{\mathrm{in}}}=P_\calE$, so the final claim
follows from \cref{eq:normalization-adjoint} and invariance of the
operator norm under transposition.
\end{proof}

In particular, a channel of Kraus rank at most $r$ satisfying
$\lambda_{\min}^+(C_\calE)\ge cd_1/r$ for some $c>0$ has
$\kappa(\calE)\le r/c$.  This spectral condition is used in
\cref{cor:gapped}.

\section{Upper bounds}\label{sec:upper}
In this section, we provide the algorithms that achieve our upper bounds in increasing order of difficulty and generality. In Section \ref{sec:fixed-probe}, we first treat the gapped case and show that a fixed maximally entangled input state is sufficient and adaptivity is not needed. The complexity we obtain depends on the parameter $\kappa(\calE)$ and is  optimal for  channels with $\Omega(d_1/r)$-gapped Choi spectrum (which implies $\kappa(\calE) = O(r)$). Then, in Section \ref{sec:algorithm}, we treat the general (non-gapped) setting where $\kappa(\calE)$ can be large. The near-optimal upper bound is achieved by an \textit{adaptive} algorithm that chooses the input states adaptively. Their technical lemmas appear in \cref{sec:ub-technical}.

\subsection{A fixed maximally entangled input}\label{sec:fixed-probe}

\begin{algorithm}[t]
\caption{Gapped-case incoherent quantum channel tomography}
\label{alg:gapped}

\KwIn{ $d_1, d_2, r, \epsilon, \delta, c_1,$ such that $\lambda^+_{\rm min}(C_\calE)\ge \frac{c_1d_1}{r}$  }
\KwOut{A channel estimate $\widehat \calE$}

 $D\leftarrow d_1d_2$; {$n \leftarrow \left\lceil
 c_{\rm ub}\frac{r^2}{c_1^2\epsilon^2}
 \left(D+\log\frac4\delta\right)(1+\log(8/\delta))^2
 \right\rceil$} \;
\For{$t\leftarrow 1$ \KwTo $n$}{
    Prepare $\frac{1}{\sqrt{d_1}} \sum_{i=1}^{d_1}\ket{i}_{\mathrm{A}} \otimes \ket{i}_{\mathrm{A}}$ and apply $\calE\otimes \id$  \;
    Measure BA in an independent Haar-random basis, obtaining outcome $\ket{\psi_t}$\;
}
$\widehat{X} \leftarrow \frac{1}{n} \sum_{t=1}^n \left((D+1)\ketbra{\psi_t}{\psi_t} - I_D\right)$\;
$\widetilde{X} =   \arg\min_{Y\succeq 0, \tr Y =1} \left\| Y - \widehat X\right\|_2$\;
$\widehat\calE\in\arg\min_{\calF\text{ a channel}}
       \left\|\calF-\Phi_{I/d_1}(\widetilde X)\right\|_\diamond$\;
\Return{$\widehat \calE$}\;
\end{algorithm}

The protocol of \cite{SKKG22} prepares the maximally entangled state, sends it through the unknown channel, measures the output in a random basis (see \cref{alg:gapped} for the detailed algorithm). We show that this  protocol achieves optimal query complexity when   $\kappa(\calE)$, defined in \cref{eq:kappa-intro}, is $O(r)$, which by \eqref{eq:kappa-chain} happens when the smallest non-zero Choi eigenvalue is at least $\Omega(d_1/r)$.

\paragraph{Fixed-input estimator.}\label{prot:fixed}
Prepare the same maximally entangled input $d_1^{-1/2}\kett I$ in every
experiment.  Its input marginal is $I_{\mathrm{A}}/d_1$, so the measured state is
$\calE(\kettbbra{I/\sqrt{d_1}}{I/\sqrt{d_1}})=C_\calE/d_1$. Measure the output Choi state by Haar random bases (see \cref{def-9150318}), obtaining outcomes $\{\ket{\psi_1}, \ldots, \ket{\psi_n}\}$.
Using the Haar snapshots (see \cref{def-9150321}), we define $\widehat X = \frac{1}{n}\sum_{t=1}^n (D+1)\ketbra{\psi_t}{\psi_t} - I_D$, then project it onto
$\calD(\mathcal{H}_{\mathrm{B}}\otimes \mathcal{H}_{\mathrm{A}})$ in Hilbert--Schmidt norm to obtain $\widetilde X$,
and output
\begin{equation}\label{eq:diamond-proj}
 \widehat\calE\in\arg\min_{\calF\text{ a channel}}
       \|\calF-\Phi_{I/d_1}(\widetilde X)\|_\diamond.
\end{equation}

\begin{theorem}[Upper bound for a fixed maximally entangled input]\label{thm:ub-main}
There is an absolute constant $c_{\rm ub}$ such that, for every channel $\calE$ with input dimension $d_1$ and output dimension $d_2$, the fixed-input estimator satisfies
$\Pr\{\|\widehat\calE-\calE\|_\diamond>\epsilon\}\le\delta$ whenever
\begin{equation}\label{eq:complexity}
 {n\ge c_{\rm ub}\frac{\kappa(\calE)^2}{\epsilon^2}
       \left(D+\log\frac{4}{\delta}\right)
       (1+\log(8/\delta))^2}.
\end{equation}
\end{theorem}

\begin{proof}
{Put
\[
 \theta=c_{\rm H}\left(
 \sqrt{\frac{D+\log(4/\delta)}{n}}
 +\frac{D+\log(4/\delta)}{n}\right).
\]}
By
\cref{lem:haar-moments}\ref{item-9150411}, and \cref{lem:state-proj}, with probability at least
$1-\delta/2$,
\[
 \left\|\widehat X-\calE\!\left(\kettbbra{I/\sqrt{d_1}}{I/\sqrt{d_1}}\right)\right\|_\infty\le\theta,\qquad
 \left\|\widetilde X-\calE\!\left(\kettbbra{I/\sqrt{d_1}}{I/\sqrt{d_1}}\right)\right\|_\infty\le2\theta.
\]
{The sample-size assumption ensures
$\theta\le2c_{\rm H}\sqrt{(D+\log(4/\delta))/n}$ and $2\theta<1$,
after increasing $c_{\rm ub}$.}
By \cref{lem:haar-moments}\ref{item-9150413}, the distribution of $\widehat X$ is
invariant under conjugation by any unitary fixing $\calE(\kettbbra{I/\sqrt{d_1}}{I/\sqrt{d_1}})$.  The projection
onto states commutes with unitary conjugation by \cref{lem:state-proj}.
Consequently,
\cref{lem:conversion}, with $\eta=2\theta$ and $\zeta=\delta/2$, gives
\[
 \|\Phi_{I/d_1}(\widetilde X)-\calE\|_\diamond
 \le2c_{\rm conv}\kappa(\calE)\theta(1+\log(8/\delta))
\]
with probability at least $1-\delta$.  The final projection onto channels
increases this bound by at most a factor of two: by the definition of $\widehat{\calE}$ in \cref{eq:diamond-proj}, $\|\widehat \calE-\Phi_{I/d_1}(\widetilde X)\|_\diamond \le \|\calE-\Phi_{I/d_1}(\widetilde X)\|_\diamond$ so by the triangle inequality $\|\calE - \widehat \calE\|_\diamond\le 2\|\calE-\Phi_{I/d_1}(\widetilde X)\|_\diamond$.
{Substituting the preceding bound on $\theta$}
proves \cref{eq:complexity}.
\end{proof}

\begin{theorem}[Non-zero Choi eigenvalues bounded below]\label{cor:gapped}
There is an absolute constant $c_{\rm ub}$ such that, for every $\mathcal{E}\in\qchannel_{d_1,d_2}^{r}$ with non-zero Choi eigenvalues $\geq cd_1/r$ for some $c>0$, the fixed-input estimator satisfies
$\Pr\{\|\widehat\calE-\calE\|_\diamond>\epsilon\}\le\delta$ whenever
\[
 {n\ge \frac{c_{\rm ub}}{c^2} \cdot \frac{r^2}{\epsilon^2}
       \left(D+\log\frac{4}{\delta}\right)
       (1+\log(8/\delta))^2}.
\]
\end{theorem}
\begin{proof}
The spectral assumption and \cref{eq:kappa-chain} give
$\kappa(\calE)\le r/c$.  Then, we apply \cref{thm:ub-main}.
\end{proof}

In the next section, we search adaptively for input states that spread the weight of the output evenly.

\subsection{Adaptive upper bound}\label{sec:algorithm}
We remove the gapped hypothesis of \cref{cor:gapped} at the cost of a logarithmic overhead in the query complexity. To this end, we rely on a two-batch estimator that returns an estimate of the channel and a multiplicative estimate of $H_s(\sigma_{\mathrm{in}})+rI_{\mathrm{A}}$ whose operator norm plays the role of $\kappa(\calE)$ at a general input state (rather than the maximally entangled state).
An update of the input marginal then chooses an input at which
$\|H_s(\sigma_{\mathrm{in}})\|_\infty=O(r)$.  We state the procedures first and use their
technical guarantees to prove \cref{thm-9151357}. The two-batch estimator is detailed in \cref{alg:two-batch} and the adaptive incoherent channel learning algorithm is detailed in \cref{alg:adaptive}.

\begin{definition}[Two-batch estimator]\label{def:block}
Fix a channel $\calE\in\qchannel_{d_1,d_2}^{r}$, a state
$\sigma_{\mathrm{in}}\in\calD(\mathcal{H}_{\mathrm{A}})$ with $\sigma_{\mathrm{in}}\succeq I_{\mathrm{A}}/(2d_1)$, a
regularization parameter $s>0$, and a batch size $n$.  From two
independent batches of $n$ Haar snapshots at input $\sigma_{\mathrm{in}}$, form
$\widehat X_0,\widehat X_1$ using \cref{eq:ls-n}.  Let $S$ retain the $r$
largest positive eigenvalues of $\widehat X_0$ and set its other
eigenvalues to zero.  With $P_S$ the projection onto $\ran S$, set
\begin{equation*}
 R=(S+sI_D)^{-1}P_S,
\end{equation*}
and return
\begin{equation}\label{eq:block-output}
 \widehat Z=\widehat X_1R\widehat X_1,\qquad
 \widehat H=\frac{I_{\mathrm{A}}-\Phi_{\sigma_{\mathrm{in}}}(\widehat Z)^\dagger(I_{\mathrm{B}})}s.
\end{equation}
Here the reconstruction is given by \cref{eq:Lomega}.
The operator $\widehat Z$ is positive semidefinite and the procedure
uses $2n$ queries.  The first batch chooses $R$ independently of the second
batch; this independence is used in \cref{prop:block}.
\end{definition}
\begin{algorithm}[t]
\caption{Two-batch estimator}
\label{alg:two-batch}

\KwIn{ $\calE, \sigma\succ 0, s > 0, n, d_1, d_2, r$ }
\KwOut{$\widehat Z$ and $\widehat H$}

 $D\leftarrow d_1d_2$ \;
 \For{$b\in \{0,1\}$}{
\For{$t\leftarrow 1$ \KwTo $n$}{
    Prepare $\kett{\sqrt{\sigma}} = \sum_{i,j=1}^{d_1} (\sqrt{\sigma})_{i,j}\ket{i}_{\mathrm{A}} \otimes \ket{j}_{\mathrm{A}}$ and apply $\calE\otimes \id$\;
    Measure BA in an independent Haar-random basis, obtaining outcome $\ket{\psi_{b, t}}$\;
}
$\widehat{X}_b \leftarrow \frac{1}{n} \sum_{t=1}^n \left((D+1)\ketbra{\psi_{b,t}}{\psi_{b,t}} - I_D\right)$\;
}

$S\leftarrow$ positive part of $\widehat X_0$ restricted to its $r$ largest positive eigenvalues\;
$P_S\leftarrow$ projector onto $\ran S\,; \quad R \leftarrow (S+sI_D)^{-1}P_S$\;
$\widehat Z = \widehat X_1 R \widehat X_1\,; \quad \widehat H\leftarrow\frac{1}s \left(I_{\mathrm{A}}-\Phi_{\sigma}(\widehat Z)^\dagger(I_{\mathrm{B}})\right)$\;
\Return{$(\widehat Z, \widehat H)$}\;
\end{algorithm}

\paragraph{Input update.}\label{sec:design}
For $\rho\in\calD(\mathcal{H}_{\mathrm{A}})$ define
\begin{equation}\label{eq:design-objects}
 \sigma_{\mathrm{in}}(\rho)=\frac\rho2+\frac{I_{\mathrm{A}}}{2d_1},\qquad
 A(\rho)=H_s(\sigma_{\mathrm{in}}(\rho))+rI_{\mathrm{A}}.
\end{equation}
The mixture guarantees that every queried input is bounded below by
$I_{\mathrm{A}}/(2d_1)$.  Given a positive definite estimate $A_j$ of $A(\rho_j)$, we use the matrix exponentiated-gradient update \cite{TRW05,AK16}
update
\begin{equation*}
 z_j=\tr\exp(\log\rho_j+\log A_j),\qquad
 \rho_{j+1}=\frac{\exp(\log\rho_j+\log A_j)}{z_j}.
\end{equation*}

\begin{algorithm}[t]
\caption{Adaptive incoherent quantum channel tomography}
\label{alg:adaptive}

\KwIn{ $d_1, d_2, r, \epsilon, \delta$ }
\KwOut{A channel estimate $\widehat \calE$}

 $D\leftarrow d_1d_2; \quad s\leftarrow \frac{\epsilon}{64r};\quad
 T\leftarrow\max\{1,\lceil\log_2d_1\rceil\}$\;
 $\delta'\leftarrow\frac{\delta}{T+1};\quad
 b\leftarrow {D+\left(1+\frac{d_1}{r}\right)\log\frac6{\delta'}};\quad
 n\leftarrow {\left\lceil K\frac{b}{s^2}\right\rceil}$\;
 $\rho_1 \leftarrow \frac{1}{d_1}I_{\rm A}$\;
 
\For{$j\leftarrow 1$ \KwTo $T$}{
    $\sigma_j \leftarrow \frac{1}{2}\rho_j + \frac{1}{2d_1}I_{\rm A}$\;
    $(\widehat Z_j, \widehat H_j) \leftarrow \TwoBatch(\calE, \sigma_j,s,n,d_1,d_2,r)$
    using fresh copies\;
    $A_j \leftarrow \widehat H_j +r I_{\rm A}$ with every eigenvalue below $r/2$ replaced by $r/2$ \;
   $z_j=\tr\exp\bigl(\log\rho_j+\log A_j\bigr),
 \qquad
 \rho_{j+1}=\frac{1}{z_j}\exp\bigl(\log\rho_j+\log A_j\bigr)$\;
}

$\bar{\sigma} \leftarrow \frac{1}{T} \sum_{j=1}^T \sigma_j $\;
$(\widehat Z, \widehat H) \leftarrow \TwoBatch(\calE, \bar{\sigma},s,n,d_1,d_2,r)$
    using fresh copies\;
$\widehat\calE\in\arg\min_{\calF\text{ a channel}}
       \|\calF-\Phi_{\bar{\sigma}}(\widehat Z)\|_\diamond.$\;
\Return{$\widehat \calE$}\;
\end{algorithm}

\begin{theorem}\label{thm-9151357}
There is an absolute constant $c$ such that, for every
$\calE\in\qchannel_{d_1,d_2}^{r}$, the adaptive
incoherent protocol in \cref{alg:adaptive} outputs a channel $\widehat\calE$
satisfying
\[
 \Pr\{\dnorm{\widehat\calE-\calE}>\epsilon\}\le\delta
\]
using at most
\begin{equation}\label{eq:final-count}
 {c\frac{r^2}{\epsilon^2}\log(2d_1)
 \left[
 D+\left(1+\frac{d_1}{r}\right)
 \log\left(\frac{18\log(2d_1)}{\delta}\right)
 \right]}
\end{equation}
queries.  The protocol uses $O(\log(2d_1))$ adaptive updates of the input
marginal and an independent Haar-random basis measurement after each channel
use.  
\end{theorem}

\begin{proof}
{The choice $s=\epsilon/(64r)$ satisfies
$s\le\min\{1,\sqrt{d_1/r}\}$, as required by \cref{prop:block}.}
We condition on the transcript preceding each two-batch estimator.  It fixes
the input, and both batches are fresh.  By \cref{prop:block}, with
conditional probability at least $1-\delta'$, the estimate satisfies
\[
 \frac34 A(\rho_j)\preceq\widehat H_j+rI_{\mathrm{A}}
 \preceq\frac54 A(\rho_j).
\]
The same proposition applies to the final estimator.  A conditional union
bound shows that all $T+1$ guarantees hold simultaneously with probability
at least $1-\delta$.  From now on, we work on this event.

\textbf{Input selection.}
Since
$
A(\rho_j)\succeq rI_{\mathrm{A}}
$
and the lower bound in \cref{eq:block-multiplicative} gives
$
\widehat H_j+rI_{\mathrm{A}}
\succeq
\frac34 A(\rho_j)
\succeq
\frac{3r}{4}I_{\mathrm{A}},
$
no eigenvalue is replaced by the clipping step. Hence
$
A_j=\widehat H_j+rI_{\mathrm{A}}.
$
Therefore, \cref{prop:gapping} gives
\begin{equation}\label{eq:designed}
 H_s(\bar{\sigma}_{\mathrm{in}})\preceq9rI_{\mathrm{A}}.
\end{equation}

\textbf{Estimation error.}
Conditional on the transcript from the input-selection rounds, the final
input is fixed.
Combining \cref{eq:designed} with \cref{prop:block} yields
$\|\Phi_{\bar{\sigma}_{\mathrm{in}}}(\widehat Z)-\calE\|_\diamond\le11rs$.
The true channel is feasible in the final minimization, so
\[
 \|\widehat\calE-\calE\|_\diamond
 \le2\|\Phi_{\bar{\sigma}_{\mathrm{in}}}(\widehat Z)-\calE\|_\diamond
 \le22rs=\frac{22}{64}\epsilon<\epsilon.
\]

\textbf{Query complexity.}
The algorithm uses $2n(T+1)$ queries.  Since
{$T+1=O(\log(2d_1))$ and
$\delta'=\delta/(T+1)$, substituting $s=\epsilon/(64r)$ and
$T+1\le3\log(2d_1)$ gives \cref{eq:final-count}.}

\end{proof}

\subsection{Technical lemmas}\label{sec:ub-technical}
Building on the identities proved in \cref{sec:prelim}, we establish   the error bounds for a fixed maximally entangled
input and the guarantees for regularized estimation
and input selection.

\subsubsection{Error estimates for a fixed maximally entangled input}\label{sec:fixed-lemmas}
\begin{lemma}[\cite{GKKT20}]\label{lem:rank}
Let $\rho$ be a state of rank $r$ on a $D$-dimensional space and $\sigma$ any
state.  Then
\[\trnorm{\sigma-\rho}\le2r\opnorm{\sigma-\rho}.\]
\end{lemma}
\begin{proof}
 Set $M=\sigma-\rho$.  If $M$ had $r+1$ negative eigenvalues, the span $V$ of
the corresponding eigenvectors would meet $\ker\rho$, of dimension $D-r$, in a
non-zero vector $\ket{v}$; but $\bra{v}M\ket{v}<0$ and
$\bra{v}M\ket{v}=\bra{v}\sigma\ket{v}\ge0$, a contradiction.  So $M$
has at most $r$ negative eigenvalues, and since $\tr M=0$ its positive and
negative parts have equal trace; hence
$\trnorm M=2\tr M_-\le2r\opnorm M$.  
\end{proof}
\begin{lemma}[Hilbert--Schmidt projection onto states \cite{GKKT20}]\label{lem:state-proj}
Let $\widehat X$ be Hermitian with $\tr\widehat X=1$ and let
$\widetilde X=\arg\min_{\sigma\in\calD(\Co^{D})}\hsnorm{\sigma-\widehat X}$.
Then $\widetilde X$ is obtained by keeping the eigenvectors of $\widehat X$
and replacing its eigenvalues $x_i$ by $(x_i-q)_+$, where $q\in\R$ is the
unique value with $\sum_i(x_i-q)_+=1$, and $a_+=\max\{a,0\}$.
Consequently
\begin{equation*}
 \opnorm{\widetilde X-X}\le2\opnorm{\widehat X-X}
 \quad\text{for every state }X,
\end{equation*}
and the map $\widehat X\mapsto\widetilde X$ commutes with unitary
conjugation:
$\widetilde{U\widehat XU^{\dagger}}=U\widetilde XU^{\dagger}$ for every unitary $U$.
\end{lemma}

\begin{proof}
The eigenvalue formula for the Hilbert--Schmidt projection onto density
operators is given in \cite{Smolin12} and \cite[Sec.~4.2]{GKKT20}.
Since this formula changes only the eigenvalues, the projection commutes
with unitary conjugation.

If $\widehat X=X$, the claim is immediate.  Otherwise, put $\theta=\opnorm{\widehat X-X}$ and note that
$q\mapsto\tr(\widehat X-qI)_+$ is nonincreasing.  Since
$\widehat X\preceq X+\theta I$ and $Z\mapsto\tr Z_+
=\max_{0\preceq\Pi\preceq I}\tr(Z\Pi)$ is L\"owner monotone,
$\tr(\widehat X-\theta I)_+\le\tr X_+=1$, so $q\le\theta$.  If we had
$q\le-\theta$ then $\widehat X-qI\succeq\widehat X+\theta I\succeq X\succeq0$
and hence $\tr(\widehat X-qI)_+=\tr\widehat X-qD=1-qD\ge1+\theta D>1$, a
contradiction; so $|q|\le\theta$.  Finally, for each eigenvalue,
$|(x_i-q)_+-x_i|$ equals $|q|$ if $x_i\ge q$ and $|x_i|\le\max(\theta,|q|)$
otherwise, hence is at most $\theta$ in both cases.  Thus
$\opnorm{\widetilde X-\widehat X}\le\theta$ and the triangle inequality
finishes the proof.
\end{proof}

\paragraph{Gaussian estimates.}
A standard complex Gaussian scalar has independent real and imaginary
parts, each distributed as $N(0,1/2)$.  The Gaussian matrix estimates in this subsection
use this normalization.

\begin{fact}[Real Gaussian quadratic forms {\cite[Lem.~1]{LM00}}]
\label{fact:laurent-massart}
Let $\xi_1,\xi_2,\dots$ be i.i.d.\ real standard Gaussians and $a_i\ge0$ with
$\sum_ia_i<\infty$.  Then for $u\ge0$,
\[
 \Pr\Bigl\{\sum_ia_i\xi_i^{2}\ge\sum_ia_i
 +2\|a\|_2\sqrt u+2\|a\|_\infty u\Bigr\}\le e^{-u},
 \quad
 \Pr\Bigl\{\sum_ia_i\xi_i^{2}\le\sum_ia_i-2\|a\|_2\sqrt u\Bigr\}\le e^{-u}.
\]
\end{fact}

\begin{corollary}[Complex Gaussian quadratic forms]\label{cor:gauss-sandwich}
Let $T\succeq0$ act on $\Co^{N}$ and let $\ket{g_1},\dots,\ket{g_m}$ be i.i.d.\ standard
complex Gaussian vectors in $\Co^{N}$: their coordinates have
independent real and imaginary parts distributed as $N(0,1/2)$.
Then for every $u\ge0$,
\[
 \Pr\Bigl\{\sum_{j=1}^{m}\bra{g_j}T\ket{g_j}
 \ge m\tr T+\sqrt{2m\tr(T^{2})\,u}+\opnorm Tu\Bigr\}\le e^{-u}.
\]
\end{corollary}

\begin{proof}
Diagonalize $T=\sum_k\lambda_k\ketbra{u_k}{u_k}$.  Then
$\sum_j\bra{g_j}T\ket{g_j}=\sum_k\lambda_k\sum_j|\braket{u_k}{g_j}|^{2}$, and the
$|\braket{u_k}{g_j}|^{2}$ are i.i.d.\ mean-one exponentials, i.e.\
$\tfrac12\chi^{2}_{2}$.  Hence the sum equals $\sum_ia_i\xi_i^{2}$ with each
$\lambda_k$ appearing $2m$ times with weight $a=\lambda_k/2$.  Then
$\sum_ia_i=m\tr T$, $\|a\|_2^{2}=\tfrac m2\tr(T^{2})$ and
$\|a\|_\infty=\tfrac12\opnorm T$, and \cref{fact:laurent-massart} gives
the claim.
\end{proof}

\begin{lemma}[Extreme singular values of complex Gaussian matrices]
\label{lem:gaussian-sv}
Let $G$ be an $N\times m$ matrix whose entries have independent real
and imaginary parts distributed as $N(0,1/2)$.  Write $s_{\max}(G)$
and $s_{\min}(G)$ for its largest and smallest singular values.
Then for every $t\ge0$,
\begin{equation}\label{eq:smax}
 \Pr\bigl\{s_{\max}(G)\ge4(\sqrt N+\sqrt m+t)\bigr\}\le e^{-t^{2}},
\end{equation}
and if $N\ge125\,m$,
\begin{equation*}
 \Pr\bigl\{s_{\min}(G)\le\tfrac15\sqrt N\bigr\}\le2e^{-N/16}.
\end{equation*}
(Sharper constants are classical; see Davidson and Szarek
{\cite[Thm.~II.13]{DS01}} and Vershynin {\cite[Thm.~4.6.1]{Vershynin18}}.  The following proof gives the constants used here.)
\end{lemma}

\begin{proof}
For a fixed unit $\ket{x}\in\Co^{m}$ the vector $G\ket{x}$ is standard complex Gaussian in
$\Co^{N}$, so $\|G\ket{x}\|_2^{2}=\tfrac12\chi^{2}_{2N}$, i.e.\ it is
$\sum_ia_i\xi_i^{2}$ with $2N$ weights $a_i=\tfrac12$.  Thus
$\sum_ia_i=N$, $\|a\|_2=\sqrt{N/2}$, $\|a\|_\infty=\tfrac12$, and
\cref{fact:laurent-massart} gives, for every $u\ge0$,
\begin{equation}\label{eq:one-vector}
 \Pr\bigl\{\|G\ket{x}\|_2^{2}\ge N+\sqrt{2Nu}+u\bigr\}\le e^{-u},
 \qquad
 \Pr\bigl\{\|G\ket{x}\|_2^{2}\le N-\sqrt{2Nu}\bigr\}\le e^{-u}.
\end{equation}

\emph{Largest singular value.}  By \cref{lem:covering-cardinality},
choose a $\tfrac12$-covering net $\calS$ of the unit sphere of $\Co^{m}$
with $|\calS|\le5^{2m}\le e^{3.3m}$.  Equation~\eqref{eq:covering-linear} gives
$s_{\max}(G)\le2\max_{\ket{x}\in\calS}\|G\ket{x}\|_2$.  Taking $u=3.3m+t^{2}$ in the first
bound of \cref{eq:one-vector} and a union bound,
$\max_{\ket{x}\in\calS}\|G\ket{x}\|_2^{2}\le N+\sqrt{2N(3.3m+t^{2})}+3.3m+t^{2}
\le(\sqrt N+1.9\sqrt m+t)^{2}$ except with probability $e^{-t^{2}}$, which
gives \cref{eq:smax}.

\emph{Smallest singular value.}  By \cref{lem:covering-cardinality},
choose a $\tfrac1{24}$-covering net $\calS'$ of the unit sphere of $\Co^m$
with $|\calS'|\le49^{2m}\le e^{7.8m}$.  Equation~\eqref{eq:covering-smin} gives
$s_{\min}(G)\ge\min_{\ket{x}\in\calS'}\|G\ket{x}\|_2-\tfrac1{24}s_{\max}(G)$.  Using \cref{eq:one-vector} with $u=N/8$ and a union bound gives
$\min_{\ket{x}\in\calS'}\|G\ket{x}\|_2^{2}\ge N/2$ except with probability
$e^{7.8m-N/8}\le e^{-N/16}$ when $N\ge125m$.  By \cref{eq:smax} with
$t=\sqrt N$ and $m\le N$, $s_{\max}(G)\le12\sqrt N$ except with probability
$e^{-N}$.  On the intersection,
$s_{\min}(G)\ge\sqrt{N/2}-\tfrac{12}{24}\sqrt N\ge\tfrac15\sqrt N$.
\end{proof}

\paragraph{Partial traces and diamond-norm error.}

\begin{lemma}[Partial trace after a Haar-random conjugation]\label{lem:random-complement}
There is an absolute constant $c$ with the following property.  Let
$\eta>0$ and $S\subseteq \mathcal{H}_{\mathrm{B}}\otimes \mathcal{H}_{\mathrm{A}}$ have $\dim S=r$, and let
$T\succeq0$ be supported on
$S^{\perp}$ with
\[
 \opnorm T\le\eta,\qquad\tr T\le2r\eta,\qquad rd_2\ge d_1 ,
\]
and let $V$ be Haar distributed on the unitary group of $S^{\perp}$.  Then for
every $0<\zeta<1$, with probability at least $1-\zeta$,
\begin{equation}\label{eq:random-complement}
 d_1\opnorm{\tr_{\mathrm{B}}(VTV^{\dagger})}\le cr\eta\,\ell,
 \qquad\ell:=1+\log(2/\zeta).
\end{equation}
\end{lemma}

\begin{proof}
Write $Z=VTV^{\dagger}$, which is again positive, supported on $S^{\perp}$, with the
same norm and trace as $T$.  We can easily see two deterministic bounds:
\begin{equation*}
 \text{(D1)}\quad d_1\opnorm{\tr_{\mathrm{B}}Z}\le d_1d_2\opnorm Z\le D\eta,
 \qquad
 \text{(D2)}\quad d_1\opnorm{\tr_{\mathrm{B}}Z}\le d_1\tr Z\le2d_1r\eta ,
\end{equation*}
using $\bra{\varphi}\tr_{\mathrm{B}}Z\ket{\varphi}=\tr(Z(I_{\mathrm{B}}\otimes \ketbra{\varphi}{\varphi}))
\le\min\{d_2\opnorm Z,\tr Z\}$.  Let $K$ be a large absolute constant, chosen
at the end.  If $D\le Kr\ell$ then (D1) already gives
\cref{eq:random-complement} with $c=K$; if $d_1\le K$ then (D2) gives it with
$c=2K$.  So assume from now on
\begin{equation}\label{eq:main-regime}
 D>Kr\ell\qquad\text{and}\qquad d_1>K .
\end{equation}
Set $D'=\dim S^{\perp}=D-r$.  From \cref{eq:main-regime}, $r<D/K$, so
$D'>D/2$ for $K\ge2$; from $d_1\le rd_2$ we get $D\le rd_2^{2}$, hence
$d_2^{2}>K\ell\ge K$ and $d_2>\sqrt K$; and
$D'>D/2=d_1d_2/2\ge Kd_2/2\ge125\,d_2$ for $K\ge250$.

\smallskip\noindent\textbf{Step 1 (covering net).}
By \cref{lem:covering-cardinality}, choose a $\tfrac14$-covering net $\calN$
of the unit sphere of $\mathcal{H}_{\mathrm{A}}$ with $|\calN|\le9^{2d_1}$.
Applying \cref{eq:covering-hermitian} to the positive operator
$\tr_{\mathrm{B}}Z$ gives
\[
 \opnorm{\tr_{\mathrm{B}}Z}
 \le2\max_{\ket{\varphi}\in\calN}\bra{\varphi}\tr_{\mathrm{B}}Z\ket{\varphi}.
\]
Thus it suffices to bound this quadratic form for each fixed $\ket{\varphi}$
and take a union bound.  Set
\begin{equation}\label{eq:u-choice}
 u:=\log\bigl(3|\calN|/\zeta\bigr)\le 2d_1\log9+\log(3/\zeta)
   \le4.4\,d_1+2\ell .
\end{equation}

\smallskip\noindent\textbf{Step 2 (reduction to a projection of rank at most $d_2$).}
Fix a unit $\ket{\varphi}\in \mathcal{H}_{\mathrm{A}}$.  Because $Z$ is supported on $S^{\perp}$,
\[
 \bra{\varphi}\tr_{\mathrm{B}}Z\ket{\varphi}
 =\tr\bigl(Z\,Q_\varphi\bigr),
 \qquad
 Q_\varphi:=\Pi_{S^{\perp}}(I_{\mathrm{B}}\otimes \ketbra{\varphi}{\varphi})\Pi_{S^{\perp}} .
\]
$Q_\varphi$ is a positive contraction of rank at most $d_2$; let $R_\varphi$
be the orthogonal projection onto its range, of rank $m\le d_2$.  Then
$Q_\varphi\preceq R_\varphi$ and hence
$\bra{\varphi}\tr_{\mathrm{B}}Z\ket{\varphi}\le\tr(TV^{\dagger}R_\varphi V)$.

\smallskip\noindent\textbf{Step 3 (Gaussian model for the random subspace).}
If $m=0$, the quadratic form is zero and the desired bound is immediate.
Assume $m\ge1$.  Then $V^{\dagger}R_\varphi V$ is a Haar-random rank-$m$ orthogonal projection of
$S^{\perp}\simeq\Co^{D'}$, so it is distributed as $G(G^{\dagger}G)^{-1}G^{\dagger}$ with
$G$ an $D'\times m$ standard complex Gaussian matrix.  Since
$D'\ge125\,d_2\ge125\,m$, \cref{lem:gaussian-sv} gives
$\Pr\bigl\{G^{\dagger}G\succeq\tfrac{D'}{25}I_m\bigr\}\ge1-2e^{-D'/16}$, and on that event
\[
 \tr\bigl(TG(G^{\dagger}G)^{-1}G^{\dagger}\bigr)
 =\tr\bigl((G^{\dagger}TG)(G^{\dagger}G)^{-1}\bigr)
 \le\frac{25}{D'}\tr(G^{\dagger}TG).
\]
\cref{cor:gauss-sandwich} bounds
$\tr(G^{\dagger}TG)\le m\tr T+\sqrt{2m\tr(T^{2})u}+\eta u$ except with probability
$e^{-u}$.

\smallskip\noindent\textbf{Step 4 (bounding the exceptional probabilities).}
By \cref{eq:main-regime} and $d_2>\sqrt K$ we have
$D'>D/2\ge d_1\sqrt K/2$, and also $D'>Kr\ell/2\ge K\ell/2$.  Hence
$D'/16\ge\tfrac{\sqrt K}{64}d_1+\tfrac K{64}\ell$, which for $K$ a large enough
absolute constant exceeds $4.4d_1+2\ell+\log2\ge u+\log2$ by
\cref{eq:u-choice}.  Therefore $2e^{-D'/16}\le e^{-u}$ and the two exceptional
events together have probability at most $2e^{-u}\le\zeta/|\calN|$.

\smallskip\noindent\textbf{Step 5 (union bound and final estimate).}
On the intersection over $\ket{\varphi}\in\calN$, an event of probability at least
$1-\zeta$, using $m\le d_2$, $\tr T\le2r\eta$,
$\tr(T^{2})\le\opnorm T\tr T\le2r\eta^{2}$ and $D'>D/2$,
\[
 d_1\opnorm{\tr_{\mathrm{B}}Z}
 \le\frac{50\,d_1}{D'}\Bigl(2r\eta d_2+2\eta\sqrt{d_2ru}+\eta u\Bigr)
 \le\frac{100}{d_2}\Bigl(2r\eta d_2+2\eta\sqrt{d_2ru}+\eta u\Bigr).
\]
By \cref{eq:u-choice} and $d_1\le rd_2$ we have
$u/d_2\le4.4r+2\ell\le7r\ell$ and
$\sqrt{ru/d_2}\le\sqrt{4.4r^{2}+2r\ell}\le2.1r+1.5r\ell\le4r\ell$.
Therefore $d_1\opnorm{\tr_{\mathrm{B}}Z}\le(200+800+700)\,r\eta\ell$, which is
\cref{eq:random-complement}.
\end{proof}

\begin{lemma}[Diamond-norm error under unitary invariance]\label{lem:conversion}
There are absolute constants $c,c_{\rm conv}$ with the following property.  Let $\calE\in\qchannel_{d_1,d_2}^r$ have Kraus rank exactly $r$.  Let $P$ be the
orthogonal projector onto $\supp(\calE(\kettbbra{I/\sqrt{d_1}}{I/\sqrt{d_1}}))$, set $P^\perp=I_D-P$, and put
$\kappa=\kappa(\calE)=d_1\|\tr_{\mathrm{B}}P\|_\infty$ as in
\cref{eq:kappa-intro}.  Recall the reconstruction map
$\Phi_{I/d_1}(\cdot)$ defined in \cref{eq:Lomega}, where
$C_{\Phi_{I/d_1}(M)}=d_1M$.  Let $\widetilde X$ be a random density
operator on $\mathcal{H}_{\mathrm{B}}\otimes \mathcal{H}_{\mathrm{A}}$
whose distribution is invariant under conjugation by every unitary of the form
$U=I_P\oplus V$, where $I_P$ is the identity on $\supp(P)$ and $V$ is
unitary on $\supp(P^{\perp})$.  Set $\Delta=\widetilde X-\calE(\kettbbra{I/\sqrt{d_1}}{I/\sqrt{d_1}})$.  Then for every $\eta,\zeta\in(0,1)$, except on an event of probability at most $\zeta$, $\|\Delta\|_\infty \leq \eta$ implies
\begin{equation*}
  \dnorm{\Phi_{I/d_1}(\widetilde X)-\calE}\leq c_{\rm conv}\,\kappa\eta\,\ell,
 \qquad \ell:=1+\log(4/\zeta).
\end{equation*}
More precisely, we have
\begin{equation}\label{eq:conversion-terms}
 \dnorm{\Phi_{I/d_1}(\widetilde X)-\calE}
 \le\underbrace{cr\eta\ell}_{P^\perp\Delta P^\perp}
   +\underbrace{2\eta\sqrt{cr\kappa\ell}}_{\text{off-diagonal terms}}
   +\underbrace{\kappa\eta}_{P\Delta P}.
\end{equation}
\end{lemma}

\begin{proof}
Write $L_\rho=L_{\rho,I/d_1}=I_{\mathrm{B}}\otimes(\sqrt{d_1\rho^{\mathrm T}})$, so that by
\cref{lem:weighted-norms}
\begin{equation}\label{eq:conv-master}
 \dnorm{\Phi_{I/d_1}(\Delta)}=\max_{\rho\in\calD(\mathcal{H}_{\mathrm{A}})}\trnorm{L_\rho\Delta L_\rho}.
\end{equation}
Also, $ \tr(L_\rho GL_\rho)=d_1\tr\bigl(G(I_{\mathrm{B}}\otimes \rho^{\mathrm T})\bigr) \le d_1\opnorm{\tr_{\mathrm{B}}G}$ for $G\succeq0$.  Throughout, we will work on the event $\mathsf E=\{\opnorm\Delta\le\eta\}$.
Note that the distribution of $\Delta$ is invariant under the group $\calG=\{I_P\oplus V\}$. On $\mathsf E$, \cref{lem:rank} gives $\trnorm\Delta\le2r\eta$.  Put
\[
 T_1=P^{\perp}\Delta P^{\perp},\qquad T_2=P^{\perp}|\Delta|P^{\perp}.
\]
Both are positive semidefinite (for $T_1$, we know $P^{\perp}\calE(\kettbbra{I/\sqrt{d_1}}{I/\sqrt{d_1}})P^{\perp}=0$ and $\widetilde X\succeq0$), and on $\mathsf E$ both satisfy $\opnorm{T_i}\le\eta$ and $\tr T_i\le\trnorm\Delta\le2r\eta$. We also know $rd_2\ge d_1$ since $\calE$ is a quantum channel.

\smallskip\noindent\textbf{Randomization.}
Let $V$ be Haar-random on $\supp(P^{\perp})$, and $U=I_P\oplus V$. We know that $U\Delta U^{\dagger}$ has the same distribution as $\Delta$, and $T_i(U\Delta U^{\dagger})=UT_i(\Delta)U^{\dagger}$
because $P^{\perp}$ commutes with $U$ and $|U\Delta U^{\dagger}|=U|\Delta|U^{\dagger}$.
Consequently, for the event
\[
 \calF:=\mathsf E\cap\bigl\{d_1\opnorm{\tr_{\mathrm{B}}T_i}>cr\eta\ell\ \text{for some }i\in\{1,2\}\bigr\},
\]
we have $\Pr(\calF)=\Pr(\calF')$ where $\calF'$ is the same event computed for
$U\Delta U^{\dagger}$.  Conditioning on $\Delta$ and applying
\cref{lem:random-complement} with $\zeta/2$ to each of the two fixed
admissible operators $T_1,T_2$ (legitimate because its conclusion holds for
every fixed admissible $T$) gives $\Pr(\calF)\le\zeta$.  From here on we work
on $\mathsf E\setminus\calF$, i.e.\ we may use
\begin{equation}\label{eq:two-random-bounds}
 d_1\opnorm{\tr_{\mathrm{B}}T_i}\le cr\eta\ell,\qquad i=1,2 .
\end{equation}

\smallskip\noindent\textbf{The four blocks.}
Fix $\rho\in\calD(\mathcal{H}_{\mathrm{A}})$ and split
$\Delta=P\Delta P+P\Delta P^{\perp}+P^{\perp}\Delta P+T_1$.

\emph{The $P^\perp\Delta P^\perp$ term.}  $L_\rho T_1L_\rho\succeq0$, so its trace norm is its
trace, and \cref{eq:conv-master} with \cref{eq:two-random-bounds} gives
$\trnorm{L_\rho T_1L_\rho}\le cr\eta\ell$.

\emph{The $P\Delta P$ term.}  $P\Delta P$ is supported in $\ran P$ and has operator
norm at most $\eta$, so \cref{lem:loewner-tools}(ii) and
\cref{eq:conv-master} give
\[
 \trnorm{L_\rho P\Delta PL_\rho}\le\eta\,\tr(L_\rho PL_\rho)
 \le\eta\, d_1\opnorm{\tr_{\mathrm{B}}P}=\kappa\eta .
\]
The coefficient is $\kappa(\calE)$ from \cref{eq:kappa-intro}.

\emph{The off-diagonal terms.}  By $\trnorm{MK}\le\hsnorm M\hsnorm K$ with
$M=L_\rho P^{\perp}\Delta$ and $K=PL_\rho$,
\[
 \trnorm{L_\rho P^{\perp}\Delta PL_\rho}^{2}
 \le\tr\bigl(L_\rho P^{\perp}\Delta^{2}P^{\perp}L_\rho\bigr)\,
    \tr\bigl(L_\rho PL_\rho\bigr).
\]
On $\mathsf E$ we have $\Delta^{2}\preceq\eta|\Delta|$, hence
$P^{\perp}\Delta^{2}P^{\perp}\preceq\eta T_2$; using \cref{eq:conv-master}
twice, once with $T_2$ and \cref{eq:two-random-bounds} and once with $P$,
\[
 \trnorm{L_\rho P^{\perp}\Delta PL_\rho}^{2}
 \le\bigl(\eta\cdot cr\eta\ell\bigr)\cdot\kappa,
 \qquad\text{i.e.}\qquad
 \trnorm{L_\rho P^{\perp}\Delta PL_\rho}\le\eta\sqrt{cr\kappa\ell},
\]
and the adjoint block obeys the same bound.

Adding the four contributions and maximizing over $\rho$ gives
\cref{eq:conversion-terms}.  Finally $r\le\kappa$ by
\cref{lem:kappa} and $\ell\ge1$, so
$\sqrt{r\kappa\ell}\le\kappa\ell$ and each of the three terms is at most a
constant multiple of $\kappa\eta\ell$.
\end{proof}

\subsubsection{Regularized estimation}\label{sec:ridge}

\begin{lemma}[Positive semidefinite rank truncation]\label{lem:rank-trunc}
Let $X\succeq0$ be a $D\times D$ matrix of rank at most $r$, and
let $\widehat X$ be Hermitian with
$\opnorm{\widehat X-X}\le\theta$.  Let $S$ be obtained from $\widehat X$ by
keeping its $r$ largest eigenvalues if they are positive and setting all other
eigenvalues to zero.  Then
\begin{equation}\label{eq:rank-trunc}
 S\succeq0,\qquad\rank S\le r,\qquad\opnorm{S-X}\le2\theta .
\end{equation}
\end{lemma}

\begin{proof}
By Weyl's inequality every eigenvalue of $\widehat X$ that is negative has
modulus at most $\theta$, and if $r<D$ its $(r+1)$-st largest eigenvalue is at
most $\lambda_{r+1}(X)+\theta=\theta$.  Every eigenvalue discarded in forming
$S$ therefore has modulus at most $\theta$, so
$\opnorm{S-\widehat X}\le\theta$, and the triangle inequality gives
\cref{eq:rank-trunc}.
\end{proof}

\begin{lemma}[Weighted quadratic estimate]\label{lem:weighted-conc}
There is an absolute constant $c_{\rm B}$ with the following property.
Let $\sigma_{\mathrm{in}}\in\calD(\mathcal{H}_{\mathrm{A}})$ satisfy $\sigma_{\mathrm{in}}\succeq I_{\mathrm{A}}/(2d_1)$, assume
$d_1\le rd_2$ and $1\le r\le D$, and use the reconstruction in
\cref{eq:Lomega}.  Let $X$ be a
state on $\mathcal{H}_{\mathrm{B}}\otimes \mathcal{H}_{\mathrm{A}}$, let $\widehat X$ be the linear estimator in
\cref{eq:ls-n} built from $n$ independent Haar snapshots of $X$, and put
$E=\widehat X-X$.  Let $W$ be a fixed $D\times k$ matrix of rank at most $r$,
independent of those snapshots.  Then, for every $0<\delta<1$, with probability
at least $1-\delta$,
\begin{equation}\label{eq:weighted-conc}
 \opnorm{\Phi_{\sigma_{\mathrm{in}}}(EWW^{\dagger}E)^{\dagger}(I_{\mathrm{B}})}^{1/2}
 \;\le\;{c_{\rm B}\opnorm W\sqrt{d_1}
 \left(
 \sqrt{\frac{rd_2+\log(4/\delta)}{n}}
 +\frac{rd_2+\log(4/\delta)}{n}
 \right)}.
\end{equation}
The statement holds conditionally on any earlier classical transcript that
determines $\sigma_{\mathrm{in}}$ and $W$, provided the $n$ snapshots come from fresh measurements.
\end{lemma}

{
\begin{proof}

A compact singular value decomposition replaces $W$ by a factor with at most
$r$ columns, so we may assume $k\le r$ without
changing $WW^{\dagger}$ or $\opnorm W$; the case $W=0$ is trivial.  Write
$\ket{w_1},\dots,\ket{w_k}$ for the columns of $W$ and define, for a $D\times D$ matrix
$M$, the $(kd_2)\times d_1$ matrix
\begin{equation*}
 \calA_W(M):=
 \begin{pmatrix}\mathrm{mat}(M\ket{w_1})\,\sigma_{\mathrm{in}}^{-1/2}\\ \vdots\\
                \mathrm{mat}(M\ket{w_k})\,\sigma_{\mathrm{in}}^{-1/2}\end{pmatrix}.
\end{equation*}
By \cref{eq:vec-facts}, $\mathrm{mat}(\ket{v})^{\dagger}\mathrm{mat}(\ket{v})
=(\tr_{\mathrm{B}}\ketbra{v}{v})^{\mathrm T}$, so summing over the blocks gives the
identity
\begin{equation}\label{eq:AW-identity}
 \calA_W(E)^\dagger\calA_W(E)
 =\Phi_{\sigma_{\mathrm{in}}}(EWW^\dagger E)^\dagger(I_{\mathrm B}),
 \qquad
 \opnorm{\calA_W(E)}
 =\opnorm{\Phi_{\sigma_{\mathrm{in}}}(EWW^\dagger E)^\dagger(I_{\mathrm B})}^{1/2}.
\end{equation}
It therefore suffices to control $\opnorm{\calA_W(E)}$.

\medskip
\noindent\textbf{Step 1: direct Haar moment bounds.}
Let $\ket{\varphi}$ be Haar distributed on the unit sphere of $\Co^D$.  We claim
that, for every complex $D\times D$ matrix $M$ and every real $p\ge2$,
\begin{equation}\label{eq:haar-HS-Lp}
 \left(\E\left|
 D\bra{\varphi}M\ket{\varphi}-\tr M
 \right|^p\right)^{1/p}
 \le c_1 p\hsnorm M .
\end{equation}

We first take $M=M^\dagger$ and $\tr M=0$.  For every positive integer $m$,
the Haar moment identity \cref{eq:sym-moment} gives
\begin{align}
 \mathbb E
 \left(D\bra{\varphi}M\ket{\varphi}\right)^m
 &=
 \frac{D^m}{D(D+1)\cdots(D+m-1)}
 \sum_{\pi\in S_m}
 \prod_{c\in\mathcal C(\pi)}
 \tr(M^{|c|}),
 \label{eq:haar-cycle-moment-HS}
\end{align}
 where $\mathcal C(\pi)$ is the set of cycles of $\pi$ and $|c|$ is the length of the cycle $c$.

If $\pi$ has a one-cycle, its contribution to
\cref{eq:haar-cycle-moment-HS} vanishes because $\tr M=0$.  For every
cycle length $\ell\ge2$,
\[
 |\tr(M^\ell)|
 \le \sum_i|\lambda_i(M)|^\ell
 \le \left(\sum_i|\lambda_i(M)|^2\right)^{\ell/2}
 =\hsnorm M^\ell.
\]
Therefore, for every even integer $m\ge2$,
\begin{align*}
 \mathbb E
 \left|D\bra{\varphi}M\ket{\varphi}\right|^m
 &\le
 \frac{D^m}{D(D+1)\cdots(D+m-1)}
 \sum_{\pi\in S_m}\hsnorm M^m\le m!\,\hsnorm M^m.
\end{align*}
Hence
\begin{equation}\label{eq:haar-even-moment}
 \left(\E\left|
 D\bra{\varphi}M\ket{\varphi}
 \right|^m\right)^{1/m}
 \le (m!)^{1/m}\hsnorm M
 \le m\hsnorm M.
\end{equation}

For a general Hermitian $M$, set
\[
 M_0=M-\frac{\tr M}{D}I_D.
\]
Then
\[
 \hsnorm{M_0}^2
 =\hsnorm M^2-\frac{|\tr M|^2}{D}
 \le\hsnorm M^2.
\]
Given $p\ge2$, choose an even integer $m$ such that $p\le m\le p+2$.
Monotonicity of $L^p$ norms and \cref{eq:haar-even-moment} then give
\[
 \left(\E\left|
 D\bra{\varphi}M\ket{\varphi}-\tr M
 \right|^p\right)^{1/p}
 \le m\hsnorm M
 \le 2p\hsnorm M.
\]

Finally, for a general complex matrix $M$, write
\[
 M=M_1+iM_2,\qquad
 M_1=\frac{M+M^\dagger}{2},\qquad
 M_2=\frac{M-M^\dagger}{2i}.
\]
Both $M_1$ and $M_2$ are Hermitian and
\[
 \hsnorm{M_1}^2+\hsnorm{M_2}^2=\hsnorm M^2.
\]
Applying the Hermitian estimate to the real and imaginary parts and using
the triangle inequality proves \cref{eq:haar-HS-Lp}.

\medskip
\noindent\textbf{Step 2: transfer to the observed Haar outcome.}
Let $\ket{\psi}$ be the outcome obtained by measuring the state $X$ in a
Haar-random basis.  By \cref{lem:haar-moments}(i), $\ket{\psi}$ has density
\[
 q_X(\ket{\varphi})=D\bra{\varphi}X\ket{\varphi}
\]
with respect to spherical Haar measure.  The second Haar moment gives
\begin{align}
 \E q_X^2
 &=D^2\mathbb E_{\mathrm{Haar}}
   \bigl(\bra{\varphi}X\ket{\varphi}\bigr)^2 \notag\\
 &=\frac{D\bigl((\tr X)^2+\tr(X^2)\bigr)}{D+1}
 =\frac{D(1+\tr(X^2))}{D+1}
 \le2,
 \label{eq:qX-L2}
\end{align}
where we used $\tr X=1$ and $\tr(X^2)\le1$.

For the snapshot
\[
 Y=(D+1)\ketbra{\psi}{\psi}-I_D,
\]
and a fixed complex matrix $M$, define
\[
 \xi_M:=\tr((Y-X)M).
\]
As a function of a Haar vector $\ket{\varphi}$,
\begin{align}
 \xi_M(\ket{\varphi})
 &=(D+1)\bra{\varphi}M\ket{\varphi}
   -\tr M-\tr(XM) \notag\\
 &=\frac{D+1}{D}
   \left(
    D\bra{\varphi}M\ket{\varphi}-\tr M
   \right)
   +\frac{\tr M}{D}-\tr(XM).
 \label{eq:xi-decomposition}
\end{align}
The deterministic term is bounded by
\begin{align}
 \left|\frac{\tr M}{D}-\tr(XM)\right|
 &\le \frac{|\tr M|}{D}+|\tr(XM)| \notag\\
 &\le \frac{\hsnorm M}{\sqrt D}
       +\|X\|_1\opnorm M
 \le2\hsnorm M.
 \label{eq:xi-centering-term}
\end{align}
Combining \cref{eq:haar-HS-Lp,eq:xi-decomposition,eq:xi-centering-term}
shows that, for every $p\ge2$,
\[
 \left(\E |\xi_M|^{2p}\right)^{1/(2p)}
 \le c_1 p\hsnorm M.
\]
Cauchy--Schwarz and \cref{eq:qX-L2} now give
\begin{align}
 \left(\E_{q_X}|\xi_M|^p\right)^{1/p}
 &=
 \left(
  \mathbb E_{\mathrm{Haar}}
  \left[q_X(\ket{\varphi})|\xi_M(\ket{\varphi})|^p\right]
 \right)^{1/p} \notag\\
 &\le
 \left(\E|q_X|^2\right)^{1/(2p)}\left(\E
 |\xi_M|^{2p}\right)^{1/(2p)}
 \le c_1 p\hsnorm M.
 \nonumber
\end{align}
For $1\le p<2$, the same estimate follows by monotonicity from the
case $p=2$. Thus, by the equivalent moment characterization of the
$\psi_1$ norm stated in \cref{fact:bernstein}, there is an absolute
constant $c_1$ such that
\begin{equation}\label{eq:biased-psi1}
 \|\xi_M\|_{\psi_1}
 \le  c_1\hsnorm M.
\end{equation}
Moreover,
\[
 \mathbb E_{q_X}\xi_M
 =\tr((\mathbb EY-X)M)=0
\]
by \cref{lem:haar-moments}(ii).

Let $Y_1,\dots,Y_n$ be the independent snapshots forming $\widehat X$ and
write
\[
 \xi_{M,t}:=\tr((Y_t-X)M).
\]
The scalar Bernstein inequality for independent centered subexponential
variables (\cref{fact:bernstein}), applied to \cref{eq:biased-psi1}, gives, for every $u\ge1$,
\begin{equation}\label{eq:biased-hw-average}
 \Pr\left\{
 \left|
 \frac1n\sum_{t=1}^n\xi_{M,t}
 \right|
 >
 c_1\hsnorm M
 \left(
  \sqrt{\frac un}+\frac un
 \right)
 \right\}
 \le2e^{-u}.
\end{equation}
For a complex $M$, this follows by applying the real-valued inequality to
the Hermitian and skew-Hermitian parts and taking a union bound; the factor
from this union bound is absorbed by replacing $u$ with $u+\log2$ and
enlarging the absolute constant.  Since
\[
 E=\widehat X-X=\frac1n\sum_{t=1}^n(Y_t-X),
\]
\cref{eq:biased-hw-average} can equivalently be written as
\begin{equation}\label{eq:biased-hw-E}
 \Pr\left\{
 |\tr(EM)|
 >
 c_1\hsnorm M
 \left(
  \sqrt{\frac un}+\frac un
 \right)
 \right\}
 \le2e^{-u}.
\end{equation}

\medskip
\noindent\textbf{Step 3: fixed bilinear forms of $\calA_W(E)$.}
Fix unit vectors
\[
 \ket{a}\in\Co^{kd_2},
 \qquad
 \ket{v}\in\Co^{d_1},
\]
and decompose $\ket{a}=(\ket{a_1},\dots,\ket{a_k})$ with $\ket{a_j}\in\Co^{d_2}$.  Define
\[
 \ket{h}:=\sigma_{\mathrm{in}}^{-1/2}\ket{v},
 \qquad
 \ket{q_j}:=\ket{a_j} \otimes \ket{h^\star},
 \qquad
 M_{a,v}:=\sum_{j=1}^k \ket{w_j}\bra{q_j}.
\]
For every $D\times D$ matrix $M$, we have
\[
 \bra{a_j}\mathrm{mat}(M\ket{w_j})\ket{h} = (\bra{a_j}\otimes \bra{h^\star})M\ket{w_j}.
\]
Therefore,
\begin{align}
 \bra{a}\calA_W(M)\ket{v}
 &=\sum_{j=1}^k
   \bra{a_j}\mathrm{mat}(M\ket{w_j})
   \sigma_{\mathrm{in}}^{-1/2}\ket{v} \notag\\
 &=\sum_{j=1}^k \bra{q_j} M\ket{w_j}
 =\tr(MM_{a,v}).
 \label{eq:AW-bilinear}
\end{align}

Let $Q$ be the $D\times k$ matrix with columns $\ket{q_1},\dots,\ket{q_k}$.
Then $M_{a,v}=WQ^\dagger$.  Moreover,
\begin{align}
 \hsnorm Q^2
 &=\sum_{j=1}^k\|\ket{a_j}\otimes \ket{h^\star}\|_2^2
 =\sum_{j=1}^k\|\ket{a_j}\|_2^2\|\ket{h^\star}\|_2^2 \notag\\
 &=\|\ket{a}\|_2^2
   \bra{v}\sigma_{\mathrm{in}}^{-1}\ket{v}
 \le2d_1,
 \nonumber
\end{align}
because $\sigma_{\mathrm{in}}\succeq I_{\mathrm A}/(2d_1)$.  It follows that
\begin{equation}\label{eq:test-matrix-norms}
 \hsnorm{M_{a,v}}
 =\hsnorm{WQ^\dagger}
 \le\opnorm W\,\hsnorm Q
 \le\sqrt{2d_1}\,\opnorm W.
\end{equation}
Applying \cref{eq:biased-hw-E} with $M=M_{a,v}$ and using
\cref{eq:AW-bilinear,eq:test-matrix-norms}, we obtain, for every fixed
pair of unit vectors $(\ket{a},\ket{v})$ and every $u\ge1$,
\begin{equation}\label{eq:AW-fixed-pair}
 \Pr\left\{
 |\bra{a}\calA_W(E)\ket{v}|
 >
 c_1\opnorm W\sqrt{d_1}
 \left(
  \sqrt{\frac un}+\frac un
 \right)
 \right\}
 \le2e^{-u}.
\end{equation}

\medskip
\noindent\textbf{Step 4: passage from fixed vectors to the operator norm.}
By \cref{lem:covering-cardinality}, choose deterministic
$1/4$-covering nets $\calN_1$ and $\calN_2$ of the unit spheres of
$\Co^{kd_2}$ and $\Co^{d_1}$, respectively, with
\[
 |\calN_1|\le9^{2kd_2},
 \qquad
 |\calN_2|\le9^{2d_1}.
\]
Set
\[
 u=c_0(kd_2+d_1)+\log\frac4\delta,
\]
where $c_0$ is a sufficiently large absolute constant.  Taking a union
bound in \cref{eq:AW-fixed-pair} over
$\calN_1\times\calN_2$, we find that, with probability at least
$1-\delta$,
\begin{align}
 \max_{\substack{\ket{a}\in\calN_1\\\ket{v}\in\calN_2}}
 |\bra{a}\calA_W(E)\ket{v}|
 &\le
 c_1\opnorm W\sqrt{d_1}
 \left(
  \sqrt{\frac{kd_2+d_1+\log(4/\delta)}{n}}
  +\frac{kd_2+d_1+\log(4/\delta)}{n}
 \right).
 \label{eq:AW-net-bound}
\end{align}

The passage from this maximum to the operator norm is
\cref{eq:bilinear-net} with $\eta=1/4$, applied to $B=\calA_W(E)$.

Since $k\le r$ and $d_1\le rd_2$,
\[
 kd_2+d_1\le2rd_2.
\]
Combining \cref{eq:AW-net-bound,eq:bilinear-net} and absorbing numerical
constants gives
\[
 \opnorm{\calA_W(E)}
 \le
 c_{\rm B}\opnorm W\sqrt{d_1}
 \left(
  \sqrt{\frac{rd_2+\log(4/\delta)}{n}}
  +\frac{rd_2+\log(4/\delta)}{n}
 \right).
\]
Finally, \cref{eq:AW-identity} yields
\[
 \opnorm{
 \Phi_{\sigma_{\mathrm{in}}}(EWW^\dagger E)^\dagger(I_{\mathrm B})
 }^{1/2}
 \le
 c_{\rm B}\opnorm W\sqrt{d_1}
 \left(
  \sqrt{\frac{rd_2+\log(4/\delta)}{n}}
  +\frac{rd_2+\log(4/\delta)}{n}
 \right),
\]
which is \cref{eq:weighted-conc}.
\end{proof}

}

\begin{lemma}[Stability of the regularized quadratic form]\label{lem:sandwich}
Let $s>0$ and let $X,S\succeq0$ be $D\times D$ matrices with
$\opnorm{S-X}\le\eta\le s/4$.  Let $F=f_s(X)$ and $Z=X_s$
be the matrices defined in \cref{eq:ridge}, let $P_S$ be the
projection onto $\ran S$, and set
\begin{equation*}
 R:=(S+sI_D)^{-1}P_S .
\end{equation*}
Then $0\preceq R\preceq s^{-1}I_D$, $\rank R\le\rank S$, and
\begin{equation}\label{eq:sandwich}
 -3\eta F\ \preceq\ XRX-Z\ \preceq\ 3\eta F .
\end{equation}
Moreover, writing $F^+$ for the Moore--Penrose inverse and $F^{+1/2}=(F^+)^{1/2}$,
\begin{equation}\label{eq:W-def}
 W:=RXF^{+1/2}
 \quad\text{satisfies}\quad
 WF^{1/2}=RX,\qquad\rank W\le\rank S,\qquad\opnorm W\le2 .
\end{equation}
\end{lemma}

\begin{proof}
Write $R_0=(S+sI)^{-1}$ and $R_X=(X+sI)^{-1}$, so $F=XR_X=R_XX$ and
$Z=XR_XX$.  Iterating the resolvent identity
$R_0=R_X-R_X(S-X)R_0$ once gives the exact second-order form
\[
 R_0=R_X-R_X(S-X)R_X+R_X(S-X)R_0(S-X)R_X ,
\]
and multiplying by $X$ on both sides,
\begin{equation*}
 XR_0X-Z=-F(S-X)F+F(S-X)R_0(S-X)F .
\end{equation*}
Since $-\eta I\preceq S-X\preceq\eta I$ and
$0\preceq R_0\preceq s^{-1}I$, conjugation gives
\begin{equation}\label{eq:second-order-bound}
 -\eta F^{2}\ \preceq\ XR_0X-Z\ \preceq\ (\eta+\eta^{2}/s)F^{2}.
\end{equation}
Next, $SQ=0$ for $Q:=I-P_S$ implies $R_0Q=s^{-1}Q$, hence
$R=R_0-s^{-1}Q$ and $XRX=XR_0X-s^{-1}XQX$.  Also
$\opnorm{XQ}=\opnorm{(X-S)Q}\le\eta$ and $0\preceq QXQ\preceq\eta Q$.  On
$\supp X=\supp F$ one has $F^{-1}=(X+sI)X^{-1}$, so
$XF^+X=X^{2}+sX$ (both sides also vanish on $\ker X$) and
\[
 \opnorm{F^{+1/2}XQ}^{2}
 =\opnorm{QXF^+XQ}
 =\opnorm{Q(X^{2}+sX)Q}\le\eta^{2}+s\eta ,
\]
using $\opnorm{QX^{2}Q}=\opnorm{XQ}^{2}$.  Since
$F^{+1/2}F^{1/2}=P_X$ and $XP_X=X$, we may write
$QX=\bigl(QXF^{+1/2}\bigr)F^{1/2}$, and therefore
\[
 XQX=(QX)^{\dagger}(QX)
 =F^{1/2}\bigl(QXF^{+1/2}\bigr)^{\dagger}
        \bigl(QXF^{+1/2}\bigr)F^{1/2}
 \ \preceq\ \bigl(\eta^{2}+s\eta\bigr)F,
\]
that is,
\begin{equation}\label{eq:XQX}
 s^{-1}XQX\ \preceq\ (\eta+\eta^{2}/s)F .
\end{equation}
Subtracting \cref{eq:XQX} from \cref{eq:second-order-bound} and using
$F^{2}\preceq F$ and $\eta/s\le1/4$ gives \cref{eq:sandwich}, since both
coefficients are then at most $2\eta+\eta^{2}/s\le\tfrac94\eta\le3\eta$.

For \cref{eq:W-def}, note $XF^{+1/2}=(X+sI)F^{1/2}$ (both sides
vanish on $\ker X$ and agree on $\supp X$).  Hence
\[
 W=R(X+sI)F^{1/2}=\bigl[R(S+sI)+R(X-S)\bigr]F^{1/2}
  =\bigl[P_S+R(X-S)\bigr]F^{1/2},
\]
so $\opnorm W\le(1+\eta/s)\opnorm{F^{1/2}}\le\tfrac54\le2$ and
$\rank W\le\rank R\le\rank S$.  Finally
$F^{+1/2}F^{1/2}=P_X$ and $XP_X=X$, so $WF^{1/2}=RX$.
\end{proof}

\begin{proposition}[Error bounds for the two-batch estimator]\label{prop:block}
There is an absolute constant $K$ with the following property.  Let
{$0<s\le\min\{1,\sqrt{d_1/r}\}$}, $0<\delta<1$, and $\calE\in\qchannel_{d_1,d_2}^{r}$.
Fix $\sigma_{\mathrm{in}}\in\calD(\mathcal{H}_{\mathrm{A}})$ with $\sigma_{\mathrm{in}}\succeq I_{\mathrm{A}}/(2d_1)$, and let
$H=H_s(\sigma_{\mathrm{in}})$, where $H_s$ is defined in \cref{eq:gram}.  Set
\begin{equation}\label{eq:block-n}
 {b:=D+\left(1+\frac{d_1}{r}\right)\log\frac{6}{\delta},
 \qquad
 n=\left\lceil K\,\frac{b}{s^{2}}\right\rceil} .
\end{equation}
Let $\widehat Z,\widehat H$ be the outputs of the two-batch estimator
in \cref{def:block} with this batch size, and use the reconstruction rule $M\mapsto\Phi_{\sigma_{\mathrm{in}}}(M)$ from
\cref{eq:Lomega}.  With probability at least $1-\delta$,
\begin{align}
 \tfrac34\bigl(H+rI_{\mathrm{A}}\bigr)&\ \preceq\ \widehat H+rI_{\mathrm{A}}
   \ \preceq\ \tfrac54\bigl(H+rI_{\mathrm{A}}\bigr),
 \label{eq:block-multiplicative}\\[2pt]
 \dnorm{\Phi_{\sigma_{\mathrm{in}}}(\widehat Z)-\calE}
 &\ \le\ \Bigl(s+\tfrac{3s}{64}\Bigr)\opnorm H
   +4s\sqrt{q\,r\opnorm H}+q\,rs,
 \qquad q:=2^{-16}.
 \label{eq:block-channel}
\end{align}
In particular, if $\opnorm H\le10r$ then
$\dnorm{\Phi_{\sigma_{\mathrm{in}}}(\widehat Z)-\calE}\le11\,rs$.
\end{proposition}

\begin{proof}
Write $E=\widehat X_1-\calE(\kettbbra{\sqrt{\sigma_{\mathrm{in}}}}{\sqrt{\sigma_{\mathrm{in}}}})$, $\eta=s/64$, and let
$F=f_s(\calE(\kettbbra{\sqrt{\sigma_{\mathrm{in}}}}{\sqrt{\sigma_{\mathrm{in}}}}))$ and $Z=\calE(\kettbbra{\sqrt{\sigma_{\mathrm{in}}}}{\sqrt{\sigma_{\mathrm{in}}}})F$ be as in \cref{eq:ridge}.  

\smallskip\noindent\textbf{Step 1: the three events.}
{By \cref{lem:haar-moments}(iii) and
\cref{lem:rank-trunc},
\begin{equation}\label{eq:pilot-event}
 \opnorm{S-\calE(\kettbbra{\sqrt{\sigma_{\mathrm{in}}}}{\sqrt{\sigma_{\mathrm{in}}}})}\le\eta
\end{equation}
with probability at least $1-\delta/3$, provided
$n\ge c(D+\log(6/\delta))/s^2$, which follows from \cref{eq:block-n} for
$K$ large.}  Condition on the first batch and on \cref{eq:pilot-event}.  Then
$R$ and $W=R\calE(\kettbbra{\sqrt{\sigma_{\mathrm{in}}}}{\sqrt{\sigma_{\mathrm{in}}}})F^{+1/2}$ are fixed and independent of $E$.  For the fixed
unknown state $\calE(\kettbbra{\sqrt{\sigma_{\mathrm{in}}}}{\sqrt{\sigma_{\mathrm{in}}}})$, the matrix $W$ is determined by the first batch; it is
used only in the proof and need not be computed by the protocol.  By
\cref{lem:sandwich}, $\rank W\le r$, $\opnorm W\le2$,
$\rank R^{1/2}\le r$ and $\opnorm{R^{1/2}}\le s^{-1/2}$.  Apply
\cref{lem:weighted-conc} twice --- once to $W$ and once to $R^{1/2}$ ---
each with failure probability $\delta/3$.  Put
$a=rd_2+\log(12/\delta)$.  After an absolute adjustment of $K$,
\cref{eq:block-n} gives $n\ge Kd_1a/(rs^2)$, and hence
\[
 \sqrt{d_1}\left(\sqrt{\frac an}+\frac an\right)
 \le s\sqrt{\frac rK}+\frac{rs^2}{K\sqrt{d_1}}
 \le s\sqrt r\left(\frac1{\sqrt K}+\frac1K\right),
\]
where the last inequality uses $s\le\sqrt{d_1/r}$.

Taking $K$ sufficiently large in terms of $c_{\rm B}$ and $q$ makes the
last expression at most
$
\frac{s\sqrt{qr}}{c_{\rm B}}.
$
Since $\|W\|_\infty\le2$ and
$\|R^{1/2}\|_\infty\le s^{-1/2}$, \cref{eq:weighted-conc} yields
\begin{equation}\label{eq:two-quadratics}
 \opnorm{\Phi_{\sigma_{\mathrm{in}}}(EWW^{\dagger}E)^{\dagger}(I_{\mathrm{B}})}\le4qrs^{2},
 \qquad
 \opnorm{\Phi_{\sigma_{\mathrm{in}}}(ERE)^{\dagger}(I_{\mathrm{B}})}\le qrs .
\end{equation}
{A union bound gives total failure probability at most $\delta$.}

\smallskip\noindent\textbf{Step 2: the exact expansion.}
By \cref{lem:sandwich}, $R\calE(\kettbbra{\sqrt{\sigma_{\mathrm{in}}}}{\sqrt{\sigma_{\mathrm{in}}}})=WF^{1/2}$ and $\calE(\kettbbra{\sqrt{\sigma_{\mathrm{in}}}}{\sqrt{\sigma_{\mathrm{in}}}})R=F^{1/2}W^{\dagger}$, so
\begin{equation}\label{eq:expansion}
 \widehat Z-Z=\underbrace{(\calE(\kettbbra{\sqrt{\sigma_{\mathrm{in}}}}{\sqrt{\sigma_{\mathrm{in}}}})R\calE(\kettbbra{\sqrt{\sigma_{\mathrm{in}}}}{\sqrt{\sigma_{\mathrm{in}}}})-Z)}_{\text{approximation error}}
 +\underbrace{\bigl(EWF^{1/2}+F^{1/2}W^{\dagger}E\bigr)}_{\text{linear}}
 +\underbrace{ERE}_{\succeq0}.
\end{equation}
For every $a>0$, positivity of
$(a^{1/2}F^{1/2}\mp a^{-1/2}W^{\dagger}E)^{\dagger}(a^{1/2}F^{1/2}\mp a^{-1/2}W^{\dagger}E)$
gives
\begin{equation}\label{eq:cross-loewner}
 \pm\bigl(EWF^{1/2}+F^{1/2}W^{\dagger}E\bigr)\preceq aF+a^{-1}EWW^{\dagger}E .
\end{equation}

\smallskip\noindent\textbf{Step 3: the multiplicative estimate.}
Combining \cref{eq:sandwich,eq:expansion,eq:cross-loewner}
and $ERE\succeq0$,
\[
 -\bigl[(3\eta+a)F+a^{-1}EWW^{\dagger}E\bigr]
 \preceq\widehat Z-Z
 \preceq(3\eta+a)F+a^{-1}EWW^{\dagger}E+ERE .
\]
Apply the positive linear map in \cref{eq:momega}.  By
\cref{eq:normalization-adjoint,eq:gram},
$\Phi_{\sigma_{\mathrm{in}}}(F)^\dagger(I_{\mathrm{B}})=H$, and \cref{eq:block-output,lem:ridge-tail} give
$\widehat H-H=-\Phi_{\sigma_{\mathrm{in}}}(\widehat Z-Z)^\dagger(I_{\mathrm{B}})/s$.
With $a=s/16$ and $\eta=s/64$, so that $(3\eta+a)/s=7/64$, the two
estimates in \cref{eq:two-quadratics} give
\[
 -\Bigl[\tfrac7{64}H+65qrI_{\mathrm{A}}\Bigr]
 \ \preceq\ \widehat H-H\ \preceq\
 \Bigl[\tfrac7{64}H+64qrI_{\mathrm{A}}\Bigr].
\]
Since $7/64\le1/4$ and $65q\le1/4$, the right-hand sides are at most
$\tfrac14(H+rI_{\mathrm{A}})$ in the L\"owner order, so $ -\frac{1}{4}(H+rI_{\mathrm{A}})\ \preceq\ \widehat H-H\ \preceq\ \frac{1}{4}(H+rI_{\mathrm{A}})$. Hence, by adding $H+rI_{\mathrm{A}} $, we obtain $ \frac{3}{4}(H+rI_{\mathrm{A}})\ \preceq\ \widehat H+rI_{\mathrm{A}}\ \preceq\ \frac{5}{4}(H+rI_{\mathrm{A}})$, which is
\cref{eq:block-multiplicative}.

\smallskip\noindent\textbf{Step 4: the channel estimate.}
Write $\widehat Z-\calE(\kettbbra{\sqrt{\sigma_{\mathrm{in}}}}{\sqrt{\sigma_{\mathrm{in}}}})=(\widehat Z-Z)-sF$ and apply $M\mapsto\Phi_{\sigma_{\mathrm{in}}}(M)$.  By
\cref{lem:weighted-norms}, the approximation term  of \cref{eq:expansion} contributes at most
$3\eta\opnorm H$ (\cref{it:sandwich} with $G=F$); the linear term  of \cref{eq:expansion} 
contributes at most
$2\sqrt{\opnorm{\Phi_{\sigma_{\mathrm{in}}}(EWW^{\dagger}E)^{\dagger}(I_{\mathrm{B}})}\opnorm H}
\le4s\sqrt{qr\opnorm H}$ (\cref{it:cross} with $U=EW$, $V=F^{1/2}$); the
term $ERE\succeq0$  of \cref{eq:expansion} contributes exactly $\opnorm{\Phi_{\sigma_{\mathrm{in}}}(ERE)^{\dagger}(I_{\mathrm{B}})}\le qrs$ (\cref{it:pos}); and $-sF$ contributes exactly $s\opnorm H$
(\cref{lem:ridge-tail}).  Adding gives \cref{eq:block-channel}.  If
$\opnorm H\le10r$ then the right-hand side is at most
$10rs(1+\tfrac3{64})+4rs\sqrt{10q}+qrs\le11rs$.

\end{proof}

\subsubsection{Input selection}\label{sec:design-lemmas}
We first record the matrix inequalities used to choose the input marginal.
\begin{lemma}[Operator monotonicity and concavity; cf.\ {\cite[Ch.~V]{Bhatia97}}]
\label{lem:operator-monotone}
On positive definite matrices, $x\mapsto x^{-1}$ is operator convex and
operator decreasing, and $x\mapsto\log x$ and, for each $s>0$,
$x\mapsto x/(x+s)$ are operator monotone and operator concave.  
\end{lemma}

\begin{proof}
\emph{The inverse.}  For $M\succ0$ and $\ket{x}\in\Co^{D}$,
$\bra{x}M^{-1}\ket{x}=\sup_{\ket{z}}\bigl(2\operatorname{Re}\braket{x}{z}
-\bra{z}M\ket{z}\bigr)$, attained at $\ket{z}=M^{-1}\ket{x}$.  Each quadratic form of
$M^{-1}$ is thus a supremum of functions affine in $M$, hence convex, and
visibly nonincreasing in $M$; so $M\mapsto M^{-1}$ is operator convex and
operator decreasing.

\emph{The others.}  For $x>0$ and $s>0$,
\begin{equation*}
 \frac{x}{x+s}=1-\frac{s}{x+s},
 \qquad
 \log x=\int_0^{\infty}\Bigl(\frac1{1+t}-\frac1{x+t}\Bigr)dt ,
\end{equation*}
the second being the standard integral representation (both sides vanish at
$x=1$ and have derivative $1/x$, since
$\int_0^\infty(x+t)^{-2}dt=1/x$).  By the previous paragraph applied to
$M+tI$, each integrand $\tfrac1{1+t}I-(M+tI)^{-1}$ is operator concave and
operator monotone in $M$, and these properties are preserved by taking
nonnegative combinations and pointwise limits; the same applies to
$I-s(M+sI)^{-1}$.
\end{proof}

\begin{fact}[Golden--Thompson {\cite[Thm.~IX.3.7]{Bhatia97}}]
\label{fact:golden-thompson}
For Hermitian $M,N$ one has $\tr e^{M+N}\le\tr(e^{M}e^{N})$.
\end{fact}

\begin{lemma}[Trace bound for $A(\rho)$]\label{lem:trace-budget}
Fix $s>0$ and, as in \cref{eq:design-objects}, set
$\sigma_{\mathrm{in}}(\rho)=\rho/2+I_{\mathrm{A}}/(2d_1)$ and
$A(\rho)=H_s(\sigma_{\mathrm{in}}(\rho))+rI_{\mathrm{A}}$, where $H_s$ is defined in
\cref{eq:gram}.  Then every $\rho\in\calD(\mathcal{H}_{\mathrm{A}})$ satisfies
$\tr(\rho A(\rho))\le3r$.
\end{lemma}

\begin{proof}
Put $\nu=\sigma_{\mathrm{in}}(\rho)\succeq\rho/2$.  By \cref{eq:H-pairing} and positivity
of $\calE^c$, $\calE^c(\rho)\preceq2\calE^c(\nu)$, so
\[
 \tr\bigl(\rho H_s(\nu)\bigr)=\tr\bigl(\calE^c(\rho)(\calE^c(\nu)+sI_r)^{-1}\bigr)
 \le2\tr\bigl(\calE^c(\nu)(\calE^c(\nu)+sI_r)^{-1}\bigr)\le2r .
\]
Adding $\tr(\rho\cdot rI_{\mathrm{A}})=r$ finishes the proof.
\end{proof}

\begin{lemma}[Operator convexity of $Q\mapsto\log\Phi(Q^{-1})$]
\label{lem:opconvex}
Let $\Phi:\calL(\Co^{m})\to\calL(\mathcal{H}_{\mathrm{A}})$ be completely positive with
$\Phi(I_m)\succ0$.  Then the map $Q\mapsto\log\Phi(Q^{-1})$ is operator
convex on positive definite $Q$.
\end{lemma}

\begin{proof}
Take a Stinespring representation $\Phi(Y)=V^{\dagger}(I_k\otimes Y)V$ with
$V:\mathcal{H}_{\mathrm{A}}\to\Co^{k}\otimes\Co^{m}$.  The hypothesis $\Phi(I_m)=V^{\dagger}V\succ0$ makes
$V$ injective, so $V^{\dagger}$ is onto.  Set
$\Psi(Q):=\bigl[\Phi(Q^{-1})\bigr]^{-1}$.  For every $\ket{x}\in \mathcal{H}_{\mathrm{A}}$,
\begin{equation}\label{eq:variational}
 \bra{x}\Psi(Q)\ket{x}
 =\min\bigl\{\bra{y}(I_k\otimes Q)\ket{y}:\ V^{\dagger}\ket{y}=\ket{x}\bigr\}.
\end{equation}
Indeed, for positive definite $M$ the constrained minimum of
$\bra{y}M\ket{y}$ subject to $V^{\dagger}\ket{y}=\ket{x}$ equals
$\bra{x}(V^{\dagger}M^{-1}V)^{-1}\ket{x}$, attained at
$\ket{y}=M^{-1}V(V^{\dagger}M^{-1}V)^{-1}\ket{x}$; take $M=I_k\otimes Q$ and note
$V^{\dagger}(I_k\otimes Q^{-1})V=\Phi(Q^{-1})$.

By \cref{eq:variational} each quadratic form of $\Psi$ is an infimum of
functions affine in $Q$, hence concave; so $\Psi$ is operator concave.  Since
$\log$ is operator monotone and operator concave
(\cref{lem:operator-monotone}), $\log\circ\,\Psi$ is operator concave:
for $t\in[0,1]$,
\[
 \log\Psi\bigl(tQ_0+(1-t)Q_1\bigr)
 \succeq\log\bigl(t\Psi(Q_0)+(1-t)\Psi(Q_1)\bigr)
 \succeq t\log\Psi(Q_0)+(1-t)\log\Psi(Q_1).
\]
Finally $\log\Phi(Q^{-1})=-\log\Psi(Q)$, which is therefore operator convex.
\end{proof}

\begin{corollary}[Operator convexity of $\rho\mapsto\log A(\rho)$]\label{cor:design-convex}
Fix $s>0$ and let $A(\rho)=H_s(\rho/2+I_{\mathrm{A}}/(2d_1))+rI_{\mathrm{A}}$ be the
matrix-valued map defined in \cref{eq:design-objects}, with $H_s$ defined in
\cref{eq:gram}.  Then $\rho\mapsto\log A(\rho)$ is operator convex
on $\calD(\mathcal{H}_{\mathrm{A}})$.
\end{corollary}

\begin{proof}
On $\Co^{r}\oplus \mathcal{H}_{\mathrm{A}}$ define the affine, positive definite matrix
$Q(\rho)=\mathrm{diag}\bigl(\calE^c(\sigma_{\mathrm{in}}(\rho))+sI_r,\ I_{\mathrm{A}}\bigr)$ and the
completely positive map
$\Phi\bigl(\mathrm{diag}(Y_1,Y_2)\bigr)=(\calE^c)^{\dagger}(Y_1)+rY_2$, extended
to all of $\calL(\Co^{r}\oplus \mathcal{H}_{\mathrm{A}})$ by compressing to the two diagonal blocks
first.  Then $\Phi(I)=(1+r)I_{\mathrm{A}}\succ0$ and, by \cref{eq:resolvent-marginal},
$\Phi\bigl(Q(\rho)^{-1}\bigr)=H_s(\sigma_{\mathrm{in}}(\rho))+rI_{\mathrm{A}}=A(\rho)$.  Apply
\cref{lem:opconvex} and compose with the affine map $\rho\mapsto Q(\rho)$.
\end{proof}

\begin{proposition}[Input selection from multiplicative estimates]
\label{prop:gapping}
Fix $s>0$, and let $\sigma_{\mathrm{in}}(\rho)$ and $A(\rho)$ be the maps defined
in \cref{eq:design-objects}, formed from $H_s$ in \cref{eq:gram}.
Put $T:=\max\{1,\lceil\log_2d_1\rceil\}$ and $\rho_1:=I_{\mathrm{A}}/d_1$.
Suppose that
for each $j=1,\dots,T$ a positive definite $A_j$ is given with
\begin{equation}\label{eq:mult-hyp}
 \tfrac34\,A(\rho_j)\ \preceq\ A_j\ \preceq\ \tfrac54\,A(\rho_j),
\end{equation}
and define
\begin{equation}\label{eq:gapping-update}
 z_j=\tr\exp\bigl(\log\rho_j+\log A_j\bigr),
 \qquad
 \rho_{j+1}=\frac{\exp\bigl(\log\rho_j+\log A_j\bigr)}{z_j}.
\end{equation}
Then $\bar\rho:=\frac1T\sum_{j=1}^{T}\rho_j$ satisfies
\begin{equation}\label{eq:gapped}
 A(\bar\rho)\preceq10rI_{\mathrm{A}},
 \qquad\text{equivalently}\qquad
 H_s\bigl(\sigma_{\mathrm{in}}(\bar\rho)\bigr)\preceq9rI_{\mathrm{A}} .
\end{equation}
\end{proposition}

\begin{proof}
Each $\rho_j$ is a density operator, positive definite by induction.  By
Golden--Thompson (\cref{fact:golden-thompson}), \cref{eq:mult-hyp} and
\cref{lem:trace-budget},
\begin{equation*}
 z_j\le\tr\bigl(\rho_jA_j\bigr)\le\tfrac54\tr\bigl(\rho_jA(\rho_j)\bigr)
 \le\tfrac{15r}4 .
\end{equation*}
Taking logarithms in \cref{eq:gapping-update} is exact and gives
$\log\rho_{j+1}=\log\rho_j+\log A_j-(\log z_j)I_{\mathrm{A}}$; summing over
$j=1,\dots,T$,
\[
 \sum_{j=1}^{T}\log A_j
 =\log\rho_{T+1}-\log\rho_1+\Bigl(\sum_{j=1}^{T}\log z_j\Bigr)I_{\mathrm{A}}
 \preceq(\log d_1)I_{\mathrm{A}}+T\log\tfrac{15r}4\,I_{\mathrm{A}} ,
\]
using $\rho_{T+1}\preceq I_{\mathrm{A}}$, hence $\log\rho_{T+1}\preceq0$, and
$-\log\rho_1=(\log d_1)I_{\mathrm{A}}$.  By operator monotonicity of the logarithm (\cref{lem:operator-monotone})  and
the lower bound in \cref{eq:mult-hyp},
$\log A(\rho_j)\preceq\log A_j+\log\tfrac43\,I_{\mathrm{A}}$.  Therefore
\[
 \frac1T\sum_{j=1}^{T}\log A(\rho_j)
 \preceq\Bigl(\frac{\log d_1}{T}+\log(5r)\Bigr)I_{\mathrm{A}} ,
\]
because $\tfrac{15}4\cdot\tfrac43=5$.  \cref{cor:design-convex}
bounds $\log A(\bar\rho)$ by the left-hand side, so the largest eigenvalue of
$\log A(\bar\rho)$ is at most $\log(5r)+(\log d_1)/T$; exponentiating this
\emph{scalar} inequality gives
$A(\bar\rho)\preceq5r\,d_1^{1/T}I_{\mathrm{A}}\preceq10rI_{\mathrm{A}}$, since
$d_1^{1/T}\le2$ by the choice of $T$.  Subtracting  $rI_{\mathrm{A}}$ gives \cref{eq:gapped}.
\end{proof}

\section{Lower bound}\label{sec:lower}
In this section we prove the lower bound   of \cref{thm:lb-main} which we restate here as \cref{thm-9151403}.

\begin{theorem}\label{thm-9151403}
There are absolute constants $c,\epsilon_0>0$ such that the following holds.
Assume $d_2\ge 2$, $1\le r\le D$, $d_1\le rd_2$, $Dr\ge12$, and $0<\epsilon\le\epsilon_0$.
Every adaptive incoherent protocol that learns every channel in $\qchannel_{d_1,d_2}^r$ to diamond-norm error $\epsilon$ with probability at least $2/3$
uses at least
\begin{equation*}
 c\frac{Dr^2}{\epsilon^2}
\end{equation*}
queries.  The hard instance family used to prove the lower bound can be chosen so that every non-zero eigenvalue of its Choi operator belongs to $[d_1/(4r),4d_1/r]$.  The conclusion therefore also holds for learning channels with $\Omega(d_1/r)$-gapped Choi spectrum.
\end{theorem}

For the proof, we construct a local family of exactly trace-preserving
channels with $\Theta(Dr)$ real parameters. We then show that the Fisher
information matrix of a single incoherent query has trace at most $16D/r$.
A Gaussian prior with variance $\Theta(t^2/D)$ in each coordinate,
conditioned on the parameter domain, yields the mutual-information bound
$O(nt^2/r)$ between the channel parameter and the algorithm transcript.
Finally, learning to error $\Theta(t)$ requires $\Omega(Dr)$ nats of
information, which gives the desired lower bound.

\subsection{Main proof}\label{sec:lb-proof}
\paragraph{The local family.}\label{sec:lb-family}
For a Kraus list $\{K_1, \dots, K_r\}$, we write its stack and scaled Choi factor as
\begin{equation}\label{eq:stack}
 \mathbf K=\begin{pmatrix}K_1\\\vdots\\K_r\end{pmatrix},\qquad
 U=\sqrt{\frac r{d_1}}
     [\,\kett{K_1}\ \cdots\ \kett{K_r}\,],\qquad
 C_\calE=\frac{d_1}{r}UU^\dagger.
\end{equation}
The reshaping map
\begin{equation}\label{eq:reshuffle}
 \sS(M)=\begin{pmatrix}\operatorname{mat}(\ket{m_1})\\\vdots\\
                         \operatorname{mat}(\ket{m_r})\end{pmatrix},
 \qquad M=[\ket{m_1}\ \cdots\ \ket{m_r}],
\end{equation}
is a Hilbert--Schmidt isometry and satisfies
$\sS(U)=\sqrt{r/d_1}\,\mathbf K$.
By \cref{lem:anchor}, choose reference Kraus operators with stack
$\mathbf K_0$ and scaled Choi factor $U_0$ satisfying
\begin{equation}\label{eq:anchor}
 \mathbf K_0^\dagger\mathbf K_0=I_{\mathrm{A}},\qquad
 \tfrac12 I_r\preceq U_0^\dagger U_0\preceq2I_r.
\end{equation}
For $\Theta\in\Co^{D\times r}$, put
\begin{equation}\label{eq:Ttheta}
 T_\Theta=\sqrt{d_1/r}\,\sS(\Theta),\qquad
 \|T_\Theta\|_2=\sqrt{d_1/r}\,\|\Theta\|_2,
\end{equation}
and consider the real linear space
\begin{equation}\label{eq:Esubspace}
 \calV=\{\Theta:
 U_0^\dagger\Theta=\Theta^\dagger U_0,\quad
 \mathbf K_0^\dagger T_\Theta+T_\Theta^\dagger\mathbf K_0=0\}.
\end{equation}
Unitary mixing of the Kraus operators changes $U_0$ to $U_0V$, for a
unitary $V$, without changing $U_0U_0^\dagger$.  The first constraint makes
$\Theta$ orthogonal, in the real Hilbert--Schmidt inner product, to all
directions $U_0W$ with $W^\dagger=-W$.  The second constraint is the
linearized trace-preservation condition.  Define
\begin{equation}\label{eq:chart}
 Q_\Theta=(I_{\mathrm{A}}+T_\Theta^\dagger T_\Theta)^{-1/2},\qquad
 \mathbf K_\Theta=(\mathbf K_0+T_\Theta)Q_\Theta.
\end{equation}
For every $\Theta\in\calV$,
\begin{equation*}
 (\mathbf K_0+T_\Theta)^\dagger(\mathbf K_0+T_\Theta)
 =I_{\mathrm{A}}+T_\Theta^\dagger T_\Theta,
 \qquad \mathbf K_\Theta^\dagger\mathbf K_\Theta=I_{\mathrm{A}}.
\end{equation*}
Thus $\mathbf K_\Theta$ defines a channel $\calE_\Theta$.  Its scaled factor
and Choi operator are
\begin{equation}\label{eq:flagged}
 U_\Theta=(I_{\mathrm{B}}\otimes Q_\Theta^{\mathrm T})(U_0+\Theta),\qquad
 C_\Theta:=C_{\calE_\Theta}=\frac{d_1}{r}U_\Theta U_\Theta^\dagger.
\end{equation}
For $0<t\le t_0:=2^{-5}$ restrict to the compact convex set
\begin{equation}\label{eq:Kt}
 \calK_t=\{\Theta\in\calV:\|\Theta\|_\infty\le t,
                                  \ \|T_\Theta\|_\infty\le t\}.
\end{equation}
Choose an orthonormal basis of the real Hilbert space $\calV$, with
inner product $\langle M,N\rangle_\R=\Re\tr(M^\dagger N)$, and denote
the coordinates of $\Theta$ by $\theta\in\R^k$, where
$k=\dim_\R\calV$.  
We identify $\calK_t$ with its coordinate image in $\R^k$. Set
\begin{equation}\label{eq:sigma}
\sigma:=\frac{t}{32\sqrt D},
\qquad
\gamma_\sigma:=\mathcal N(0,\sigma^2 I_k).
\end{equation}
For every Borel set $A\subseteq\R^k$, define
\begin{equation}\label{eq:prior}
\mu_t(A)
:=
\Pr_{\theta\sim\gamma_\sigma}
\left\{
\theta\in A
\,\middle|\,
\theta\in\calK_t
\right\}.
\end{equation}

The dimension, spectral bounds, and prior estimates needed below are
proved in \cref{lem:dim,lem:spectral-control,lem:prior}.

\begin{proof}[Proof of \cref{thm-9151403}]
Take $M=2^{21}$, $\epsilon_0=2^{-26}$, and $t=M\epsilon\le t_0$.
Run the protocol on $\calE_\Theta$ with $\theta\sim\mu_t$ and let $Z$
be its classical transcript.  By \cref{lem:spectral-control}, every
channel in the family has rank exactly $r$ and non-zero Choi eigenvalues
in $[d_1/(4r),4d_1/r]$.

\textbf{Information available from the transcript.}
\cref{lem:fisher} bounds the trace of the one-query Fisher information matrix by
$16D/r$, uniformly over the input, measurement, and parameter.
Conditional scores have mean zero, so this bound adds over adaptive
rounds.  The log-Sobolev inequality for $\mu_t$ then gives
\begin{equation}\label{eq:main-mi}
 I(\theta;Z)\le\frac{nt^2}{128r}
\end{equation}
by \cref{prop:mi}.

\textbf{Information required for reconstruction.}
On the success event, \cref{eq:choi-diamond-comparison} implies
\[
 \frac1{d_1}\|C_\Theta-C_{\widehat\calE(Z)}\|_1\le\epsilon.
\]
By \cref{lem:smallball}, every such ball has prior mass at most
$2e^{-Dr/2}$.  Fano's inequality for reconstruction sets
(\cref{lem:fano}) and success probability $2/3$ yield
\[
 I(\theta;Z)\ge\frac23\left(\frac{Dr}{2}-\log2\right)-\log2
 \ge\frac{Dr}{6},
\]
where the last inequality holds for $Dr\ge12$.
Comparing with \cref{eq:main-mi} gives
\[
 n\ge\frac{128}{6}\frac{Dr^2}{M^2\epsilon^2}
 \ge\frac{21}{M^2}\frac{Dr^2}{\epsilon^2}.
\]
This proves the theorem with $c=21/M^2$.  Since the family has rank exactly
$r$, it is also contained in the rank-at-most-$r$ class.
\end{proof}

\subsection{Technical lemmas}\label{sec:lb-technical}
We establish the geometric estimates for the parameterized family, the
mutual-information bound for adaptive transcripts, and the bound on the
prior mass of trace-norm balls used in \cref{sec:lb-proof}.

\subsubsection{Geometry of the lower-bound family}\label{sec:lb-geometry}
\begin{lemma}[Reference Kraus operators]\label{lem:anchor}
Assume $d_2\ge2, 1\le r\le D=d_1d_2$ and $ d_1\le rd_2$.  There are Kraus operators
$K_1,\dots,K_r\in\Co^{d_2\times d_1}$, with stack $\mathbf K_0$ and Choi factor
$U_0$ as in \cref{eq:stack}, such that
\begin{equation*}
 \mathbf K_0^{\dagger}\mathbf K_0=I_{\mathrm{A}},
 \qquad
 \tfrac12I_r\ \preceq\ U_0^{\dagger}U_0\ \preceq\ 2\,I_r .
\end{equation*}
The $K_i$ may be taken Hilbert--Schmidt orthogonal, so that $U_0^{\dagger}U_0$ is
diagonal.
\end{lemma}

\begin{proof}
Note $(U_0^{\dagger}U_0)_{ij}=(r/d_1)\tr(K_i^{\dagger}K_j)$.

\emph{Case $r\le d_1$.}  Partition the input basis into $r$ nonempty blocks of
sizes $a_i\in\{\lfloor d_1/r\rfloor,\lceil d_1/r\rceil\}$ summing to $d_1$.
Since $d_1\le rd_2$ we have $d_1/r\le d_2$, and as $d_2$ is an integer,
$\lceil d_1/r\rceil\le d_2$; so each block admits an isometry into $\mathcal{H}_{\mathrm{B}}$.  Let
$K_i$ be such an isometry on the $i$-th block and zero elsewhere.  Then
$\sum_iK_i^{\dagger}K_i=I_{\mathrm{A}}$, the $K_i$ have orthogonal supports, and
$\tr(K_i^{\dagger}K_j)=a_i\delta_{ij}$.  Hence $U_0^{\dagger}U_0$ is diagonal with entries
$ra_i/d_1$.  For $x=d_1/r\ge1$ one has $\lfloor x\rfloor\ge x/2$ and
$\lceil x\rceil\le2x$, which gives \cref{eq:anchor}.

\emph{Case $r\ge d_1$.}  Choose integers
$k_x\in\{\lfloor r/d_1\rfloor,\lceil r/d_1\rceil\}$, $1\le x\le d_1$, with
$\sum_xk_x=r$.  Since $r\le d_1d_2$ we have $\lceil r/d_1\rceil\le d_2$, so
for each $x$ we may pick $k_x$ orthonormal vectors $\ket{v_{x,1}},\dots,
\ket{v_{x,k_x}}$ in $\mathcal{H}_{\mathrm{B}}$.  Put $K_{x,j}=k_x^{-1/2}\ket{v_{x,j}}\bra{x}$.
There are exactly $r$ of them, $\sum_{x,j}K_{x,j}^{\dagger}K_{x,j}=I_{\mathrm{A}}$, and they
are Hilbert--Schmidt orthogonal: different $x$ have orthogonal input supports,
equal $x$ have orthogonal output vectors.  Their squared norms are $1/k_x$, so
$U_0^{\dagger}U_0$ is diagonal with entries $r/(d_1k_x)$, and the same
floor--ceiling estimates applied to $y=r/d_1\ge1$ give \cref{eq:anchor}.
The two cases agree when $r=d_1$.
\end{proof}

\begin{lemma}[Dimension of the parameter space]\label{lem:dim}
Assume \cref{eq:feasible}.  Let $\calV$ be the real linear space
defined in \cref{eq:Esubspace} from reference factors satisfying
\cref{eq:anchor}, and set $k=\dim_\R\calV$.  Then
\begin{equation}\label{eq:dim}
 \tfrac34\,Dr\ \le\ \Bigl(1-\tfrac1{d_2^{2}}\Bigr)Dr
 \ \le\ 2Dr-r^{2}-d_1^{2}\ \le\ k\ \le\ 2Dr .
\end{equation}
\end{lemma}

\begin{proof}
The ambient real dimension is $2Dr$ so $k\le 2Dr$.  The Hermiticity constraint says that the
skew-Hermitian part of $U_0^{\dagger}\Theta$ vanishes, which is $r^{2}$ real linear
conditions; the trace-preservation constraint takes values in the Hermitian
$d_1\times d_1$ matrices, which is $d_1^{2}$ real linear conditions.  By
rank--nullity $k\ge2Dr-r^{2}-d_1^{2}$; independence of the constraints is not
needed.

By \cref{eq:feasible} we have $d_1/d_2\le r\le d_1d_2$, so
$(r-d_1/d_2)(r-d_1d_2)\le0$.  Expanding,
$r^{2}+d_1^{2}\le rd_1d_2+rd_1/d_2=Dr(1+d_2^{-2})$.  Since $d_2\ge2$ this is
at most $\tfrac54Dr$, and \cref{eq:dim} follows.
\end{proof}

\begin{lemma}[Inverse square root]\label{lem:invsqrt}
Let $X,Y\succeq I$ be positive definite.  Then, in the operator norm and in
the Hilbert--Schmidt norm, we have
\begin{equation*}
 \bigl\|X^{-1/2}-Y^{-1/2}\bigr\|\ \le\ \tfrac12\|X-Y\| .
\end{equation*}
Consequently the Fr\'echet derivative of $X\mapsto X^{-1/2}$ at any
$X\succeq I$ has norm at most $\tfrac12$ in both norms.
\end{lemma}

\begin{proof}
Use $X^{-1/2}=\frac1\pi\int_0^\infty s^{-1/2}(X+sI)^{-1}\,ds$ and the resolvent
identity $(X+sI)^{-1}-(Y+sI)^{-1}=(X+sI)^{-1}(Y-X)(Y+sI)^{-1}$.  For
$X,Y\succeq I$ both resolvents have operator norm at most $(1+s)^{-1}$, so the
integrand is bounded in either norm by $s^{-1/2}(1+s)^{-2}\|X-Y\|$; here we
use that a two-sided multiplication is bounded on the Hilbert--Schmidt norm by
the product of the operator norms of its factors.  Since
$\int_0^\infty s^{-1/2}(1+s)^{-2}ds=\pi/2$, the claim
follows.  The derivative statement follows by taking limit. 
\end{proof}

\begin{lemma}[Remainder estimates]\label{lem:remainder}
Let $0<t\le t_0=2^{-5}$ and $\Theta,\Xi\in\calK_t$, where
$\calK_t\subseteq\calV$ is the parameter domain defined in
\cref{eq:Esubspace,eq:Kt}.  Let $Q_\Theta$ and $U_\Theta$ be given by
\cref{eq:chart,eq:flagged}, and define
\begin{equation}\label{eq:remainder-def}
 U_\Theta=U_0+\mathbf S_\Theta,
 \qquad
 \mathbf S_\Theta:=\Theta+R_\Theta,
 \qquad
 R_\Theta:=\bigl(I_{\mathrm{B}}\otimes (Q_\Theta-I_{\mathrm{A}})^{\mathrm T}\bigr)(U_0+\Theta).
\end{equation}
Put $\Delta=\Theta-\Xi$.  For every direction $H\in\calV$, with
$DU_\Theta[H]$ denoting the Fr\'echet derivative of
$\Theta\mapsto U_\Theta$ in direction $H$, the following estimates hold:
\begin{align}
 \opnorm{Q_\Theta-I_{\mathrm{A}}}\le\tfrac12t^{2},
 \qquad
 \opnorm{R_\Theta}\le t^{2},
 \qquad
 \opnorm{\mathbf S_\Theta}\le\tfrac{17}{16}t,
 \label{eq:rem-op}\\[2pt]
 \hsnorm{R_\Theta-R_\Xi}\le\tfrac54t\,\hsnorm\Delta,
 \qquad
 \hsnorm{\mathbf S_\Theta-\mathbf S_\Xi}\le\tfrac98\hsnorm\Delta,
 \qquad
 \hsnorm{DU_\Theta[H]}\le2\hsnorm H .
 \label{eq:pullback-norm}
\end{align}
\end{lemma}

\begin{proof}
Since $\opnorm{T_\Theta}\le t$ we have
$I_{\mathrm{A}}\preceq I_{\mathrm{A}}+T_\Theta^{\dagger}T_\Theta\preceq(1+t^{2})I_{\mathrm{A}}$, so
\cref{lem:invsqrt} with $Y=I$ gives
$\opnorm{Q_\Theta-I}\le\tfrac12\opnorm{T_\Theta^{\dagger}T_\Theta}\le\tfrac12t^{2}$.
As $\opnorm{U_0}\le\sqrt2$ by \cref{eq:anchor} and $\opnorm\Theta\le t$,
\[
 \opnorm{R_\Theta}\le\tfrac12t^{2}(\sqrt2+t)\le t^{2},
 \qquad
 \opnorm{\mathbf S_\Theta}\le t+t^{2}\le\tfrac{17}{16}t ,
\]
using $t\le2^{-5}$.  This is \cref{eq:rem-op}.

For the Hilbert--Schmidt estimates, use the reshaping map $\sS$ from
\cref{eq:reshuffle}.  Under this map, $M\mapsto(I_{\mathrm{B}}\otimes N^{\mathrm T})M$ is
$\sS(M)\mapsto\sS(M)N$, so for any $N\in\Co^{d_1\times d_1}$,
\begin{equation}\label{eq:cancel}
 \hsnorm{(I_{\mathrm{B}}\otimes N^{\mathrm T})(U_0+\Theta)}
 =\hsnorm{\sS(U_0+\Theta)N}
 \le\opnorm{\sS(U_0+\Theta)}\hsnorm N
 \le\sqrt{\tfrac r{d_1}}\,(1+t)\,\hsnorm N,
\end{equation}
because $\sS(U_0+\Theta)=\sqrt{r/d_1}(\mathbf K_0+T_\Theta)$ with
$\opnorm{\mathbf K_0}=1$ and $\opnorm{T_\Theta}\le t$.  Combined with
$\hsnorm{T_\Delta}=\sqrt{d_1/r}\,\hsnorm\Delta$ from \cref{eq:Ttheta}, the
factor $\sqrt{d_1/r}$ cancels, giving a bound independent of the
dimensions.

Now decompose
\[
 R_\Theta-R_\Xi
 =\bigl(I_{\mathrm{B}}\otimes (Q_\Theta-Q_\Xi)^{\mathrm T}\bigr)(U_0+\Theta)
 +\bigl(I_{\mathrm{B}}\otimes (Q_\Xi-I)^{\mathrm T}\bigr)\Delta .
\]
By \cref{lem:invsqrt},
$\hsnorm{Q_\Theta-Q_\Xi}\le\tfrac12\hsnorm{T_\Theta^{\dagger}T_\Theta-T_\Xi^{\dagger}T_\Xi}
\le\tfrac12(\opnorm{T_\Theta}+\opnorm{T_\Xi})\hsnorm{T_\Delta}
\le t\sqrt{d_1/r}\,\hsnorm\Delta$.  By \cref{eq:cancel} the first term has
Hilbert--Schmidt norm at most $t(1+t)\hsnorm\Delta$, and the second at most
$\tfrac12t^{2}\hsnorm\Delta$.  Hence
$\hsnorm{R_\Theta-R_\Xi}\le(t+t^{2}+\tfrac12t^{2})\hsnorm\Delta
\le\tfrac54t\hsnorm\Delta$ and
$\hsnorm{\mathbf S_\Theta-\mathbf S_\Xi}\le(1+\tfrac54t)\hsnorm\Delta\le\tfrac98\hsnorm\Delta$.

For the derivative, differentiate \cref{eq:remainder-def} in a direction
$H\in\calV$:
\[
 DU_\Theta[H]=H+\bigl(I_{\mathrm{B}}\otimes (DQ_\Theta[H])^{\mathrm T}\bigr)(U_0+\Theta)
 +\bigl(I_{\mathrm{B}}\otimes (Q_\Theta-I)^{\mathrm T}\bigr)H .
\]
\cref{lem:invsqrt} gives
$\hsnorm{DQ_\Theta[H]}\le\tfrac12\hsnorm{T_H^{\dagger}T_\Theta+T_\Theta^{\dagger}T_H}
\le\opnorm{T_\Theta}\hsnorm{T_H}\le t\sqrt{d_1/r}\hsnorm H$, so by
\cref{eq:cancel} the middle term is at most $t(1+t)\hsnorm H$ and the last is
at most $\tfrac12t^{2}\hsnorm H$.  Therefore
$\hsnorm{DU_\Theta[H]}\le(1+\tfrac54t)\hsnorm H\le2\hsnorm H$.
\end{proof}

\begin{lemma}[Spectral and metric control]\label{lem:spectral-control}
Let $0<t\le t_0=2^{-5}$ and $\Theta,\Xi\in\calK_t$, with
$\calV$ and $\calK_t$ as in \cref{eq:Esubspace,eq:Kt}.  Use the
scaled Choi factor $U_\Theta$ and Choi operator
$C_\Theta=C_{\calE_\Theta}$ defined in \cref{eq:flagged}; derivatives
are taken along the real space $\calV$.  Then
\begin{enumerate}[label=\rm(\roman*)]
\item $C_\Theta$ has rank $r$ and its non-zero eigenvalues lie in
      $[d_1/(4r),4d_1/r]$.
\item $\|DU_\Theta[H]\|_2\le2\|H\|_2$ for all $H\in\calV$.
\item $d_1^{-1}\|C_\Theta-C_\Xi\|_2\ge(2r)^{-1}\|\Theta-\Xi\|_2$.
\item $d_1^{-1}\|C_\Theta-C_\Xi\|_\infty\le16t/r$.
\item $\|\Theta-\Xi\|_2^2\le4rt^2$.
\end{enumerate}
\end{lemma}

\begin{proof}
(i) By \cref{eq:anchor}, $\sigma_{\min}(U_0)\ge2^{-1/2}$ and
$\opnorm{U_0}\le\sqrt2$.  The matrix $I_{\mathrm{B}}\otimes Q_\Theta^{\mathrm T}$ is
positive with the spectrum of $Q_\Theta$, hence with eigenvalues in
$[(1+t^{2})^{-1/2},1]$.  Therefore, by \cref{eq:flagged},
\[
 \sigma_{\min}(U_\Theta)\ \ge\ \frac{2^{-1/2}-t}{\sqrt{1+t^{2}}}\ >\ \tfrac12,
 \qquad
 \opnorm{U_\Theta}\ \le\ \sqrt2+t\ <\ 2 ,
\]
for $t\le2^{-5}$.  The non-zero eigenvalues of
$(C_\Theta/d_1)=U_\Theta U_\Theta^{\dagger}/r$ are those of $U_\Theta^{\dagger}U_\Theta/r$, so
they lie in $[1/(4r),4/r]$, and there are exactly $r$ of them.

(ii) This is the derivative estimate in \cref{eq:pullback-norm}.

(iii) Put $\Delta=\Theta-\Xi\in\calV$ and
$L(\Delta)=U_0\Delta^{\dagger}+\Delta U_0^{\dagger}$.  Expanding and using cyclicity,
\[
 \hsnorm{L(\Delta)}^{2}
 =2\tr\bigl(\Delta^{\dagger}\Delta\,U_0^{\dagger}U_0\bigr)
 +2\tr\bigl((U_0^{\dagger}\Delta)^{2}\bigr) .
\]
The Hermiticity constraint in \cref{eq:Esubspace} says that $U_0^{\dagger}\Delta$ is
Hermitian, so the second trace is nonnegative; and
$U_0^{\dagger}U_0\succeq\tfrac12I$ by \cref{eq:anchor}, so the first is at least
$\hsnorm\Delta^{2}$.  Hence
\begin{equation}\label{eq:injectivity}
 \hsnorm{L(\Delta)}\ \ge\ \hsnorm\Delta .
\end{equation}
The Hermiticity constraint excludes the $r^2$-dimensional space of
perturbations $\Delta=U_0W$ with $W^\dagger=-W$.  Without the constraint,
these directions, arising from unitary mixing of the Kraus operators,
would satisfy $L(\Delta)=0$.

Writing $U_\Theta=U_0+\mathbf S_\Theta$ and
$\mathbf S_\Theta-\mathbf S_\Xi=\Delta+(R_\Theta-R_\Xi)$,
\[
 r\bigl((C_\Theta-C_\Xi)/d_1\bigr)-L(\Delta)
 =U_0(R_\Theta-R_\Xi)^{\dagger}+(R_\Theta-R_\Xi)U_0^{\dagger}
 +\mathbf S_\Theta \mathbf S_\Theta^{\dagger}-\mathbf S_\Xi \mathbf S_\Xi^{\dagger} .
\]
By \cref{lem:remainder} its Hilbert--Schmidt norm is at most
\[
 2\sqrt2\cdot\tfrac54t\hsnorm\Delta
 +\bigl(\opnorm{\mathbf S_\Theta}+\opnorm{\mathbf S_\Xi}\bigr)\hsnorm{\mathbf S_\Theta-\mathbf S_\Xi}
 \le\Bigl(\tfrac52\sqrt2+\tfrac{17}{8}\cdot\tfrac98\Bigr)t\hsnorm\Delta
 \le6t\hsnorm\Delta .
\]
With \cref{eq:injectivity} and $6t\le\tfrac3{16}$ this gives
$r\hsnorm{(C_\Theta-C_\Xi)/d_1}\ge(1-6t)\hsnorm\Delta\ge\tfrac12\hsnorm\Delta$.

(iv)  $\opnorm{U_\Theta-U_0}=\opnorm{\mathbf S_\Theta}\le\tfrac{17}{16}t$ and
$\opnorm{U_\Theta}+\opnorm{U_0}\le2+\sqrt2$, so
$\opnorm{(C_\Theta-C_0)/d_1}\le r^{-1}\opnorm{U_\Theta-U_0}
 \bigl(\opnorm{U_\Theta}+\opnorm{U_0}\bigr)\le8t/r$.  The triangle inequality
through $(C_0/d_1)$ gives (iv).

(v) $\Theta$ has rank at most $r$, so
$\hsnorm\Theta\le\sqrt r\opnorm\Theta\le\sqrt rt$, and likewise for $\Xi$.
\end{proof}

\subsubsection{Information bounds and prior mass estimates}\label{sec:lb-fisher}
For a fixed single-query experiment, let $\{T_y\}_{y\in\calY}$ be its
tester, as in \cref{eq:tester}.
Since
\begin{equation}
\tr\mathsf T(\calY)
=
\tr\!\left(
I_{\mathrm B}\otimes
\tr_{\mathrm R}(\ketbra{\psi}{\psi})^{\mathrm T}
\right)
=d_2,
\end{equation}
the formula
\begin{equation}\label{eq:tester-control}
\nu(B):=\frac{1}{d_2}\tr\mathsf T(B),
\qquad
B\in\mathcal B(\calY),
\end{equation}
defines a probability measure on $\calY$.

The finite-dimensional Radon--Nikodym theorem gives a measurable operator
density $T_y\succeq0$ such that
\begin{equation}\label{eq:tester-density}
\mathsf T(B)=\int_B T_y\,d\nu(y).
\end{equation}
Moreover,
\begin{equation}
\tr(T_y)=d_2
\qquad
\text{for $\nu$-almost every $y$}.
\end{equation}
Recall that 
$\theta\in\R^k$ are  the coordinates of $\Theta$ in an orthonormal basis of
the real Hilbert space $\calV$.
Then, the outcome law under the channel $\calE_\Theta$ has density
\begin{equation}\label{eq:qtheta}
q_\theta(y)
=
\tr(T_yC_\Theta)
=
\frac{d_1}{r}\tr(T_yU_\Theta U_\Theta^\dagger)
\end{equation}
with respect to $\nu$.

The score is $\nabla_\theta\log q_\theta(y)$ when $q_\theta(y)>0$ and is
defined to be zero otherwise. Its Fisher information matrix is
\begin{equation}\label{eq:fisher-definition}
I_1(\theta)
=
\int_{\{y:q_\theta(y)>0\}}
\frac{
\nabla_\theta q_\theta(y)
\nabla_\theta q_\theta(y)^{\mathrm T}
}{
q_\theta(y)
}
\,d\nu(y).
\end{equation}
\begin{lemma}[One-query Fisher information]\label{lem:fisher}
Let $0<t\le t_0=2^{-5}$ and $\Theta\in\calK_t$.  
The Fisher information matrix defined in \cref{eq:fisher-definition} satisfies
\begin{equation}\label{eq:fisher-one}
\tr I_1(\theta)\le\frac{16D}{r}.
\end{equation}
Moreover, the score has mean zero.
\end{lemma}

\begin{proof}
For a fixed $G\succeq0$, define
\begin{equation}
p(U):=\frac{d_1}{r}\tr(GUU^\dagger).
\end{equation}
On the real Hilbert space of complex matrices with inner product
$\Re\tr(X^\dagger Y)$, one has
\begin{equation}
\nabla_Up
=
2\frac{d_1}{r}GU
\end{equation}
and hence
\begin{equation}
\begin{aligned}
\|\nabla_Up\|_2^2
&=
4\left(\frac{d_1}{r}\right)^2
\tr(U^\dagger G^2U)\\
&\le
4\frac{d_1}{r}\tr(G)\,p(U).
\end{aligned}
\end{equation}
The derivative of $\theta\mapsto U_\Theta$ has operator norm at most $2$
by \cref{lem:spectral-control}. Taking $G=T_y$ therefore gives
\begin{equation}\label{eq:score-gradient-bound}
\|\nabla_\theta q_\theta(y)\|_2^2
\le
16\frac{d_1}{r}\tr(T_y)q_\theta(y).
\end{equation}
Dividing by $q_\theta(y)$ on its positive set and integrating with respect
to $\nu$ yields
\begin{equation}
\begin{aligned}
\tr I_1(\theta)
&\le
16\frac{d_1}{r}
\int_{\calY}\tr(T_y)\,d\nu(y)\\
&=
16\frac{d_1d_2}{r}
=
\frac{16D}{r},
\end{aligned}
\end{equation}
where we used \cref{eq:tester-density}. If $q_\theta(y)=0$, then
\cref{eq:score-gradient-bound} implies
\begin{equation}
\nabla_\theta q_\theta(y)=0.
\end{equation}

It remains to prove that the score is centered. For every coordinate
$\theta_a$,
\begin{equation}
|\partial_a q_\theta(y)|
=
|\tr(T_y\partial_a C_\Theta)|
\le
\tr(T_y)\|\partial_a C_\Theta\|_\infty.
\end{equation}
The derivative of $C_\Theta$ is locally bounded, while
$\tr(T_y)=d_2$ for $\nu$-almost every $y$. Differentiation under the
integral is therefore justified. Since
\begin{equation}
\int_{\calY}q_\theta(y)\,d\nu(y)=1,
\end{equation}
we obtain
\begin{equation}
\begin{aligned}
\E_\theta\!\left[
\nabla_\theta\log q_\theta(Y)
\right]
&=
\int_{\calY}
\nabla_\theta q_\theta(y)\,d\nu(y)\\
&=
\nabla_\theta
\int_{\calY}q_\theta(y)\,d\nu(y)
=
0.
\end{aligned}
\end{equation}
\end{proof}

\paragraph{Adaptive transcripts.}\label{sec:lb-mi}
Let $
Z=(W,Y_1,\ldots,Y_n)$
be the transcript, where $W$ is the private random seed, sampled
independently of $\theta$. For each fixed seed $W=w$ and classical history,
apply \cref{eq:tester-control} to the tester selected at that history.  Since all
outcome spaces are standard Borel, the corresponding Radon--Nikodym
densities may be chosen jointly measurable in the history and current
outcome. The resulting parameter-independent control kernels induce a
parameter-independent reference measure on the transcript space. Relative
to this measure, the transcript likelihood is the product of the
conditional likelihoods.

The transcript score is therefore the sum of the conditional scores.
Since each conditional score has conditional mean zero, the cross terms
between distinct rounds vanish. Hence
\begin{equation}\label{eq:adaptive-fisher}
I_Z(\theta)
=
\sum_{j=1}^n
\E_\theta\!\left[
I_{Y_j\mid W,Y_{<j}}(\theta)
\right],
\qquad
\tr I_Z(\theta)\le\frac{16nD}{r}.
\end{equation}
The expectation in \cref{eq:adaptive-fisher} is over $W$ and the previous
outcomes. The bound is uniform over all seeds and histories.

\begin{lemma}[Operator norm of a Gaussian matrix]\label{lem:gauss-net}
Let $G\in\Co^{p\times q}$ be a centered Gaussian random matrix obtained as a
real-linear image of a real Gaussian vector, and suppose that for every pair
of complex unit vectors $\ket{x},\ket{y}$ both $\Re(\bra{x}G\ket{y})$ and $\Im(\bra{x}G\ket{y})$ have
variance at most $v^{2}$.  Then for $u>0$,
\begin{equation}\label{eq:gauss-net}
 \Pr\bigl\{\opnorm G>u\bigr\}
 \ \le\ 4\exp\Bigl(2(p+q)\log9-\frac{u^{2}}{16v^{2}}\Bigr).
\end{equation}
\end{lemma}

\begin{proof}
For a fixed pair of unit vectors $\ket{x}\in\Co^p$, $\ket{y}\in\Co^q$, the event
$|\bra{x}G\ket{y}|>s$ forces the real or the imaginary part to exceed $s/\sqrt2$ in
absolute value, an event of probability at most $4e^{-s^{2}/(4v^{2})}$.
The $1/4$-covering-net union bound in \cref{eq:covering-union-bilinear},
with $s=u/2$, therefore gives
\[
 \Pr\{\opnorm G>u\}
 \le9^{2(p+q)}\,4e^{-u^2/(16v^2)},
\]
which is \cref{eq:gauss-net}. 
\end{proof}
\begin{fact}[Bakry--\'Emery log-Sobolev criterion on a convex domain
{\cite[Theorem~2.1]{KolesnikovMilman16}}]\label{fact:lsi}
Let $\Omega\subseteq\R^k$ be open and convex, and let
$\Phi\in C^2(\Omega)$ satisfy
\begin{equation}
0<Z_\Phi:=\int_\Omega e^{-\Phi(\theta)}\,d\theta<\infty
\end{equation}
and
\begin{equation}
\nabla^2\Phi(\theta)\succeq\varkappa I_k
\qquad
\text{for every $\theta\in\Omega$},
\end{equation}
where $\varkappa>0$. Let $\pi$ be the probability measure on $\Omega$
with density $Z_\Phi^{-1}e^{-\Phi}$. Then, for every
$f\in C^1(\Omega)$,
\begin{equation}\label{eq:lsi-fact}
\Ent_\pi(f^2)
\le
\frac{2}{\varkappa}
\int_\Omega\|\nabla f\|_2^2\,d\pi,
\end{equation}
where
\begin{equation}
\Ent_\pi(g)
:=
\int_\Omega g\log g\,d\pi
-
\left(\int_\Omega g\,d\pi\right)
\log\left(\int_\Omega g\,d\pi\right).
\end{equation}
\end{fact}

\begin{lemma}[Mass and log-Sobolev inequality for the prior]\label{lem:prior}
In the orthonormal coordinates fixed above, let $\gamma_\sigma$ and
$\mu_t$ be the measures defined in \cref{eq:sigma,eq:prior}, and set
\begin{equation}
\Omega_t:=\operatorname{int}_{\R^k}(\calK_t).
\end{equation}
Let $0<t\le t_0=2^{-5}$.  Use the
space $\calV$, reshaped perturbation $T_\Theta$, and domain
$\calK_t$ from \cref{eq:Esubspace,eq:Ttheta,eq:Kt}.     Then
\begin{equation}\label{eq:prior-mass}
 \gamma_\sigma\bigl\{\opnorm\Theta\le t/2,\ \opnorm{T_\Theta}\le t/2\bigr\}
 \ \ge\ \tfrac12,
 \qquad\text{so in particular}\qquad
 \gamma_\sigma(\calK_t)\ \ge\ \tfrac12 .
\end{equation}
 Then $\mu_t$ is supported on
$\calK_t$, satisfies $\mu_t\le2\gamma_\sigma$ as measures, and obeys the
logarithmic Sobolev inequality
\begin{equation}\label{eq:lsi}
\Ent_{\mu_t}(f^2)
\le
2\sigma^2
\int_{\Omega_t}\|\nabla f\|_2^2\,d\mu_t
\end{equation}
for every $f\in C^1(\Omega_t)$.
\end{lemma}

\begin{proof}
For unit vectors $\ket{x}\in\Co^{D}$, $\ket{y}\in\Co^{r}$ we have
$\Re(\bra{x}\Theta \ket{y})=\langle \ket{x}\bra{y},\Theta\rangle_\R$, whose coefficient in the
coordinates of $\calV$ is the orthogonal projection of $\ket{x}\bra{y}$ onto $\calV$, of
Hilbert--Schmidt norm at most $\hsnorm{\ket{x}\bra{y}}=1$; the same holds for the
imaginary part with $i\ket{x}\bra{y}$.  So $\Theta$ satisfies the hypothesis of
\cref{lem:gauss-net} with $v=\sigma$, $p=D$, $q=r$.  Since $r\le D$,
$p+q\le2D$, and $u=t/2$ with \cref{eq:sigma} gives
$u^{2}/(16v^{2})=t^{2}/(64\sigma^{2})=16D$, hence
\[
 \Pr\bigl\{\opnorm\Theta>t/2\bigr\}\le4\exp\bigl(4D\log9-16D\bigr)\le4e^{-7D},
\]
as $4\log9<9$.  For $T_\Theta$: by \cref{eq:Ttheta} the corresponding
coefficient has norm at most $\sqrt{d_1/r}$, so the hypothesis holds with
$v^{2}=(d_1/r)\sigma^{2}$, $p=rd_2$, $q=d_1$.  \emph{Here the second
feasibility condition enters:} $d_1\le rd_2$ gives $p+q\le2rd_2$, and
\[
 \frac{u^{2}}{16v^{2}}=\frac{t^{2}r}{64d_1\sigma^{2}}
 =\frac{r}{d_1}\cdot16D=16rd_2 ,
\]
so $\Pr\bigl\{\opnorm{T_\Theta}>t/2\bigr\}\le4\exp(4rd_2\log9-16rd_2)
\le4e^{-7rd_2}$.
Both $D\ge2$ and $rd_2\ge2$, so the two exceptional probabilities sum to at
most $8e^{-14}<\tfrac12$, proving \cref{eq:prior-mass}.  The event in
\cref{eq:prior-mass} is contained in $\calK_t$, and $\mu_t\le2\gamma_\sigma$
follows from $\gamma_\sigma(\calK_t)\ge\tfrac12$.

 The set $\calK_t$ is a full-dimensional convex body in $\R^k$. Indeed, under
the chosen orthonormal coordinates,
$\|\theta\|_2=\|\Theta\|_2$, and the Euclidean ball centered at the origin
with radius
$
t\min\left\{1,\sqrt{\frac{r}{d_1}}\right\}$
is contained in $\calK_t$, because
$
\|\Theta\|_\infty\le\|\Theta\|_2,
\|T_\Theta\|_\infty
\le\|T_\Theta\|_2
=
\sqrt{\frac{d_1}{r}}\|\Theta\|_2.
$
Consequently, $\Omega_t$ is open and convex. Moreover, the boundary of the
full-dimensional convex body $\calK_t$ has Lebesgue measure zero and hence
also $\gamma_\sigma$-measure zero. Therefore, $\mu_t$ agrees almost
everywhere with the probability measure on $\Omega_t$ having density
proportional to
\begin{equation}
\exp\left(-\frac{\|\theta\|_2^2}{2\sigma^2}\right).
\end{equation}
The corresponding potential has Hessian $\sigma^{-2}I_k$. Applying
\cref{fact:lsi} on $\Omega_t$ with
$\varkappa=\sigma^{-2}$ proves \cref{eq:lsi}.
\end{proof}

\begin{proposition}[Mutual information of an adaptive transcript]\label{prop:mi}
Let $0<t\le t_0=2^{-5}$.  Let
$\theta\sim\mu_t$, where $\mu_t$ is the conditional Gaussian prior
in \cref{eq:sigma} on the parameter domain $\calK_t$ of
\cref{eq:Kt}.  The vector $\theta$ gives orthonormal coordinates for
the real Hilbert--Schmidt inner product on $\calV$ from
\cref{eq:Esubspace}.  Let $Z$ be the transcript of any
$n$-query adaptive incoherent protocol applied to the channel
$\calE_\Theta$ defined in \cref{eq:flagged}.  Then its mutual information
with the parameter satisfies
\begin{equation*}
 I(\theta;Z)\ \le\ \frac{n\,t^{2}}{128\,r} .
\end{equation*}
\end{proposition}

\begin{proof}Fix a seed value $W=w$ outside a null set, and let $\lambda_w$ be the
parameter-independent reference probability measure on transcripts obtained
from the control kernels. Write
\begin{equation}
q_\theta(z)
:=
\frac{dP_{Z\mid\theta,W=w}}{d\lambda_w}(z).
\end{equation}
For $\lambda_w$-almost every $z$, the successive tester densities are fixed.
Hence $q_\theta(z)$ is a finite product of functions of the form
\cref{eq:qtheta}. It is therefore nonnegative and real analytic in $\theta$,
and is bounded on the compact set $\calK_t$.

For $\eta>0$, define
\begin{equation}
f_{\eta,z}(\theta)
:=
\sqrt{q_\theta(z)+\eta}.
\end{equation}
Then $f_{\eta,z}\in C^1(\Omega_t)$, so \cref{eq:lsi} gives
\begin{equation}\label{eq:regularized-lsi}
\Ent_{\mu_t}\bigl(q_\cdot(z)+\eta\bigr)
\le
\frac{\sigma^2}{2}
\int_{\Omega_t}
\frac{\|\nabla_\theta q_\theta(z)\|_2^2}
     {q_\theta(z)+\eta}
\,d\mu_t(\theta).
\end{equation}

Because $q_\theta(z)$ is nonnegative and differentiable on the open set
$\Omega_t$, every point at which $q_\theta(z)=0$ is a local minimum.
Therefore,
\begin{equation}
q_\theta(z)=0
\quad\Longrightarrow\quad
\nabla_\theta q_\theta(z)=0.
\end{equation}

Because $q_\theta(z)$ is bounded on the compact set $\calK_t$, dominated
convergence applies to the entropy term on the left-hand side of
\cref{eq:regularized-lsi}. Moreover,
\begin{equation}
\frac{\|\nabla_\theta q_\theta(z)\|_2^2}
     {q_\theta(z)+\eta}
\uparrow
\frac{\|\nabla_\theta q_\theta(z)\|_2^2}
     {q_\theta(z)}
\mathbf 1_{\{q_\theta(z)>0\}}
\end{equation}
as $\eta\downarrow0$, because
$q_\theta(z)=0$ implies $\nabla_\theta q_\theta(z)=0$. Hence monotone
convergence applies to the right-hand side.
Hence 
\begin{equation}\label{eq:entropy-fisher-one-transcript}
\Ent_{\mu_t}\bigl(q_\cdot(z)\bigr)
\le
\frac{\sigma^2}{2}
\int_{\{\theta:q_\theta(z)>0\}}
\frac{\|\nabla_\theta q_\theta(z)\|_2^2}
     {q_\theta(z)}
\,d\mu_t(\theta).
\end{equation}

Integrating with respect to $\lambda_w$ and applying Tonelli's theorem gives
\begin{equation}\label{eq:mi-fisher}
\begin{aligned}
I(\theta;Z\mid W=w)
&=
\int
\Ent_{\mu_t}\bigl(q_\cdot(z)\bigr)
\,d\lambda_w(z)\\
&\le
\frac{\sigma^2}{2}
\int_{\Omega_t}
\tr I_{Z\mid W=w}(\theta)
\,d\mu_t(\theta)\\
&\le
\frac{\sigma^2}{2}\frac{16nD}{r}.
\end{aligned}
\end{equation}
Since $W$ is independent of $\theta$ and is included in the transcript,
\begin{equation}
I(\theta;Z)=I(\theta;Z\mid W).
\end{equation}
Averaging over $W$ and using
$\sigma^2=t^2/(1024D)$ therefore gives
\begin{equation}
I(\theta;Z)
\le
\frac{\sigma^2}{2}\frac{16nD}{r}
=
\frac{nt^2}{128r}.
\end{equation}
\end{proof}

\begin{lemma}[Prior mass of trace-norm balls]\label{lem:smallball}
Assume \cref{eq:feasible}.  Let $0<t\le t_0=2^{-5}$ and
$0<\epsilon\le t/M$ with $M:=2^{21}$.  Use the Choi operator
$C_\Theta=C_{\calE_\Theta}$, parameter domain $\calK_t$, and conditional
Gaussian prior $\mu_t$ defined in \cref{eq:flagged,eq:Kt,eq:sigma}.
Then for every
$\varsigma\in\Co^{D\times D}$,
\begin{equation*}
 \mu_t\bigl\{\Theta\in\calK_t:\ \frac1{d_1}\trnorm{C_\Theta-\varsigma}\le\epsilon\bigr\}
 \ \le\ 2\,e^{-Dr/2} .
\end{equation*}
\end{lemma}

\begin{proof}
Write $\calS$ for the set in question and assume it is nonempty, else there is
nothing to prove; fix $\Xi\in\calS$.  For $\Theta\in\calS$ the triangle inequality
gives $\frac1{d_1}\trnorm{C_\Theta-C_\Xi}\le2\epsilon$.  Combining
\cref{lem:spectral-control}\,(iii), the singular-value inequality
$\hsnorm H^{2}\le\opnorm H\trnorm H$, and
\cref{lem:spectral-control}\,(iv),
\begin{equation*}
 \frac{\hsnorm{\Theta-\Xi}^{2}}{4r^{2}}
 \ \le\ \frac1{d_1^2}\hsnorm{C_\Theta-C_\Xi}^{2}
 \ \le\ \frac1{d_1}\opnorm{C_\Theta-C_\Xi}\,\frac1{d_1}\trnorm{C_\Theta-C_\Xi}
 \ \le\ \frac{16t}{r}\cdot2\epsilon ,
\end{equation*}
so $\hsnorm{\Theta-\Xi}^{2}\le128\,rt\epsilon$.  Thus $\calS$ is contained in the
Euclidean ball $\calB(\xi,\rho)\subset\R^{k}$ of radius
$\rho=\sqrt{128\,rt\epsilon}$.

The Gaussian density is everywhere at most $(2\pi\sigma^2)^{-k/2}$.
Consequently, for any center $\xi$,
\[
 \gamma_\sigma(\calB(\xi,\rho))
 \le\frac{\rho^k}{2^{k/2}\sigma^k\Gamma(k/2+1)}
 \le\left(\frac{\rho\sqrt e}{\sigma\sqrt k}\right)^k.
\]
Here we used the volume of a Euclidean ball \cite[Sec.~2.4]{EG15} and
$\Gamma(k/2+1)\ge(k/(2e))^{k/2}$.
By \cref{lem:dim} and \cref{eq:sigma},
$\sigma\sqrt k\ge\tfrac{t}{32\sqrt D}\sqrt{\tfrac34Dr}=\tfrac{t\sqrt{3r}}{64}
\ge t\sqrt r/38$, hence
\[
 \frac{\rho\sqrt e}{\sigma\sqrt k}
 \ \le\ \frac{38\sqrt e\,\sqrt{128\,rt\epsilon}}{t\sqrt r}
 \ =\ 38\sqrt{128e}\,\sqrt{\epsilon/t}
 \ \le\ 709\,M^{-1/2}\ \le\ \tfrac12 ,
\]
since $M=2^{21}$ and $709\cdot2^{-21/2}<\tfrac12$.  Using
$\mu_t\le2\gamma_\sigma$ from \cref{lem:prior} and $k\ge\tfrac34Dr$ from
\cref{lem:dim},
\[
 \mu_t(\calS)\ \le\ 2\cdot2^{-k}\ \le\ 2e^{-\frac34(\log2)Dr}
 \ \le\ 2e^{-Dr/2}. \qedhere
\]
\end{proof}

\begin{lemma}[Fano's inequality for reconstruction sets]\label{lem:fano}
Let $\theta\sim\mu$ be a random parameter, let $Y$ be data with conditional
law $P_{Y\mid\theta}$, and let $\widehat\theta=\widehat\theta(Y)$ be an
estimator.  Fix $\epsilon>0$ and $0\le\delta<1$.  For each possible estimate
$\vartheta$, let $\calB(\vartheta,\epsilon)$ be the measurable set of
parameter values reconstructed to accuracy $\epsilon$ by that estimate.
Assume that the success probability $\Pr\{\theta\in\calB(\widehat\theta(Y),\epsilon)\}\ge1-\delta$.
If $\mu\bigl(\calB(\vartheta,\epsilon)\bigr)\le\alpha_0$ for every
$\vartheta$, where $0<\alpha_0\le1$, then
\begin{equation}\label{eq:fano}
 I(\theta;Y)\ \ge\ (1-\delta)\log\frac1{\alpha_0}-\log2 .
\end{equation}
\end{lemma}

\begin{proof}
Let $\calT=\mathbf{1}_{\{\theta\in\calB(\widehat\theta(Y),\epsilon)\}}$, a binary function of
$(\theta,Y)$.  By the data-processing inequality for relative entropy applied
to the joint law $P_{\theta Y}$ and the product $\mu\otimes P_Y$,
\[
 I(\theta;Y)=D_{\mathrm{KL}}\bigl(P_{\theta Y}\,\|\,\mu\otimes P_Y\bigr)
 \ \ge\ h\bigl(P_{\theta Y}\{\calT=1\},\ (\mu\otimes P_Y)\{\calT=1\}\bigr),
\]
where $h(p,q)=p\log\frac pq+(1-p)\log\frac{1-p}{1-q}$.  By hypothesis
$P_{\theta Y}\{\calT=1\}\ge1-\delta$, while
$(\mu\otimes P_Y)\{\calT=1\}=\E_Y\mu(\calB(\widehat\theta(Y),\epsilon))\le\alpha_0$.
For $p=P_{\theta Y}\{\calT=1\}$ and
$q=(\mu\otimes P_Y)\{\calT=1\}$, use
$h(p,q)\ge p\log(1/q)-\log2$ to obtain \cref{eq:fano}.
\end{proof}

\section*{Statement on the use of AI}
The upper-bound proof for channels with gapped Choi spectrum was developed
by the human authors. After being provided with this proof as context,
GPT-6 Astra generated the proof strategy for the general upper bound.
GPT-5.6 Sol generated the lower-bound proof after being provided with the
earlier construction from \cite{CGOYZ26} and the suggestion to pursue a
van Trees-type argument. The new hard family of channels used in the
lower bound was generated by GPT-6 Astra. The authors subsequently checked,
revised, and completed all AI-generated arguments and take full
responsibility for the correctness and presentation of the results.

\bibliographystyle{alpha}
\bibliography{bib}

\end{document}